%% file: gadget_extraction_paper_qpl_proc2.tex
\documentclass[submission]{eptcs}

\usepackage{breakurl}
\usepackage{microtype}
\usepackage{cite}
\usepackage{amsmath}
\usepackage{amsfonts}
\usepackage{braket}
\usepackage{amsthm}
\usepackage{algorithm2e}
\usepackage{stmaryrd}
\usepackage{pifont}
\usepackage{scalerel}
\usepackage{hyperref}
\usepackage{xcolor}
\usepackage{tikzit}
\usepackage{tikzcircuits}
\usepackage{tikzzx}
\usepackage{tikz-cd}
\input{mbqc_style.tikzstyles}
\input{circuits.tikzstyles}
\allowdisplaybreaks

\makeatletter
\def\namedlabel#1#2{\begingroup
#2%
\def\@currentlabel{#2}%
\phantomsection\label{#1}\endgroup
}
\makeatother

\newcommand{\cmark}{\ding{51}}
\newcommand{\xmark}{\ding{55}}
\DeclareMathOperator*{\bigdelta}{\scalerel*{\Delta}{\sum}}

\theoremstyle{definition}
\newtheorem{theorem}{Theorem}[section]
\newtheorem{lemma}[theorem]{Lemma}

\newtheorem{definition}[theorem]{Definition}
\newtheorem{remark}[theorem]{Remark}

\newtheorem{example}[theorem]{Example}

\newcommand{\CZ}{$CZ$}
\newcommand{\CX}{$CX$}
\newcommand{\RZ}{$RZ$}
\newcommand{\RX}{$RX$}

\newcommand{\T}{$T$}

\title{Relating Measurement Patterns to Circuits via Pauli Flow}
\author{Will Simmons
\institute{Cambridge Quantum Computing Ltd \\ 9a Bridge Street, Cambridge, UK}
\institute{Department of Computer Science, University of Oxford \\
Wolfson Building, Parks Road, Oxford, UK}
\email{will.simmons@cambridgequantum.com}}
\def\titlerunning{Relating Measurement Patterns to Circuits via Pauli Flow}
\def\authorrunning{W. Simmons}

\begin{document}

\maketitle

\begin{abstract}
The one-way model of Measurement-Based Quantum Computing and the gate-based circuit model give two different presentations of how quantum computation can be performed. There are known methods for converting any gate-based quantum circuit into a one-way computation, whereas the reverse is only efficient given some constraints on the structure of the measurement pattern. Causal flow and generalised flow have already been shown as sufficient, with efficient algorithms for identifying these properties and performing the circuit extraction. Pauli flow is a weaker set of conditions that extends generalised flow to use the knowledge that some vertices are measured in a Pauli basis. In this paper, we show that Pauli flow can similarly be identified efficiently and that any measurement pattern whose underlying graph admits a Pauli flow can be efficiently transformed into a gate-based circuit without using ancilla qubits. We then use this relationship to derive simulation results for the effects of graph-theoretic rewrites in the ZX-calculus using a more circuit-like data structure we call the Pauli Dependency DAG.
\end{abstract}

\section{Introduction}\label{sec:Introduction}


There are numerous approaches to quantum computation that differ in how problems are encoded and how the system evolves over time. The one-way model \cite{Raussendorf2001} is a variant of Measurement-Based Quantum Computing (MBQC) where the inputs are embedded into a general graph state, and the chosen computation is specified by a sequence of single-qubit measurements to apply, consuming the state and inducing a desired output state on the remaining qubits. In contrast, the gate-based model\cite{Nielsen2010} applies unitary gates from some universal gateset to the input qubits directly. Gate-based circuits may use ancilla qubits and measurements, though they are only necessary for expanding to a larger state space for outputs and reading out final data.


There is a wide literature on rewriting both measurement patterns \cite{Duncan2010, DaSilva2013, Eslamy2018} and circuits \cite{Maslov2017, Fagan, Kissinger2019, Zhang2019a, Amy2019, Cowtan2020} to reduce the resource cost. Of particular interest here is the use of graph-theoretic rewrites of ZX-diagrams (in close correspondence with measurement patterns \cite{Duncan2010}) by Kissinger and van de Wetering \cite{Kissinger2019}, and of representations focussed on sequences of rotations about Pauli tensors such as by Zhang and Chen \cite{Zhang2019a} herein referred to as Pauli Dependency DAGs (PDDAGs) as they capture the temporal dependencies between these rotations. Both works presented these as techniques for \T~gate reduction with identical performance. The extraction method in this paper maps a measurement pattern/ZX-diagram directly into the form of a PDDAG, allowing us to formally demonstrate that each of the graph-theoretic rewrites can be simulated by moving Clifford-angled rotations through the PDDAG. This complements existing work on deriving equivalences between the PathSum formalism and the ZH calculus \cite{Backens2018, Amy2019a, Lemonnier2020} and between different graphical calculi \cite{FengNgQuanlongWang2017, Kuijpers2019} which have shown to provide useful insight on the semantics of operations for the more abstract calculi.


Previous work in this area has given algorithms for extracting circuits from measurement patterns that satisfy flow properties which describe how the errors from unwanted measurement outcomes can be corrected using the stabilizers of the graph state \cite{Beaudrap2010, Backens2020}. The weakest such property that allows for interesting (i.e. non-Pauli) measurements is generalised flow (gflow) for which the single-plane form exactly describes the set of patterns that are uniformly, strongly, and stepwise deterministic \cite{Browne}. Pauli flow \cite{Browne} weakens the uniformity condition by labelling some of the measurements as being in a Pauli basis rather than at some arbitrary angle in a plane of the Bloch sphere. This allows for more interesting corrections that make use of Pauli projections to absorb some terms of the graph state stabilizers. This paper directly extends the works of Mhalla and Perdrix \cite{Mhalla2008} and Backens et al. \cite{Backens2020} for gflow to handle this wider class of measurement patterns.


This paper starts by laying out the key definitions and concepts behind measurement patterns and existing results on causal flow and gflow in Section \ref{sec:MBQC}. Section \ref{sec:PDDAG} will give a short overview of the Pauli Dependency DAG structure and how it can be used to rewrite circuits. We will then introduce the notion of Pauli flow and decompose the circuit extraction problem to generate a procedure that works for any measurement pattern with Pauli flow in Section \ref{sec:Extraction}. Then Section \ref{sec:Rewrites} will cover a number of methods of rewriting measurement patterns in turn and investigate their effects on the PDDAG extracted to show how they can be simulated.

\section{Measurement-Based Quantum Computing}\label{sec:MBQC}


The one-way model follows the typical structure of MBQC of building some resource state which is then consumed by single-qubit measurements. The particular resource considered is a graph state, constructed by matching qubits (inputs and ancillas prepared in the $\ket{+} = \tfrac{1}{\sqrt{2}} \left( \ket{0} + \ket{1} \right)$ state) with vertices in a graph, and entangling them (with a \CZ~gate) according to the edges. Measurements are generally non-deterministic so, in order to have a deterministic effect overall, local gates can be applied to the remaining qubits to counteract the difference between the projectors for each outcome, giving the net effect of post-selecting the desired outcome. The introduction of such correction gates can be viewed as adapting the choice of basis for the future measurements.


Single qubit measurements are characterised by the bases they project into, which themselves can be described by a vector in the Bloch sphere. For MBQC, we conventionally restrict measurements into a plane of the Bloch sphere spanned by two of the Pauli bases. Such planar measurement bases are described by the following for $\alpha \in \left[ 0, 2\pi \right)$:
\begin{align}
\ket{+_{XY, \alpha}} &= \tfrac{1}{\sqrt{2}} \left(\ket{0} + e^{i\alpha}\ket{1}\right) &
\ket{-_{XY, \alpha}} &= \tfrac{1}{\sqrt{2}} \left(\ket{0} - e^{i\alpha}\ket{1}\right) \nonumber \\
\ket{+_{XZ, \alpha}} &= \cos \left( \tfrac{\alpha}{2} \right) \ket{0} + \sin \left( \tfrac{\alpha}{2} \right) \ket{1} &
\ket{-_{XZ, \alpha}} &= \sin \left( \tfrac{\alpha}{2} \right) \ket{0} - \cos \left( \tfrac{\alpha}{2} \right) \ket{1} \\
\ket{+_{YZ, \alpha}} &= \cos \left( \tfrac{\alpha}{2} \right) \ket{0} + i \sin \left( \tfrac{\alpha}{2} \right) \ket{1} &
\ket{-_{YZ, \alpha}} &= \sin \left( \tfrac{\alpha}{2} \right) \ket{0} - i \cos \left( \tfrac{\alpha}{2} \right) \ket{1} \nonumber
\end{align}
with the following special cases for Pauli measurements for $\alpha = a \pi$ ($a \in \{0, 1\}$):
\begin{align}
\ket{+_{X, a\pi}} &= \tfrac{1}{\sqrt{2}} \left(\ket{0} + (-1)^a \ket{1}\right) &
\ket{-_{X, a\pi}} &= \tfrac{1}{\sqrt{2}} \left(\ket{0} - (-1)^a \ket{1}\right) \nonumber \\
\ket{+_{Y, a\pi}} &= \tfrac{1}{\sqrt{2}} \left(\ket{0} + i (-1)^a \ket{1}\right) &
\ket{-_{Y, a\pi}} &= \tfrac{1}{\sqrt{2}} \left(\ket{0} - i (-1)^a \ket{1}\right) \\
\ket{+_{Z, a\pi}} &= \left(1 - a\right) \ket{0} + a \ket{1} &
\ket{-_{Z, a\pi}} &= a \ket{0} + \left(1 - a\right) \ket{1} \nonumber
\end{align}

For any measurement basis, the negative outcome at angle $\alpha$ is equivalent to the positive outcome at angle $\alpha + \pi$. We usually treat the positive measurement outcome as the desired branch, i.e. the projector we want to apply to achieve our desired end state.


Formally, measurement patterns (i.e. MBQC programs) are defined as follows:

\begin{definition}[Measurement pattern]
A \textit{measurement pattern} consists of a collection $V$ of qubits with distinguished subsets $I, O \subseteq V$ of inputs and outputs, and a sequence of commands from:
\begin{itemize}
\item Preparations $N_u$, initialising qubit $u \notin I$ to $\ket{+}$;
\item Entangling operators $E_{uv}$, applying a \CZ~gate between distinct qubits $u, v \in V$;
\item Destructive measurements $M_u^{\lambda, \alpha}$, projecting qubit $u \notin O$ onto either $\ket{+_{\lambda, \alpha}}$ with outcome $0$ or $\ket{-_{\lambda, \alpha}}$ with outcome $1$;
\item Corrections $[X_u]^v$ or $[Z_u]^v$, conditionally applying an $X$ gate or a $Z$ gate to qubit $u \in V$ if the outcome of the measurement for qubit $v$ was $1$.
\end{itemize}
A measurement pattern is \textit{runnable} if additionally:
\begin{itemize}
\item All non-input qubits are prepared exactly once;
\item A non-input qubit is not acted on by any other command before its preparation;
\item All non-output qubits are measured exactly once;
\item A non-output qubit is not acted on by any other command after its measurement;
\item No correction depends on an outcome not yet measured.
\end{itemize}
\end{definition}

The intended branch of a measurement pattern (where all measurement outcomes are $0$ and hence no corrections are required) can be summarised by a tuple $(\Gamma, \alpha)$ of a labelled open graph $\Gamma$ describing the entanglement and measurement planes, and an assignment of measurement angles $\alpha : \overline{O} \to \left[ 0, 2\pi \right)$.

\begin{definition}[Labelled open graph]
A \textit{labelled open graph} is a tuple $\Gamma = \left(G, I, O, \lambda\right)$ where $G = (V, E)$ is an undirected graph, $I, O \subseteq V$ are (possibly overlapping) subsets of vertices for inputs and outputs respectively, and $\lambda : \overline{O} \to \{XY, XZ, YZ, X, Y, Z\}$ is a labelling function assigning a measurement plane or Pauli to each non-output vertex.
\end{definition}

\begin{remark}
To fix notation, we will use $\overline{I} = V \setminus I$ for non-input (prepared) vertices and $\overline{O} = V \setminus O$ for non-output (measured) vertices. $u \sim v$ denotes vertices $u, v \in V$ being adjacent in $G$. Neighbour sets are $N_G(u) := \{w \in V | w \sim u\}$ and odd neighbourhoods are $\mathrm{Odd}_G(A) := \left\{w \in V \mid \left|N_G(w) \cap A\right| \text{ is odd} \right\}$, writing $\mathrm{Odd}(A)$ when the choice of graph is obvious. The symmetric difference of sets is written as $A \Delta B := (A \setminus B) \cup (B \setminus A)$. When drawing diagrams for measurement patterns, we will distinguish between measured and output vertices as filled and unfilled points respectively, and inputs are specified by boxes around the vertices. For linear maps, subscripts may be used to specify that a map acts on some given qubit(s) and is the identity on all others, such as $\bra{+_{\lambda(v), \alpha(v)}}_v$ or $P_v$ for some Pauli $P$.
\end{remark}


Each branch (combination of measurement outcomes) gives rise to a linear map from the inputs to the outputs based on the commands run and the projections observed. The overall channel from considering all branches is a completely-positive trace-preserving map whose Kraus maps are exactly the branch maps. When we have a strongly deterministic pattern (all branches are equal up to global phase), we just associate it with the linear map of the intended branch. Since runnable patterns can be standardised to perform all initialisations first, then all entangling gates, and finally alternate measurements and corrections, this linear map also has a standard representation.

\begin{definition}[Linear map of a pattern]
The linear map associated with a measurement pattern $(\Gamma, \alpha)$ is given by
\begin{equation}\label{eq:PatternSemantics}
M_{\Gamma, \alpha} := \left( \prod_{u \in \overline{O}} \bra{+_{\lambda(u), \alpha(u)}}_u \right) E_G N_{\overline{I}}
\end{equation}
where $E_G := \prod_{u \sim v} E_{uv}$ entangles pairs of adjacent qubits in the graph with \CZ~gates and $N_{\overline{I}} := \prod_{u \in \overline{I}} N_u$ initialises every non-input in the $\ket{+}$ state.
\end{definition}


Interpreting the graph state as $E_G N_{\overline{I}}$ gives rise to a stabilizer per initialised vertex $u \in \overline{I}$:
\begin{equation}\label{eq:GraphStateStabilizer}
E_G N_{\overline{I}} = E_G X_u N_{\overline{I}} = \left( \prod_{v \in N_G(u)} Z_v \right) X_u E_G N_{\overline{I}}
\end{equation}

Since the linear map of a pattern is not just a simple graph state but a projected graph state, it is possible for some Pauli terms introduced by graph state stabilizers to be absorbed by the Pauli projections:
\begin{equation}\label{eq:PostEigen}
\begin{split}
\bra{+_{X, a\pi}}_u &= (-1)^a \bra{+_{X, a\pi}}_u X_u \\
\bra{+_{Y, a\pi}}_u &= (-1)^a \bra{+_{Y, a\pi}}_u Y_u \\
\bra{+_{Z, a\pi}}_u &= (-1)^a \bra{+_{Z, a\pi}}_u Z_u \\
\end{split}
\end{equation}

These stabilizers and Pauli absorptions are practical for deriving possible means of correcting for unwanted measurement outcomes.


The intention behind different types of flow conditions is to capture the ability to propagate errors from unwanted measurement outcomes forward to the rest of the circuit in order to correct them, aiming for stepwise determinism (each measurement can be corrected independently). Causal flow is the simplest case where we suppose all vertices are measured in the $XY$ basis and each error can be corrected by considering a single stabilizer of the graph state.

\begin{definition}[Causal flow \cite{Danos2006}]
Given a labelled open graph $\Gamma = (G, I, O, \lambda)$ such that $\forall u \in \overline{O} . \lambda(u) = XY$, a \textit{causal flow} for $\Gamma$ is a tuple $(f, \prec)$ of a map $f : \overline{O} \to \overline{I}$ and a strict partial order $\prec$ over $V$ such that for all $v \in \overline{O}$:
\begin{itemize}
\item $v \sim f(v)$
\item $v \prec f(v)$
\item $\forall w \in N_G(f(v)) . w = v \vee v \prec w$
\end{itemize}
\end{definition}

The idea here is that $Z_v$ from the graph stabilizer $\left(\prod_{w \in N_G(f(v))} Z_w\right) X_{f(v)}$ will eliminate the effect of the measurement error on qubit $v$, meaning we can correct the error by applying $\left(\prod_{w \in N_G(f(v)) \setminus \{v\}} Z_w \right) X_{f(v)}$ and implicitly invoking the stabilizer. The partial order $\prec$ indicates a required order of the measurements, ensuring that none of the vertices required for correcting $v$ have been measured yet.

Generalised flow takes a similar approach, but allows us to take combinations of the basic stabilizers. If the stabilizer of a candidate $f(v)$ would require a $Z$ correction on some $u \in N_G(f(v))$ that has already been measured, there may exist some other stabilizer we can apply that cancels out the $Z$ for us. Now, instead of the stabilizer being determined by a single vertex $f(v) \in \overline{I}$, we have a set $g(v) \subseteq \overline{I}$ giving $\left( \prod_{w \in \mathrm{Odd}(g(v))} Z_w \right) \left( \prod_{w \in g(v)} X_w \right)$ up to phase. By relaxing these restrictions on the stabilizers used, we can also generate $Y$ or $X$ on the measured vertex, allowing the correction of measurements in the $XZ$ and $YZ$ planes respectively.

\begin{definition}[Generalised flow \cite{Browne}]
Given a labelled open graph $\Gamma = (G, I, O, \lambda)$ such that $\forall u \in \overline{O} . \lambda(u) \in \{XY, XZ, YZ\}$, a \textit{generalised flow} (or \textit{gflow}) for $\Gamma$ is a tuple $(g, \prec)$ of a map $g : \overline{O} \to \mathcal{P}[\overline{I}]$ and a strict partial order $\prec$ over $V$ such that for all $v \in \overline{O}$:
\begin{itemize}
\item $\forall w \in g(v) . v \neq w \Rightarrow v \prec w$
\item $\forall w \in \mathrm{Odd} (g(v)) . v \neq w \Rightarrow v \prec w$
\item $\lambda(v) = XY \Rightarrow v \notin g(v) \wedge v \in \mathrm{Odd}(g(v))$
\item $\lambda(v) = XZ \Rightarrow v \in g(v) \wedge v \in \mathrm{Odd}(g(v))$
\item $\lambda(v) = YZ \Rightarrow v \in g(v) \wedge v \notin \mathrm{Odd}(g(v))$
\end{itemize}
\end{definition}

There exist polynomial-time algorithms for identifying whether or not a labelled open graph admits causal flow or gflow \cite{DeBeaudrap2008, Mhalla2008, Backens2020}, and for extracting an equivalent unitary circuit for the measurement pattern given either type of flow \cite{Beaudrap2010, Backens2020}.

\section{Pauli Dependency DAGs}\label{sec:PDDAG}


The Pauli Dependency DAG is a data structure for representing the action of a pure quantum circuit that abstracts away gate set, Clifford gates, and gate commutations. This was notably covered in the work of Zhang and Chen \cite{Zhang2019a} where it was used to identify possible pairs of \T~gates which could be merged through some sequence of gate commutations and Clifford gate relations in order to reduce the number of \T~gates in the circuit. Similar structures are also covered by Litinski \cite{Litinski2019} for compiling circuits for lattice surgery and by Gosset et al. \cite{Gosset2013} as an intermediate for synthesis of Clifford+\T~circuits.

The Clifford group is the group of linear maps that can be formed from \CX, Hadamard, and $RZ(\tfrac{\pi}{2})$ with qubit initialisation in the $\ket{0}$ state. These hold a special place in quantum theory as the maximal group of actions under which the Pauli group is closed. Their behaviour is conveniently captured by the stabilizer framework, allowing for canonical representations like the stabilizer tableau \cite{VandenNest2004, Aaronson2008}. They are often viewed as the ``easy'' gates in a circuit because of their simple algebra, efficient simulation \cite{Gottesman1998, Aaronson2008}, and low resource cost in many error correction schemes \cite{Calderbank1997}.


Most common gate types can be expressed as (combinations of) Pauli exponentials ($e^{i \theta P}$ for some $P \in \{I, X, Y, Z\}^{\otimes n}$), which also benefit from elegant relations with stabilizers.

\begin{lemma}[\namedlabel{lemma:GadgetStabCorrespondence}{Product Rotation Lemma}]
Let $A$ and $B$ be commuting operators such that $B \mathcal{C} = \mathcal{C}$ for some linear map $\mathcal{C}$. Then $e^{i \theta A} \mathcal{C} = e^{i \theta AB} \mathcal{C}$.
\end{lemma}

\begin{proof}
For any analytic function $F(A)$, we can expand its Taylor series in $F(A) \mathcal{C}$, then introduce and commute $B$ in each term to form the Taylor series for $F(AB)\mathcal{C}$.
\end{proof}


\begin{table}
\centering
\begin{tabular}[t]{c|l}
Gate & Equivalent Pauli exponentials \\
\hline
$CX_{ct}$ & $e^{-i\tfrac{\pi}{4} Z_c X_t} e^{i\tfrac{\pi}{4} Z_c} e^{i\tfrac{\pi}{4} X_t}$ \\
$CZ_{ct}$ & $e^{-i\tfrac{\pi}{4} Z_c Z_t} e^{i\tfrac{\pi}{4} Z_c} e^{i\tfrac{\pi}{4} Z_t}$ \\
$RZ(\alpha)$ & $e^{-i\tfrac{\alpha}{2} Z}$ \\
$RX(\alpha)$ & $e^{-i\tfrac{\alpha}{2} X}$ \\
$H$ & $e^{-i\tfrac{\pi}{4} Z} e^{-i\tfrac{\pi}{4} X} e^{-i\tfrac{\pi}{4} Z}$ \\
$CCX_{abt}$ & $e^{-i\tfrac{\pi}{8} Z_a Z_b X_t} e^{i\tfrac{\pi}{8} Z_a Z_b} e^{i\tfrac{\pi}{8} Z_a X_t} e^{i\tfrac{\pi}{8} Z_b X_t} e^{-i\tfrac{\pi}{8} Z_a} e^{-i\tfrac{\pi}{8} Z_b} e^{-i\tfrac{\pi}{8} X_t}$ \\
\end{tabular}
\caption{Summary of conversions from common gates into Pauli exponentials (up to global phase). Subscripts are used to indicate specific qubits for multi-qubit gates. Other gates can similarly be represented by first decomposing into a universal gate set such as $\{CX, RZ, RX\}$ and optionally tidying up using the \ref{lemma:CommutationRules}.}\label{tab:GatesAsExponentials}
\end{table}

Pauli exponentials where the coefficient is an integral multiple of $\tfrac{\pi}{4}$ correspond to Clifford gates. The action of Clifford gates on both Paulis and arbitrary Pauli exponentials can be summarised in a few equations.

\begin{lemma}[\namedlabel{lemma:CommutationRules}{Reorder Rules}]
For any Pauli strings $A, B \in \{I, X, Y, Z\}^{\otimes n}$ and angles $\theta, \phi$, if $A$ and $B$ commute, then
\begin{align}
e^{i \theta A} B &= B e^{i \theta A} \\
e^{i \theta A} e^{i \phi B} &= e^{i \phi B} e^{i \theta A}
\end{align}
and otherwise (i.e. they anticommute)
\begin{align}
e^{i \tfrac{\pi}{4} A} B &= (iAB) e^{i\tfrac{\pi}{4} A} \\
e^{i \tfrac{\pi}{4} A} e^{i \phi B} &= e^{i \phi (iAB)} e^{i\tfrac{\pi}{4} A} \label{eq:AntiCommutingExpRule}
\end{align}
\end{lemma}

\begin{proof}
Any operator satisfying $A^2 = I$ and real $\theta$ permit the decomposition $e^{i \theta A} = \cos \theta I + i \sin \theta A$ by grouping terms in the Taylor expansion. Each equation follows from decomposing one of the exponentials in this way. Whilst the right-hand side of Equation \ref{eq:AntiCommutingExpRule} appears to contain a real exponent, remember that $iAB$ is a real Pauli string when $A$ and $B$ anticommute.
\end{proof}


We can represent any circuit as a set of qubit initialisations followed by a sequence of Pauli exponentials by decomposing each gate in turn. The above rules can then be used to move any Clifford-angled exponentials to the start of the circuit, resulting in the product form $\left( \prod_{k} e^{i \phi_k A_k} \right) \mathcal{C}$ where $\mathcal{C}$ is a stabilizer process and each $\phi_k$ is not an integral multiple of $\tfrac{\pi}{4}$. Previous presentations \cite{Litinski2019,Zhang2019a} collected the Clifford gates at the end of the circuit rather than at the start. When considering unitary circuits, the choice of direction is arbitrary, but in the general case we may not always be able to transport qubit initialisations through the Pauli exponentials to the end.


$\mathcal{C}$ can be expressed canonically by its stabilizer tableau. Specifically, since we could have an arbitrary number of inputs we need a variant mixing the extremes of the unitary case \cite{Aaronson2008} and the usual state case \cite{VandenNest2004} to capture both how Paulis over the inputs are transformed to Pauli strings over the outputs and additional generators for free stabilizers over the outputs. In examples, we will describe this by the stabilizers of the Choi operator of $\mathcal{C}$ (which we shall refer to as an \textit{isometry tableau} for clarity), from which it is simple to convert it to any binary matrix format of choice. Such tableaux can be identified and synthesised by embedding them into unitary tableaux by breaking down the isometry into an initialisation of $|O|-|I|$ qubits in the $\ket{0}$ state followed by a unitary circuit which maps $Z$ on the initialised qubits to each of the free output stabilizers and $X$ to any choice of operators that extend our generators to span the full Pauli group. Since the rows are just generators for the stabilizer group, we still have the notion of ``free actions'' given by reordering rows or multiplying two rows together.

For the rotation list, there is still some obvious redundancy in the product form resulting from commutations. Because anticommuting Pauli strings prevent commutation of their exponentials, any valid ordering of the rotations will preserve the relative order of rotations with anticommuting strings, inducing a temporal dependency between them. Taking the transitive closure of these dependencies gives a partial order between the exponentials which represents $\prod_{k} e^{i \phi_k A_k}$ up to any number of commutations.

\begin{definition}[Pauli Dependency DAG]
A \textit{Pauli Dependency DAG} (PDDAG) for a circuit is a pair of an isometry tableau for a Clifford map $\mathcal{C}$ and a directed acyclic graph for a partial order $\prec$ over terms $\{(A_k, \alpha_k)\}_k$ ($A_k \in \{I, X, Y, Z\}^n$, $\alpha_k \in \left[ 0, 2\pi \right)$) such that the linear map of the circuit is given by $\left( \prod_k^\succ e^{i \tfrac{\alpha_k}{2} A_k} \right) \mathcal{C}$\footnote{The convention of using $e^{i \tfrac{\alpha_k}{2} A_k}$ rather than $e^{i \alpha_k A_k}$ is to give a closer semblance to the definitions of usual rotation gates like \RZ~and their multi-qubit counterparts from literature on Pauli gadgets \cite{Cowtan2020, Cowtan2020a} and to make the notation easier when we come to discuss extraction but is just a notational choice.} up to global phase.
\end{definition}

Since the full $\prec$ relation is rather large, we use the minimal representation/Hasse diagram which includes edges in the DAG for immediate successors (those pairs of exponentials whose strings anticommute and could appear adjacent to one another in some valid topological ordering).

\begin{figure}
\centering
\input{tikz_figs/circ_gadget_long.tikz}

\vspace{0.5cm}

\addtolength{\tabcolsep}{-3pt}
\begin{tabular}{cc|cc|c}
\multicolumn{2}{c|}{Ins} & \multicolumn{2}{c|}{Outs} & Sign \\
\hline
$X$ & & $Y$ & & $+$ \\
& $X$ & & $X$ & $+$ \\
\hline
$Z$ & & $Z$ & & $+$ \\
& $Z$ & & $Z$ & $+$ \\
\end{tabular}
\addtolength{\tabcolsep}{3pt}
\begin{tikzcd}[ampersand replacement=\&,row sep=0.cm,column sep=0.5cm]
(Z_1 I_2, -\alpha_0) \arrow{r} \& (X_1 I_2, \alpha_2) \arrow{r} \arrow{rd} \& (Y_1 X_2, -\alpha_3) \arrow{r} \arrow{rdd} \& (Z_1 I_2, -\alpha_6) \\
\& \& (Y_1 X_2, -\alpha_5) \arrow{ru} \arrow{rd} \& \\
\& (I_1 Z_2, -\alpha_1) \arrow{r} \arrow{ru} \arrow{ruu} \& (I_1 X_2, \alpha_4) \arrow{r} \& (I_1 Y_2, -\alpha_7)
\end{tikzcd}
\caption{An example quantum circuit and the corresponding isometry tableau and rotation DAG which form the PDDAG. Since the nodes of angles $\alpha_3$ and $\alpha_5$ have no relative dependencies and share the same Pauli string, it is possible to commute them to be adjacent and merge them to a single rotation. The $\alpha_4$ rotation can commute through the \CX gates so is only blocked by $\alpha_1$ in the past and $\alpha_7$ in the future.}
\end{figure}
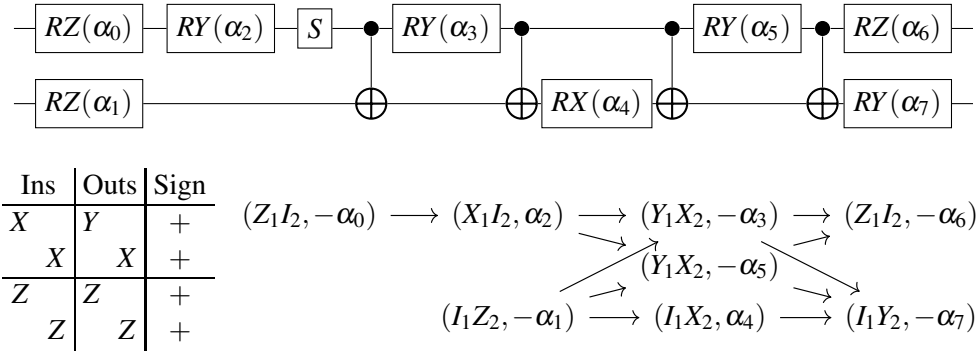


Rewrite rules for PDDAGs change the structure but preserve the linear map. The most useful rewrite for circuit reduction is rotation merging: two nodes in the graph $j, k$ can be merged to give a single node with $(A_j, \alpha_j + \alpha_k)$ if $A_j = A_k$ and both $j \not\prec k$ and $k \not\prec j$ (i.e. there is some valid topological ordering of the graph which puts these rotations together).

The PDDAG definition permits Clifford-angled rotations, allowing the movement of Clifford rotations through the circuit via the \ref{lemma:CommutationRules} to be an explicit rewrite. Combining this with merging/splitting of phases and eliminating rotations with angle $0$ permits a more rigid canonical form where the range of angles for each rotation is reduced from $\left[0, 2\pi\right)$ to $\left(0, \tfrac{\pi}{2}\right)$.

When the initial Clifford process includes some qubit initialisations, it will have some free stabilizers over its outputs that can be used to change the Pauli strings of rotations in the DAG using the \ref{lemma:GadgetStabCorrespondence}. We can use the \ref{lemma:CommutationRules} to propagate the stabilizers beyond the initial rotations to find more places to perform this rewrite. Changing the Pauli strings of rotations in this way can affect how they commute/anticommute with their neighbours, modifying the dependency relation $\prec$.


Extracting a circuit back out of an arbitrary PDDAG can be naively done by synthesising the stabilizer tableau \cite{Aaronson2008, Maslov} and each of the Pauli exponentials in turn according to some topological ordering using some standard decompositions \cite{Barkoutsos2018, Cowtan2020} although this will typically add an extremely high amount of redundant Clifford gates. More efficient synthesis can be performed using techniques for synthesising pairs of rotations simultaneously \cite{Cowtan2020} or by diagonalising sets of mutually commuting rotations \cite{VandenBerg2020, Cowtan2020a}. It is also possible that future architectures may find efficient ways to perform each Pauli rotation natively or employ lattice surgery where it is practical to just perform them directly \cite{Litinski2019}.

Alternatively, similar to the action of phase teleportation in ZX-diagram rewriting \cite{Kissinger2019}, one can also retain the structure of a circuit and just use the PDDAG structure to spot where non-adjacent gates can be merged via the rotation merging rewrites. This is typically good when the original circuit already has a relatively low density of Clifford gates when it is unlikely that resynthesis will give as efficient a circuit structure.

\section{Circuit Extraction}\label{sec:Extraction}


The goal of circuit extraction is to identify a sequence of gates that implements the same linear map as a given measurement pattern. Simply calculating the linear map and using standard matrix decomposition techniques is not sufficient since it will scale exponentially with the size of the pattern, making it impractical in the long-run. The method presented here will make use of a Pauli flow to determine the effect that each measurement angle has on the outputs. Doing so will yield the PDDAG representation directly, giving a rotation per planar measurement and a stabilizer process.


To motivate our method, recall the possible measurement bases from Section \ref{sec:MBQC}. Each planar basis projection can be constructed as a basic rotation of some Pauli basis:
\begin{equation}\label{eq:RotateBasis}
\begin{split}
\bra{\pm_{XY, \alpha}} &\approx \bra{\pm_{X, 0}} e^{i\tfrac{\alpha}{2} Z} \\
\bra{\pm_{XZ, \alpha}} &\approx \bra{\pm_{Z, 0}} e^{i\tfrac{\alpha}{2} Y} \\
\bra{\pm_{YZ, \alpha}} &\approx \bra{\pm_{Z, 0}} e^{-i\tfrac{\alpha}{2} X}
\end{split}
\end{equation}

The key idea driving this method of circuit extraction is to apply the \ref{lemma:GadgetStabCorrespondence} to alter these rotations and move them to onto the output qubits, but this requires identifying an appropriate stabilizer of the graph state to use. In a similar way, flow conditions already describe stabilizers for propagating errors onto other qubits. Pauli flow is one such example which can handle both planar measurements $\lambda(v) \in \{XY, XZ, YZ\}$ and special cases for Pauli measurements $\lambda(v) \in \{X, Y, Z\}$. Similar to gflow, it gives a stabilizer that applies a Pauli $X$ to vertices in $p(v)$ and a Pauli $Z$ to vertices in $\mathrm{Odd}(p(v))$. However, because Equation \ref{eq:PostEigen} means Pauli corrections have no observable effect if the qubit is measured in the same basis, it doesn't matter if, for example, a vertex $u \in p(v)$ with $\lambda(u) = X$ has already been measured before $v$.

\begin{equation}
\bra{+_{X, \alpha(u)}}_u E_G N_{\overline{I}} \approx \bra{+_{X, \alpha(u)}}_u \left( \prod_{\substack{w \in p(v) \\ w \neq u}} X_w \right) \left( \prod_{\substack{w \in \mathrm{Odd}(p(v))}} Z_w \right) E_G N_{\overline{I}}
\end{equation}

\begin{definition}[Pauli flow \cite{Browne, Perdrix2017}]
Given a labelled open graph $\Gamma = (G, I, O, \lambda)$, a \textit{Pauli flow} for $\Gamma$ is a tuple $(p, \prec)$ of a map $p : \overline{O} \to \mathcal{P}[\overline{I}]$ and a strict partial order $\prec$ over $V$ such that for all $u \in \overline{O}$:
\begin{enumerate}
\item[\namedlabel{PF1}{\textcolor{blue}{$[\prec .X]$}}] $\forall v \in p(u) . u \neq v \wedge \lambda(v) \notin \{X, Y\} \Rightarrow u \prec v$
\item[\namedlabel{PF2}{\textcolor{blue}{$[\prec .Z]$}}] $\forall v \in \mathrm{Odd}(p(u)) . u \neq v \wedge \lambda(v) \notin \{Y, Z\} \Rightarrow u \prec v$
\item[\namedlabel{PF3}{\textcolor{blue}{$[\prec .Y]$}}] $\forall v \preceq u . u \neq v \wedge \lambda(v) = Y \Rightarrow \left( v \in p(u) \Leftrightarrow v \in \mathrm{Odd}(p(u)) \right)$
\item[\namedlabel{PF4}{\textcolor{blue}{$[\lambda .XY]$}}] $\lambda(u) = (X, Y) \Rightarrow u \notin p(u) \wedge u \in \mathrm{Odd}(p(u))$
\item[\namedlabel{PF5}{\textcolor{blue}{$[\lambda .XZ]$}}] $\lambda(u) = (X, Z) \Rightarrow u \in p(u) \wedge u \in \mathrm{Odd}(p(u))$
\item[\namedlabel{PF6}{\textcolor{blue}{$[\lambda .YZ]$}}] $\lambda(u) = (Y, Z) \Rightarrow u \in p(u) \wedge u \notin \mathrm{Odd}(p(u))$
\item[\namedlabel{PF7}{\textcolor{blue}{$[\lambda .X]$}}] $\lambda(u) = X \Rightarrow u \in \mathrm{Odd}(p(u))$
\item[\namedlabel{PF8}{\textcolor{blue}{$[\lambda .Z]$}}] $\lambda(u) = Z \Rightarrow u \in p(u)$
\item[\namedlabel{PF9}{\textcolor{blue}{$[\lambda .Y]$}}] $\lambda(u) = Y \Rightarrow (u \notin p(u) \wedge u \in \mathrm{Odd}(p(u))) \vee (u \in p(u) \wedge u \notin \mathrm{Odd}(p(u)))$
\end{enumerate}
where $u \preceq v := \neg(v \prec u)$.
\end{definition}

\begin{definition}[Extraction string]\label{def:ExtractionStringDef}
Let $(\Gamma, \alpha)$ describe a measurement pattern with Pauli flow $(p, \prec)$ and some chosen measured vertex $v \in \overline{O}$. A Pauli string $\mathbf{P}$ over the outputs is a \textit{$P$-extraction string} ($P \in \{X, Y, Z\}$) for $v$ if $P_v \mathbf{P}$ is a stabilizer of the linear map $\mathcal{C}$:
\begin{equation}\label{eq:ExtractionStringDef}
\mathcal{C} := \left( \prod_{\substack{w \in \overline{O} \\ w \succ v \\ \lambda(w) \in \{XY, XZ, YZ\}}} \bra{+_{\lambda(w), 0}}_w \right) \left( \prod_{\substack{w \in \overline{O} \setminus \{v\} \\ \lambda(w) \in \{X, Y, Z\}}} \bra{+_{\lambda(w), \alpha(w)}}_w \right) E_G N_{\overline{I}} \\
\end{equation}

A \textit{primary extraction string} $\mathbf{P}^{\bot v}$ for $v$ is a $P^{\bot v}$-extraction string where
\begin{equation}
P^{\bot v} := \begin{cases}
X & \text{if } v \in p(v) \wedge v \notin \mathrm{Odd}(p(v)) \\
Y & \text{if } v \in p(v) \wedge v \in \mathrm{Odd}(p(v)) \\
Z & \text{if } v \notin p(v) \wedge v \in \mathrm{Odd}(p(v)) \\
\end{cases}
\end{equation}
\end{definition}

Relating this definition to the goal of using the \ref{lemma:GadgetStabCorrespondence}, the cases for $P^{\bot v}$ match the rotations in Equation \ref{eq:RotateBasis}. The fact that the rest of the stabilizer is over the outputs means we are successfully removing it from the scope of the pattern. For the purposes of extraction, we can assume all measurement angles of future vertices have already been extracted, making them Pauli projections. The notion of a focussed flow \cite{Mhalla2008, Backens2020} then guarantees that all Pauli terms are absorbed, leaving something in the form of an extraction string.

\begin{definition}[Focussed]\label{def:FocussedFlow}
Given a labelled open graph $\Gamma$, a set $\hat{p} \subset \overline{I}$ is \textit{focussed over $S \subseteq \overline{O}$} if:
\begin{enumerate}
\item[\namedlabel{FOC1}{\textcolor{blue}{$[FX]$}}] $\forall w \in S \cap \hat{p} . \lambda(w) \in \{XY, X, Y\}$ 
\item[\namedlabel{FOC2}{\textcolor{blue}{$[FZ]$}}] $\forall w \in S \cap \mathrm{Odd}(\hat{p}) . \lambda(w) \in \{XZ, YZ, Y, Z\}$ 
\item[\namedlabel{FOC3}{\textcolor{blue}{$[FY]$}}] $\forall w \in S . \lambda(w) = Y \Rightarrow \left( w \in \hat{p} \Leftrightarrow w \in \mathrm{Odd}(\hat{p}) \right)$ 
\end{enumerate}

A \textit{focussed set} $\hat{p}$ for $\Gamma$ is focussed over $\overline{O}$. A Pauli flow $(p, \prec)$ is \textit{focussed} if $p(v)$ is focussed over $\overline{O} \setminus \{v\}$ for all $v \in \overline{O}$.
\end{definition}


\begin{lemma}\label{lemma:ExtractionFromFocussed}
Let $\Gamma = (G, I, O, \lambda)$ be a labelled open graph with some measurement angles $\alpha : \overline{O} \to \left[0, 2\pi \right)$ and a focussed Pauli flow $(p, \prec)$. Then for any vertex $v \in \overline{O}$, $p(v)$ determines a $P^{\bot v}$-extraction string.
\end{lemma}

\begin{proof}
Proof in Appendix (Lemma \ref{lemma:ExtractionFromFocussedProof}).
\end{proof}


To make this extraction technique practical, we need to be able to find a focussed Pauli flow when one exists. This follows similarly to the corresponding results for causal flow and gflow.

\begin{theorem}\label{thrm:Identification}(Generalisation of \cite[Theorem 2]{Mhalla2008}, \cite[Theorem 3.11]{Backens2020})
There exists an algorithm that decides whether a given labelled open graph has a Pauli flow, and that outputs such a Pauli flow if it exists. Moreover, this output is maximally delayed, and the algorithm completes deterministically in time that grows polynomially with the number of vertices in the graph.
\end{theorem}

\begin{proof}
Proof in Appendix (Theorem \ref{thrm:IdentificationProof}).
\end{proof}


\begin{lemma}\label{lemma:Focussing}(Generalisation of \cite[Proposition 3.14]{Backens2020})
If a labelled open graph has a Pauli flow, then it has a maximally delayed, focussed Pauli flow.
\end{lemma}

\begin{proof}
Proof in Appendix (Lemma \ref{lemma:FocussingProof}).
\end{proof}


We are now in the position where we can take any Pauli flow and extract all of the planar measurement angles from a measurement pattern. This just leaves a Clifford process which can be characterised completely by stabilizer theory. The rows of the isometry tableau describing how inputs are mapped are given by extraction strings of the inputs (taking input extensions of the measurement pattern as appropriate), and the remaining generators over the outputs are obtained from extraction strings.


\begin{theorem}\label{thrm:Extraction}(Generalisation of \cite[Theorem 5.5]{Backens2020})
Let $(\Gamma, \alpha)$ describe a measurement pattern where $\Gamma$ has a Pauli flow. Then there is an algorithm that identifies an equivalent circuit requiring no ancillae which completes in time polynomial in the number of vertices in $\Gamma$.
\end{theorem}

\begin{proof}
Proof in Appendix (Theorem \ref{thrm:ExtractionProof}).
\end{proof}

\section{Relating Rewrites}\label{sec:Rewrites}


The most common rewrites in PDDAGs are merging terms and moving Clifford rotations, most notably moving a Clifford phase between a node in the DAG and the initial Clifford process. The structure or redundancy being exploiting for optimisation is very clear and easy to understand. On the other hand, graph-theoretic rewrites used for measurement patterns or ZX-diagrams have less obvious interpretations in how they relate to optimisations in the gate-based model. To compare the two, we consider extracting a PDDAG from a measurement pattern before and after a rewrite to observe the changes in the order, Pauli strings, and phases of the rotations from each measurement or tableau row, and then find a sequence of simple PDDAG rewrites that produces the same effect. Detailed proofs and examples can be found in Appendix \ref{sec:RewritesProofs}.


Given the special treatment of Pauli measurements in our representation of measurement patterns, the simplest rewrite we can do is to relabel a vertex measured in a Pauli basis between a planar label and a Pauli label, without changing the graph or other vertices.

\begin{theorem}\label{thrm:FixPauliRelate}
Let $(\Gamma, \alpha)$ describe a measurement pattern with some vertex $u \in \overline{O}$ such that $\lambda(u) \in \{XY, XZ, YZ\}$ and $\alpha(u) \in \{0, \tfrac{\pi}{2}, \pi, \tfrac{3\pi}{2}\}$. Relabelling $u$ to the equivalent Pauli label corresponds to pushing the rotation from $u$ to the start of the PDDAG and absorbing it into the initial stabilizer process.
\end{theorem}

\begin{proof}
Proof in Appendix (Theorem \ref{thrm:FixPauliRelateProof}).
\end{proof}


Removing vertices from the measurement pattern typically involves reducing a vertex's measurement to the $Z$ basis, at which point it no longer needs to be entangled with the other qubits. We can consider doing this both when the vertex is labelled as a Pauli $Z$ measurement and as a planar measurement in the $XZ$ or $YZ$ planes.

\begin{theorem}\label{thrm:ZElimRelate}
Let $(\Gamma, \alpha)$ describe a measurement pattern with some vertex $u \in \overline{O}$ such that $\lambda(u) \in \{XZ, YZ, Z\}$ and $\alpha(u) \in \{0, \pi\}$. Eliminating $u$ from the graph corresponds to the following sequence of actions on the PDDAG:
\begin{enumerate}
\item If $u$ has a planar ($XZ$ or $YZ$) label, then its rotation is pulled from the rotation DAG into the stabilizer block;
\item For each neighbour $n$ of $u$ that is an output, a $Z_n$ rotation of $\alpha(u)$ is pulled from the stabilizer block through the entire rotation DAG to the end of the circuit;
\item For each neighbour $n$ of $u$ with $\lambda(n) = XY$, a $\mathbf{P}^{\bot n}$ rotation of $\alpha(u)$ is pulled from the stabilizer block and merged with the existing rotation for $n$.
\end{enumerate}
\end{theorem}

\begin{proof}
Proof in Appendix (Theorem \ref{thrm:ZElimRelateProof}).
\end{proof}

Here, we have a slightly different effect for planar labels compared to the Pauli $Z$ label, though this can be thought of as a combination of the previous rewriting rule and the basic case for Pauli $Z$ elimination. The additional rotations generated are exactly the $Z$ gates introduced on neighbouring qubits to preserve the semantics (Lemma \ref{lemma:ZElimSemantics}).


One of the more interesting rewrites on measurement patterns is performing local complementation about some vertex $u$. This inverts the connectivity between each pair of vertices neighbouring $u$. To preserve the semantics, we also update the labels and measurement angles of $u$ and its neighbours. Combining this with $Z$ vertex elimination allows the elimination of vertices measured in the $Y$ basis.

\begin{theorem}\label{thrm:LocalCompRelate}
Performing local complementation about a vertex $u$ corresponds to the following sequence of actions on the PDDAG:

\begin{enumerate}
\item For each output $w$ neighbouring $u$, a $(Z_w, \tfrac{\pi}{2})$ rotation is pulled from the initial stabilizer block all the way to the end of the rotation DAG;
\item If $u$ is an output, a $(X_u, -\tfrac{\pi}{2})$ rotation is pulled from the initial stabilizer block all the way to the end of the rotation DAG;
\item For each vertex $w$ neighbouring $u$ with $\lambda(w) = XY$, and for $u$ itself if $\lambda(u)$ is planar, we pull a $\left(\mathbf{P'}^{\bot w}, (-1)^{D_w}\tfrac{\pi}{2}\right)$ rotation from the initial stabilizer block to merge it into the existing rotation for $w$ or $u$. We do this in $\succ$-order over such $w$ vertices.
\end{enumerate}
\end{theorem}

\begin{proof}
Proof in Appendix (Theorem \ref{thrm:LocalCompRelateProof}).
\end{proof}


Another similar operation to local complementation is the act of pivoting a diagram about an edge, which can be used to prepare a Pauli $X$ measurement for elimination. This action can be decomposed into a sequence of local complementations about each end of the edge, so this can also be simulated by some sequence of Clifford transformations in the PDDAG.

In each of the above, we assume that a particular transformation is made to the focussed Pauli flow and focussed sets for the measurement pattern based on the rewrite chosen, as detailed in Appendix \ref{sec:RewritesProofs}. In truth, focussed Pauli flows need not be unique for a measurement pattern. We can view the map between focussed Pauli flows as a rewrite on the measurement pattern that changes the corrections applied after each measurement, keeping the labelled open graph and measurement angles the same. We can compare the differences to the focussed sets for the pattern to obtain our final correspondence theorem:

\begin{theorem}\label{thrm:AddFlowsAndFocussedRelate}
Let $(\Gamma, \alpha)$ describe a measurement pattern with some focussed Pauli flow $(p, \prec)$, a focussed set $\hat{p}$, and some vertex $u \in \overline{O}$ such that $\forall w \in \hat{p} \cup \mathrm{Odd}(\hat{p}) . \lambda(w) \in \{XY, XZ, YZ\} \Rightarrow w \neq u \wedge u \prec w$. Updating $p(u)$ to $p(u) \Delta \hat{p}$ corresponds to a free action on the isometry tableau if $u$ is an input, and applying the \ref{lemma:GadgetStabCorrespondence} to the rotation from $u$ with the stabilizer of $\hat{p}$ if $u$ has a planar label.

Therefore, any two focussed Pauli flows for the same labelled open graph yield PDDAGs that are related by a sequence of applications of the \ref{lemma:GadgetStabCorrespondence} and free actions on the isometry tableau.
\end{theorem}

\begin{proof}
Proof in Appendix (Theorem \ref{thrm:AddFlowsAndFocussedRelateProof}).
\end{proof}

An interesting consequence of this result is the uniqueness of focussed Pauli flow for unitary patterns (up to weakening of the partial order), since there are no free stabilizers with which to apply the \ref{lemma:GadgetStabCorrespondence}.

\section{Conclusion}\label{sec:Conclusion}


We have investigated the significance of Pauli flow in the extraction of gate-based reversible circuits from measurement patterns, demonstrating that it is sufficient. This weakens the requirements for circuit extraction from universal measurement patterns compared to previous results. The principal contributions are a polynomial-time algorithm for identifying a Pauli flow for a measurement pattern (if one exists) and using this to construct an equivalent circuit without using ancillas or measurements.



The subsequent investigation using this map formally demonstrated that the effects of common graph-theoretic rewrites on measurement patterns/ZX-diagrams can be simulated via movement of Clifford rotations in a Pauli Dependency DAG. We know that the reverse simulation must be possible from the completeness of the ZX-calculus.

There are a number of possible avenues for future development regarding the topics of this paper:

\begin{itemize}
\item The simulation of rewrites demonstrates that ZX-calculus and PDDAGs can perform equivalent effects, but analysis of the computational complexities is needed to determine preferences for practical usage such as the actual data structure of choice for quantum circuit compilation. This will be heavily dependent on the choice of the graph structure used in each case and how this can capture the dependency relation in the PDDAG.
\item Given a PDDAG (or generic gate-based circuit), the work here provides some practical constraints on Pauli flows that any equivalent measurement pattern should admit. This could give rise to a reverse construction to find a minimal measurement pattern for a given PDDAG, allowing for explicit simulations for any PDDAG rewrite in measurement patterns.
\item The flexibility of the PDDAG structure could be extended by incorporating measurements, discards or decoherence, conditional gates, and more elaborate rewrites aiming for completeness for deciding circuit equivalence.
\end{itemize}

\newpage
\bibliography{library}
\bibliographystyle{eptcs}

\newpage
\appendix

\section{An Algorithm for Identifying Pauli Flow}\label{sec:PauliFlowProofs}

Here we will build up to an algorithm for identifying a Pauli flow from a graph, which follows the key principle of previous algorithms for causal flow and gflow \cite{Mhalla2008, Backens2020} of delaying measurements as long as possible in order to have as many vertices available for corrections as possible at each step. This is achieved by working backwards from the outputs and calculating the sets of vertices of a given ``correction depth''.

\begin{definition}
For a given labelled open graph $(G, I, O, \lambda)$ and a given Pauli flow $(p, \prec)$, for each $k \geq -1$ we define the vertices at depth $k$ under $\prec$ by
\begin{equation}
V_k^\prec := \begin{cases}
\emptyset & \text{if } k = -1 \\
\max_\prec(V \setminus \bigcup_{i < k} V_i^\prec) & \text{if } k \geq 0
\end{cases}
\end{equation}
the cumulative vertices up to depth $k$ by
\begin{equation}
V_{\cup k}^\prec := \bigcup_{i \leq k} V_i^\prec
\end{equation}
and the vertices in each plane correctable at depth $k+1$ by
\begin{align}
V_{k+1}^{\prec, XY} &:= \left\{u \in \overline{O} \middle\vert \begin{array}{l}
\lambda(u) \in \{XY, X, Y\} \wedge \exists K \subseteq (V_{\cup k}^\prec \cup \Lambda_u^X \cup \Lambda_u^Y) \cap \overline{I} . \\
K \cap \Lambda_u^Y \setminus V_{\cup k}^\prec = \mathrm{Odd}(K) \cap \Lambda_u^Y \setminus V_{\cup k}^\prec \\
\wedge \mathrm{Odd}(K) \setminus (V_{\cup k}^\prec \cup \Lambda_u^Y \cup \Lambda_u^Z) = \{u\}
\end{array} \right\} \\
V_{k+1}^{\prec, XZ} &:= \left\{u \in \overline{O} \middle\vert \begin{array}{l}
\lambda(u) \in \{XZ, X, Z\} \wedge \exists K \subseteq (V_{\cup k}^\prec \cup \Lambda_u^X \cup \Lambda_u^Y) \cap \overline{I} . \\
K \cap \Lambda_u^Y \setminus V_{\cup k}^\prec = \mathrm{Odd}(K \cup \{u\}) \cap \Lambda_u^Y \setminus V_{\cup k}^\prec \\
\wedge \mathrm{Odd}(K \cup \{u\}) \setminus (V_{\cup k}^\prec \cup \Lambda_u^Y \cup \Lambda_u^Z) = \{u\}
\end{array} \right\} \\
V_{k+1}^{\prec, YZ} &:= \left\{u \in \overline{O} \middle\vert \begin{array}{l}
\lambda(u) \in \{YZ, Y, Z\} \wedge \exists K \subseteq (V_{\cup k}^\prec \cup \Lambda_u^X \cup \Lambda_u^Y) \cap \overline{I} . \\
K \cap \Lambda_u^Y \setminus V_{\cup k}^\prec = \mathrm{Odd}(K \cup \{u\}) \cap \Lambda_u^Y \setminus V_{\cup k}^\prec \\
\wedge \mathrm{Odd}(K \cup \{u\}) \setminus (V_{\cup k}^\prec \cup \Lambda_u^Y \cup \Lambda_u^Z) = \emptyset
\end{array} \right\}
\end{align}
where $\Lambda_u^P := \{v \in \overline{O} | v \neq u \wedge \lambda(v) = P\}$.
\end{definition}

The sets $V_k^{\prec, XY}$, $V_k^{\prec, XZ}$, $V_k^{\prec, YZ}$ capture the vertices $v$ at correction depth $k$ for which the $v$ component of the correcting stabilizer is $Z$, $Y$, or $X$ respectively. These definitions are intended to mirror the Pauli flow conditions, with the restrictions on $\lambda(v)$ matching \ref{PF4}-\ref{PF9} and the other conditions capturing \ref{PF1}-\ref{PF3} using $K$ for $p(v) \setminus \{v\}$ and $V_{\cup k}^\prec$ to represent the future vertices under $\prec$.

\begin{definition}[Delayed]
Given two Pauli flows $(p, \prec)$ and $(p', \prec')$ for the same labelled open graph, $(p, \prec)$ is \textit{more delayed} than $(p', \prec')$ if for all $k$,
\begin{equation}
\left| V_{\cup k}^\prec \right| \geq \left| V_{\cup k}^{\prec'} \right|
\end{equation}
and there exists a $k$ for which this inequality is strict. A Pauli flow $(p, \prec)$ is \textit{maximally delayed} if there exists no Pauli flow over the same labelled open graph that is more delayed.
\end{definition}

\begin{remark}\label{rem:MaximallyDelayed}
If a labelled open graph admits a Pauli flow, then there must be a maximally delayed Pauli flow since the delayed relation is a strict partial order and there are only finitely many possible Pauli flows (the graph is finite, so there are only finite possible choices of maps $p$ and orders $\prec$). Any sequence of increasingly delayed flows must be finite, so we can always reach a maximal point. We will freely assume this without statement in subsequent proofs.
\end{remark}

For causal flow and gflow, the set $V_0^\prec$ for a maximally delayed flow is exactly the set of outputs \cite{Mhalla2008, Backens2020} since every measured vertex needs corrections. However, since Pauli flow allows some corrections to be effectively applied in the past of a vertex, we may now have some measured vertices in $V_0^\prec$. An example is given in Figure \ref{fig:MeasuredV0}.

\begin{figure}
\centering
\input{tikz_figs/mbqc_measured_v0.tikz}
\begin{tabular}{cc|ccc}
$v$ & $\lambda(v)$ & $p(v)$ & $\mathrm{Odd}(p(v))$ & $\{u | v \prec u\}$ \\
\hline
$i$ & $XY$ & $a$ & $i, b, o$ & $b, o$ \\
$a$ & $X$ & $o$ & $a$ & $o$ \\
$b$ & $XY$ & $c$ & $b$ & $-$ \\
$c$ & $X$ & $b, o$ & $c$ & $b, o$ \\
\end{tabular}
\caption{A labelled open graph and a maximally delayed Pauli flow with a measured vertex ($b$) in $V_0^\prec$. If $(p, \prec)$ is focussed, then the primary extraction string for such a vertex is the identity, meaning the measurement angle/outcome has no effect on the overall process.}\label{fig:MeasuredV0}
\end{figure}
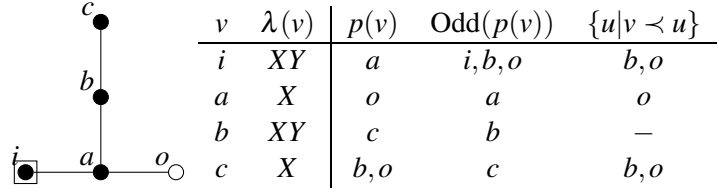

\begin{lemma}\label{lemma:V0}(Generalisation of \cite[Lemma 1]{Mhalla2008}, \cite[Lemma C.2]{Backens2020})
If $(p, \prec)$ is a maximally delayed Pauli flow of $(G, I, O, \lambda)$, then $V_0^\prec = O \cup V_0^{\prec, XY} \cup V_0^{\prec, XZ} \cup V_0^{\prec, YZ}$.
\end{lemma}

\begin{proof}
$V_0^\prec \subseteq O \cup V_0^{\prec, XY} \cup V_0^{\prec, XZ} \cup V_0^{\prec, YZ}$: Since the output case is trivial, suppose $u \in V_0^\prec$ is a non-output. We start by showing that $p(u) \subseteq (V_{-1}^\prec \cup \Lambda_u^X \cup \Lambda_u^Y \cup \{u\}) \cap \overline{I}$. $p(u) \subseteq \overline{I}$ from the definition of Pauli flow, since $p : \overline{O} \to \mathcal{P}[\overline{I}]$. By definition $V_{-1}^\prec = \emptyset$, so we just require $p(u) \subseteq \Lambda_u^X \cup \Lambda_u^Y \cup \{u\}$ which is given by condition \ref{PF1}. It remains to show that $p(u) \setminus \{u\}$ is a suitable choice for $K$ for at least one of the three $V_0^{\prec, PQ}$ sets.

Every vertex $v$ satisfies $v \preceq u$ ($u \in V_0^\prec$), so condition \ref{PF2} becomes $\mathrm{Odd}(p(u)) \setminus (\Lambda_u^Y \cup \Lambda_u^Z) \subseteq \{u\}$. Condition \ref{PF3} gives $p(u) \cap \Lambda_u^Y = \mathrm{Odd}(p(u)) \cap \Lambda_u^Y$. We can rewrite these to the appropriate form from the definition of the $V_0^{\prec, PQ}$ sets based on whether or not $u$ is in $p(u)$ and $\mathrm{Odd}(p(u))$. As the cases are mutually exclusive, we only place $u$ into a single set with conditions \ref{PF4}-\ref{PF9} restricting $\lambda(u)$ as needed.

$V_0^\prec \supseteq O \cup V_0^{\prec, XY} \cup V_0^{\prec, XZ} \cup V_0^{\prec, YZ}$: We will assume that there is some $w \in O \cup V_0^{\prec, XY} \cup V_0^{\prec, XZ} \cup V_0^{\prec, YZ}$ but $w \notin V_0^\prec$ and aim for a contradiction by generating a new Pauli flow that is more delayed than $(p, \prec)$. If $w \in O$, then we can simply construct the new Pauli flow $(p, \prec \setminus (\{w\} \times V))$. This is still a Pauli flow since the conditions only concern the measured vertices $\overline{O}$. This is also more delayed since $w \in V_0^{\prec \setminus (\{w\} \times V)}$ and no vertex is made deeper.

For the remaining case, we shall suppose that $w \in V_0^{\prec, XY} \cup V_0^{\prec, XZ} \cup V_0^{\prec, YZ}$ instead. Let $K \subseteq (\Lambda_w^X \cup \Lambda_w^y) \cap \overline{I}$ be a witness for $w \in V_0^{\prec, XY} \cup V_0^{\prec, XZ} \cup V_0^{\prec, YZ}$. We define $p' : \overline{O} \to \mathcal{P}[\overline{I}]$ and $\prec'$ be the map and partial order:

\begin{equation}
p'(u) := \begin{cases}
K & \text{if } u = w \text{ and } w \in V_0^{\prec, XY} \\
K \cup \{w\} & \text{if } u = w \text{ and } w \notin V_0^{\prec, XY} \\
p(u) & \text{if } u \neq w
\end{cases}
\end{equation}

\begin{equation}
\prec' := \prec \setminus (\{w\} \times V)
\end{equation}

For $u \neq w$, each of the Pauli flow conditions is trivially preserved, but we need to show they hold for $p'(w)$.

This satisfies condition \ref{PF1} from $K \subseteq (\Lambda_w^X \cup \Lambda_w^Y) \cap \overline{I}$. For each of the $V_0^{\prec, PQ}$ options, we get $\mathrm{Odd}(p'(w)) \setminus \left(\Lambda_w^Y \cup \Lambda_w^Z\right) \subseteq \{w\}$, so we can show \ref{PF2}. Similarly, we get $p'(w) \cap \Lambda_w^Y = \mathrm{Odd}(p'(w)) \cap \Lambda_w^Y$ in each of the $V_0^{\prec, PQ}$ options, so condition \ref{PF3} also follows.

The membership of $w$ in $p'(w)$ and $\mathrm{Odd}(p'(w))$ is decided by which of the $V_0^{\prec, PQ}$ it belongs to, which also restricts $\lambda(w)$ to the compatible planes according to conditions \ref{PF4}-\ref{PF9}.

By assumption, $w \in V_k^\prec$ for some $k > 0$. By construction, the depth of every vertex under $\prec'$ is no larger than its depth under $\prec$ and $w \in V_0^{\prec'}$. This means $\forall k \geq 0 . \left|V_{\cup k}^\prec \right| \leq \left|V_{\cup k}^{\prec'}\right|$ and $\left|V_0^\prec\right| < \left|V_0^{\prec'}\right|$, so $(p', \prec')$ is a more delayed Pauli flow than $(p, \prec)$.
\end{proof}

\begin{lemma}\label{lemma:Vk}(Generalisation of \cite[Lemma 3]{Mhalla2008}, \cite[Lemma C.4]{Backens2020})
If $(p, \prec)$ is a maximally delayed Pauli flow of $(G, I, O, \lambda)$, then $\forall k > 0 . V_k^\prec = V_k^{\prec, XY} \cup V_k^{\prec, XZ} \cup V_k^{\prec, YZ}$.
\end{lemma}

\begin{proof}
For any given $k$, the proof of this follows the same strategy as Lemma \ref{lemma:V0}. When constructing the more delayed flow given a counterexample $w$, we need to adapt the definition of $\prec'$ to $\prec' := \left( \prec \setminus (\{w\} \times V) \right) \cup (\{w\} \times V_{\cup k-1}^\prec)$. Vertices in $V_{\cup k-1}^\prec$ are not covered by the constraints in the definitions of $V_{k}^{\prec, PQ}$, so we add them all to the future of $w$ to make sure we satisfy conditions \ref{PF1}-\ref{PF3}.
\end{proof}

These characterisations of the sets $V_k^\prec$ give us an iterative method of identifying them since we can simply search for possible witness sets $K$ for each vertex.

\begin{algorithm}\label{algo:FlowIdentify}
$\texttt{PauliFlow}(V, \Gamma, I, O, \lambda)$ = 
\Begin{
    $L_X := \emptyset$; $L_Y := \emptyset$; $L_Z := \emptyset$\;
    \ForAll{$v \in V$}{
        \lIf{$v \in O$}{$d(v) := 0$}
        \lIf{$\lambda(v) = X$}{$L_X := L_X \cup \{v\}$}
        \lIf{$\lambda(v) = Y$}{$L_Y := L_Y \cup \{v\}$}
        \lIf{$\lambda(v) = Z$}{$L_Y := L_Z \cup \{v\}$}
    }
    \Return{$\texttt{PauliFlowAux}(V, \Gamma, I, \lambda, \emptyset, O, 0)$}\;
}
$\texttt{PauliFlowAux}(V, \Gamma, I, \lambda, A, B, k)$ =
\Begin{
    $C := \emptyset$\;
    \ForAll{$u \in B^C$}{
        \If{$\lambda(u) \in \{XY, X, Y\}$}{
            $K_{XY} := $ Solution for $K \subseteq (A \cup L_X \cup L_Y) \cap I^C \setminus \{u\}$ where $K \cap L_Y \setminus (\{u\} \cup A) = \mathrm{Odd}(K) \cap L_Y \setminus (\{u\} \cup A)$ and $\mathrm{Odd}(K) \cap ((A \cup L_Y \cup L_Z)^C \cup \{u\}) = \{u\}$;
        }
        \If{$\lambda(u) \in \{XZ, X, Z\}$}{
            $K_{XZ} := \{u\} \cup $ (Solution for $K \subseteq (A \cup L_X \cup L_Y) \cap I^C \setminus \{u\}$ where $K \cap L_Y \setminus (\{u\} \cup A) = \mathrm{Odd}(K \cup \{u\}) \cap L_Y \setminus (\{u\} \cup A)$ and $\mathrm{Odd}(K \cup \{u\}) \cap ((A \cup L_Y \cup L_Z)^C \cup \{u\}) = \{u\}$);
        }
        \If{$\lambda(u) \in \{YZ, Y, Z\}$}{
            $K_{YZ} := \{u\} \cup $ (Solution for $K \subseteq (A \cup L_X \cup L_Y) \cap I^C \setminus \{u\}$ where $K \cap L_Y \setminus (\{u\} \cup A) = \mathrm{Odd}(K \cup \{u\}) \cap L_Y \setminus (\{u\} \cup A)$ and $\mathrm{Odd}(K \cup \{u\}) \cap ((A \cup L_Y \cup L_Z)^C \cup \{u\}) = \emptyset$);
        }
        \If{a solution $K_0$ is found for any of $K_{XY}, K_{XZ}, K_{YZ}$}{
            $C := C \cup \{u\}$\;
            $p(u) := K_0$\;
            $d(u) := k$;
        }
    }
    \If{$C = \emptyset$ \text{ and } $k > 0$}{
        \lIf{$B = V$}{\Return{$(\mathrm{true}, p, d)$}}
        \lElse{\Return{$(\mathrm{false}, \emptyset, \emptyset)$}}
    }
    \Else{
        $B' := B \cup C$\;
        \Return{$\texttt{PauliFlowAux}(V, \Gamma, I, \lambda, B', B', k+1)$};
    }
}
\caption{An algorithm for identifying whether a labelled open graph has a Pauli flow.}
\end{algorithm}

\begin{theorem}\label{thrm:IdentificationProof}(Restatement of Theorem \ref{thrm:Identification})
There exists an algorithm that decides whether a given labelled open graph has a Pauli flow, and that outputs such a Pauli flow if it exists. Moreover, this output is maximally delayed, and the algorithm completes deterministically in time that grows polynomially with the number of vertices in the graph.
\end{theorem}

\begin{proof}
The function $\texttt{PauliFlow}(V, \Gamma, I, O, \lambda)$ in Algorithm \ref{algo:FlowIdentify} takes sets $V$ and $I, O \subseteq V$ of vertices, an adjacency matrix $\Gamma$ over $V$ and a basis labelling function $\lambda$ and returns ``$\mathrm{true}$'' with a Pauli flow if one exists and ``$\mathrm{false}$'' otherwise.

To see this, we consider the auxilliary method $\texttt{PauliFlowAux}(V, \Gamma, I, \lambda, A, B, k)$ which aims to identify $V_k^\prec$, given sets $A = V_{\cup k-1}^\prec$ of possible correctors and $B \subseteq V_{\cup k}^\prec$ ($B \supseteq A$) of solved vertices, and then proceed recursively over the remainder of the graph for higher depths. For the case of $k = 0$, $O \subseteq V_0^\prec$ has already been handled by the setup in $\texttt{PauliFlow}(V, \Gamma, I, O, \lambda)$ so it is sufficient to just find $V_k^{\prec, XY} \cup V_k^{\prec, XZ} \cup V_k^{\prec, YZ}$ due to Lemma \ref{lemma:V0}. For $k > 0$, $A = B$ so we are just finding $V_k^{\prec, XY} \cup V_k^{\prec, XZ} \cup V_k^{\prec, YZ}$ in accordance with Lemma \ref{lemma:Vk}. In each recursive call, it examines each candidate vertex in turn and looks for a witness set $K$ for membership into either $V_k^{\prec, XY}$, $V_k^{\prec, XZ}$, or $V_k^{\prec, YZ}$, from which we can identify a valid correction set. The global variables $p$ and $d$ map vertices to their correction sets and depths from the output (this defines the order $\prec$ over vertices with $v \prec w \Leftrightarrow d(v) > d(w)$).

This algorithm is guaranteed to terminate because $V$ is finite and $B \subseteq V$ strictly grows with each call for $k \geq 1$, so $C$ will eventually be an empty set.

The resulting Pauli flow must be maximally delayed because we are constructing the largest set possible at each value of $k$. If there were a more delayed Pauli flow $(p', \prec')$, there must be some minimal $k$ for which $V_k^{\prec'} \setminus V_k^\prec$ is non-empty. However, if a vertex $v$ belonged in this set then a suitable correction set for $v$ must exist. Algorithm \ref{algo:FlowIdentify} would find this when testing $v$ and place it in $V_k^\prec$.

For a given vertex $u$, finding the witness set $K$ in each of the three possible search cases can be achieved by solving the linear equation system $M_{A, u} X_K = S_{\tilde{\lambda}}$ in $\mathbb{F}_2$ defined by:

\begin{align}
M_{A, u} &:= \left[ \begin{array}{c}
\Gamma \cap \mathbb{K}_{A, u} \times \mathbb{P}_{A, u} \\
\hline
(\Gamma + \mathrm{Id}) \cap \mathbb{K}_{A, u} \times \mathbb{Y}_{A, u}
\end{array} \right] \\
S_{\tilde{\lambda}} &:= \begin{cases}
\left[ \begin{array}{c}
\{u\} \\
\hline
0
\end{array} \right] & \text{if } \tilde{\lambda} = XY \\
\left[ \begin{array}{c}
(N_\Gamma(u) \cap \mathbb{P}_{A, u}) \cup \{u\} \\
\hline
N_\Gamma(u) \cap \mathbb{Y}_{A, u}
\end{array} \right] & \text{if } \tilde{\lambda} = XZ \\
\left[ \begin{array}{c}
N_\Gamma(u) \cap \mathbb{P}_{A, u} \\
\hline
N_\Gamma(u) \cap \mathbb{Y}_{A, u}
\end{array} \right] & \text{if } \tilde{\lambda} = YZ \\
\end{cases}
\end{align}
where $\mathbb{K}_{A, u} := (A \cup \Lambda_u^X \cup \Lambda_u^Y) \cap I^C$ is the set of possible elements of the witness set $K$, $\mathbb{P}_{A, u} := (A \cup \Lambda_u^Y \cup \Lambda_u^Z)^C$ is the set of vertices in the past/present which should remain corrected afer measuring and correcting $u$, $\mathbb{Y}_{A, u} := \Lambda_u^Y \setminus A$ is the set of vertices we have to consider for condition \ref{PF3}, $\tilde{\lambda} \in \{XY, XZ, YZ\}$ denotes which of the three cases we are considering, and $X_K$ is the column vector with $1$ in the position of $v$ if $v \in K$ and a $0$ otherwise.

Taking $\tilde{\lambda} = XY$ as an example, the top block of equations encodes $\mathrm{Odd}(K) \cap \mathbb{P}_{A, u} = \{u\}$ and the lower block encodes $K \cap \Lambda_u^Y \setminus A = \mathrm{Odd}(K) \cap \Lambda_u^Y \setminus A$, so the solutions to these equations are exactly the possible witness sets $K$. Such solutions can be identified by Gaussian elimination and back substitution in $O(|V|^3)$ time.

This part of the algorithm is hit at most $|V|$ times per call to $\texttt{PauliFlowAux}$, which may be called at most $|V|$ times, hence the overall complexity of this algorithm is $O(|V|^5)$.
\end{proof}

The complexity of this algorithm is higher than the $O(|V|^4)$ for the equivalent method for identifying gflow \cite{Backens2020} \footnote{It should be noted that Eslamy et al. provide an even more efficient algorithm for finding gflow in $O(|V|^3)$ time \cite{Eslamy2018}. It may be possible to generalise this to Pauli flow in search of a more efficient routine.}. This is because we can no longer do a single gaussian elimination per depth round because the matrix $M_{A, u}$ and the range of the witness set $K$ are both dependent on the particular vertex $u$ under consideration.

\section{Focussed Sets and Pauli Flows}\label{sec:Focussed}

Even though the Pauli flow generated by Algorithm \ref{algo:FlowIdentify} is guaranteed to be maximally delayed, the correction sets may not be unique as the back substitution step may not have a single solution. In this section, we will investigate some transformations on Pauli flows resulting from this freedom to show that we can always reach a focussed Pauli flow.

\begin{lemma}\label{lemma:AddFlows}
Given a Pauli flow $(p, \prec)$ for a labelled open graph $(G, I, O, \lambda)$ with two vertices $u, v \in \overline{O}$ such that $u \prec v$, then $(p', \prec)$ is a Pauli flow where $p'(u) := p(u) \Delta p(v)$ and $\forall w \in \overline{O} \setminus \{u\} . p'(w) := p(w)$. Moreover, if $(p, \prec)$ is maximally delayed, then so is $(p', \prec)$.
\end{lemma}

\begin{proof}
The Pauli flow conditions hold trivially for any vertex in $\overline{O} \setminus \{u\}$ since the correction sets have not changed, so it is sufficient to show they are preserved for $u$. We should first observe that $\mathrm{Odd}(p'(u)) = \mathrm{Odd}(p(u) \Delta p(v)) = \mathrm{Odd}(p(u)) \Delta \mathrm{Odd}(p(v))$.

\ref{PF1}: For any $w \in p'(u)$ with $w \neq u$ and $\lambda(w) \notin \{X, Y\}$, we must have either $w \in p(u)$, $w = v$, or $w \in p(v) \wedge w \neq v$. In any of these cases, we have $u \prec w$ from \ref{PF1} for $(p, \prec)$ and $u \prec v$.

\ref{PF2}: This follows similarly from $u \prec v$ and \ref{PF2} on $\mathrm{Odd}(p(u))$ and $\mathrm{Odd}(p(v))$.

\ref{PF3}: For any $w \preceq u$ with $\lambda(w) = Y$, we also must have $w \preceq v$ since $u \prec v$. Hence by \ref{PF3}, $w \in p(u) \Leftrightarrow w \in \mathrm{Odd}(p(u))$ and the same for $p(v)$. Therefore, we find that $w \in p'(u) = p(u) \Delta p(v) \Leftrightarrow w \in \mathrm{Odd}(p(u)) \Delta \mathrm{Odd}(p(v)) = \mathrm{Odd}(p'(u))$ as required.

\ref{PF4}-\ref{PF6}: $u \notin p(v)$ and $u \notin \mathrm{Odd}(p(v))$ by \ref{PF1} and \ref{PF2} since $u \prec v$, so the requirements are given by the corresponding conditions for $(p, \prec)$.

\ref{PF7}: $u \notin \mathrm{Odd}(p(v))$ by \ref{PF2} and $u \in \mathrm{Odd}(p(u))$ by \ref{PF7}, so $u \in \mathrm{Odd}(p'(u))$.

\ref{PF8}: $u \notin p(v)$ by \ref{PF1} and $u \in p(u)$ by \ref{PF8}, so $u \in p'(u)$.

\ref{PF9}: $u \in p(v) \Leftrightarrow u \in \mathrm{Odd}(p(v))$ by \ref{PF3} and $u \in p(u) \Leftrightarrow u \notin \mathrm{Odd}(p(u))$ by \ref{PF9}, then it is straightforward to show $u \in p'(u) \Leftrightarrow u \notin \mathrm{Odd}(p'(u))$ by cases.

The maximally delayed property of a Pauli flow only concerns the partial order between the vertices, so since $(p, \prec)$ and $(p', \prec)$ both use $\prec$ the property is trivially preserved.
\end{proof}

This gives us a mechanism to generate new Pauli flows by adding correction sets together. We now show that this can help us to make progress towards satisfying the focussed property.


\begin{lemma}\label{lemma:AddFocussed}
For any labelled open graph $\Gamma$, if sets $\hat{p}, \hat{q} \subseteq \overline{I}$ are focussed over $S \subseteq \overline{O}$, then so is $\hat{p} \Delta \hat{q}$.
\end{lemma}

\begin{proof}
\ref{FOC1}: For any vertex $v \in (\hat{p} \Delta \hat{q}) \cap S$, we have either $v \in \hat{p}$ or $v \in \hat{q}$. Since $\hat{p}$ and $\hat{q}$ are focussed over $S$, the corresponding \ref{FOC1} condition gives $\lambda(v) \in \{XY, X, Y\}$.

\ref{FOC2}: Similarly, for any vertex $v \in \mathrm{Odd}(\hat{p} \Delta \hat{q}) \cap S = (\mathrm{Odd}(\hat{p}) \Delta \mathrm{Odd}(\hat{q})) \cap S$, $\lambda(v) \in \{XZ, YZ, Y, Z\}$ follows from either $v \in \mathrm{Odd}(\hat{p})$ or $v \in \mathrm{Odd}(\hat{q})$ and \ref{FOC2}.

\ref{FOC3}: For any $v \in S$ with $\lambda(v) = Y$, we have $v \in \hat{p} \Leftrightarrow v \in \mathrm{Odd}(\hat{p})$ and $v \in \hat{q} \Leftrightarrow v \in \mathrm{Odd}(\hat{q})$ from \ref{FOC3} for $\hat{p}$ and $\hat{q}$. Hence $v \in \hat{p} \Delta \hat{q} \Leftrightarrow v \in \mathrm{Odd}(\hat{p}) \Delta \mathrm{Odd}(\hat{q})$.
\end{proof}

\begin{lemma}\label{lemma:AddNotFocussed}
For any labelled open graph $\Gamma$, if sets $\hat{p}, \hat{q} \subseteq \overline{I}$ are not focussed over $\{v\}$ ($v \in \overline{O}$), then $\hat{p} \Delta \hat{q}$ is focussed over $\{v\}$.
\end{lemma}

\begin{proof}
We consider each case for $\lambda(v)$ and how the focussed conditions could fail for $\hat{p}$ and $\hat{q}$:

$\lambda(v) \in \{XY, X\}$: \ref{FOC1} and \ref{FOC3} are trivially satisfied, so we must have $v \in \mathrm{Odd}(\hat{p})$ and $v \in \mathrm{Odd}(\hat{q})$ to fail \ref{FOC2}. This means $v \notin \mathrm{Odd}(\hat{p}) \Delta \mathrm{Odd}(\hat{q}) = \mathrm{Odd}(\hat{p} \Delta \hat{q})$, satisfying \ref{FOC2} for $\hat{p} \Delta \hat{q}$.

$\lambda(v) \in \{XZ, YZ, Z\}$: Similarly, \ref{FOC2} and \ref{FOC3} hold trivially, so we must have $v \in \hat{p}$ and $v \in \hat{q}$ to fail \ref{FOC1}. We hence have $v \notin \hat{p} \Delta \hat{q}$, satisfying \ref{FOC2} for $\hat{p} \Delta \hat{q}$.

$\lambda(v) = Y$: Now \ref{FOC1} and \ref{FOC2} are trivial and we have $v \in \hat{p} \Delta \mathrm{Odd}(\hat{p})$ and $v \in \hat{q} \Delta \mathrm{Odd}(\hat{q})$ to fail \ref{FOC3}. This now satisfies \ref{FOC3} since $v \notin \left( \hat{p} \Delta \mathrm{Odd}(\hat{p}) \right) \Delta \left( \hat{q} \Delta \mathrm{Odd}(\hat{q}) \right) = (\hat{p} \Delta \hat{q}) \Delta \mathrm{Odd}(\hat{p} \Delta \hat{q})$.
\end{proof}

Combining these two lemmas, we can find combinations of correction sets that fix unfocussed vertices whilst preserving those we have already focussed.


\begin{lemma}\label{lemma:FocusStep}(Generalisation of \cite[Lemma 3.13]{Backens2020})
Let $(G, I, O, \lambda)$ be a labelled open graph with a Pauli flow $(p, \prec)$ and some vertex $v \in \overline{O}$. Then there exists $p' : \overline{O} \to \mathcal{P}[\overline{I}]$ such that:

\begin{enumerate}
\item\label{FocusStepSame} $\forall w \in \overline{O} . v = w \vee p'(w) = p(w)$;
\item\label{FocusStepFocus} $p'(v)$ is focussed over $\overline{O} \setminus \{v\}$;
\item\label{FocusStepFlow} $(p', \prec)$ is a Pauli flow for $(G, I, O, \lambda)$.
\end{enumerate}
\end{lemma}

\begin{proof}
Let $J : \mathbb{Z}_{|\overline{O}|} \to \overline{O}$ be some indexing of the vertices that respects the order $\prec$ ($\forall i, j < |\overline{O}| . J(i) \prec J(j) \Rightarrow i < j$). We define a sequence of functions $p_k$ as:

\begin{align}
p_0(u) &:= p(u) \\
p_{k+1}(u) &:= \begin{cases}
p_k(u) \Delta p_k(J(k)) & \text{if } \left( \begin{array}{l} u = v, \\ J(k) \neq v, \\ p_k(v) \text{ is not focussed over } \{J(k)\} \end{array} \right) \\
p_k(u) & \text{otherwise} \\
\end{cases}
\end{align}

\ref{FocusStepSame} is satisfied for all $(p_k, \prec)$ by construction, and \ref{FocusStepFlow} is also satisfied for all by Lemma \ref{lemma:AddFlows}. To work towards \ref{FocusStepFocus}, we proceed inductively with hypothesis $\Phi(k) := \text{``} p_k(v) \text{ is focussed over } \{J(i)\}_{i < k} \setminus \{v\} \text{''}$. The $k = 0$ case holds vacuously.

Suppose we have $\Phi(k)$. If $J(k) = v$, then $p_{k+1}(v) = p_k(v)$ and $\Phi(k+1)$ is an immediate consequence of $\Phi(k)$. If $p_k(v)$ is focussed over $\{J(k)\}$, then $p_{k+1}(v) = p_k(v)$, so $\Phi(k+1)$ follows from this assumption and $\Phi(k)$. If, on the other hand, $p_k(v)$ is not focussed over $\{J(k)\}$, we have $p_{k+1}(v) = p_k(v) \Delta p_k(J(k))$. From conditions \ref{PF4}-\ref{PF9}, $p_k(J(k))$ is also not focussed over $\{J(k)\}$, so by Lemma \ref{lemma:AddNotFocussed} we have $p_{k+1}(v)$ is focussed over $\{J(k)\}$. For any of the remaining $i < k$ (where $J(i) \neq v$), $\Phi(k)$ says that $p_k(v)$ is focussed over $\{J(i)\}$. Since $J$ respects the order $\prec$, we also have $J(i) \preceq J(k)$, and hence conditions \ref{PF1}-\ref{PF3} imply that $p_k(J(k))$ is focussed over $\{J(i)\}$. We combine these with Lemma \ref{lemma:AddFocussed} to deduce that $p_{k+1}(v)$ is also focussed over $\{J(i)\}$.

At the end of this chain, we have $(p_{|\overline{O}|}, \prec)$ where $p_{|\overline{O}|}$ is focussed over $\overline{O} \setminus \{v\}$.
\end{proof}


\begin{lemma}\label{lemma:FocussingProof}(Restatement of Lemma \ref{lemma:Focussing})
If a labelled open graph has a Pauli flow, then it has a maximally delayed, focussed Pauli flow.
\end{lemma}

\begin{proof}
Let $(p, \prec)$ be such a maximally delayed Pauli flow according to Remark \ref{rem:MaximallyDelayed}. Applying Lemma \ref{lemma:FocusStep} for each $v \in \overline{O}$ in turn, we reach some $(p', \prec)$ where, for every $v \in \overline{O}$, $p'(v)$ is focussed over $\overline{O} \setminus \{v\}$, i.e. $(p', \prec)$ is a focussed Pauli flow. Since the partial order $\prec$ remains the same, this is also still maximally delayed.
\end{proof}

For general measurement patterns, even focussed Pauli flows may not be unique. However, given multiple focussed Pauli flows, the differences between their correction sets are given by the focussed sets of the graph.


\begin{lemma}\label{lemma:AddFlowsAndFocussed}
Let $\Gamma = (G, I, O, \lambda)$ be a labelled open graph with a focussed Pauli flow $(p, \prec)$ and a focussed set $\hat{p} \subseteq \overline{I}$. Let $v \in \overline{O}$ be a vertex such that $\forall w \in \hat{p} \cup \mathrm{Odd}(\hat{p}) . \lambda(w) \in \{XY, XZ, YZ\} \Rightarrow w \neq v \wedge v \preceq w$. Then $(p', \prec')$ is a focussed Pauli flow, where:

\begin{equation}
p'(w) := \begin{cases}
p(w) \Delta \hat{p} & \text{if } w = v \\
p(w) & \text{if } w \neq v
\end{cases}
\end{equation}
and $\prec'$ is the transitive closure of $\prec \cup \{(v, w) | w \in \hat{p} \cup \mathrm{Odd}(\hat{p}) \wedge \lambda(w) \in \{XY, XZ, YZ\}\}$.
\end{lemma}

\begin{proof}
Firstly, $\prec'$ is still a strict partial order. Transitivity is immediate from the definition as a transitive closure. For antisymmetry (and similarly for strictness), suppose on the contrary that we have some $a \prec' b$ and $b \prec' a$ (or directly $a \prec' a$). This gives a transitive loop $[a, \ldots, b, \ldots, a]$ where each step is either in $\prec$ or $\{(v, w) | w \in \hat{p} \cup \mathrm{Odd}(\hat{p}) \wedge \lambda(w) \in \{XY, XZ, YZ\}\}$. Since $\prec$ is a strict partial order, we cannot have such a loop where every step is in $\prec$, so at least one step must be some such $(v, w)$. We can freely eliminate inner loops around $v$, so we can assume wlog that $v$ only appears once in the loop. This means the rest of the loop is only from $\prec$, so $w \prec v$ by transitivity. However, we assumed that $v \preceq w$ since $w \in \hat{p} \cup \mathrm{Odd}(\hat{p})$ and $\lambda(w) \in \{XY, XZ, YZ\}$, giving us the contradiction we need.

Conditions \ref{PF1}-\ref{PF3} are preserved from the extension to $\prec'$ covering planar labels and the focussed property covering Pauli labels.

Conditions \ref{PF4}-\ref{PF6} are preserved since $v \notin \hat{p} \cup \mathrm{Odd}(\hat{p})$ gives $v \in p'(v) \Leftrightarrow v \in p(v)$ and $v \in \mathrm{Odd}(p'(v)) \Leftrightarrow v \in \mathrm{Odd}(p(v))$.

For conditions \ref{PF7}-\ref{PF9}, the correction amounts to saying that $p'(v)$ is not focussed over $\{v\}$. If this were not the case, then $p'(v) \Delta \hat{p} = p(v)$ would be focussed over $\{v\}$ by Lemma \ref{lemma:AddFocussed}, contradicting the corresponding Pauli flow condition for $p(v)$.

Finally, $p'(v)$ is focussed over $\overline{O} \setminus \{v\}$ using Lemma \ref{lemma:AddFocussed}, so $(p', \prec')$ is focussed.
\end{proof}


\begin{lemma}\label{lemma:AddFlowsGivesFocussed}
Let $\Gamma$ be a labelled open graph with two focussed Pauli flows $(p, \prec)$ and $(p', \prec')$. Then for any vertex $v \in \overline{O}$, $p(v) \Delta p'(v)$ is a focussed set.
\end{lemma}

\begin{proof}
Each of $p(v)$ and $p'(v)$ are focussed over $\overline{O} \setminus \{v\}$, so their combination must also be by Lemma \ref{lemma:AddFocussed}. For each case of $\lambda(v)$, the corresponding condition from \ref{PF4}-\ref{PF9} is then enough to show that $p(v) \Delta p'(v)$ is also focussed over $\{v\}$.
\end{proof}

An important consequence of this result is that in order to fully identify the space of all focussed Pauli flows for a given labelled open graph, it is sufficient to find one using Algorithm \ref{algo:FlowIdentify} that we focus with Lemma \ref{lemma:Focussing}, and find all of the focussed sets. The following proofs show that the focussed sets form a group, allowing us to only need some generating set, and giving an algorithm for finding such generators.

\begin{lemma}
The focussed sets of a labelled open graph form a group under $\Delta$.
\end{lemma}

\begin{proof}
Closure is a direct consequence of Lemma \ref{lemma:AddFocussed} with $S = \overline{O}$. We then extend this to a group with identity $\emptyset$ and each focussed set is self-inverse.
\end{proof}


\begin{lemma}\label{lemma:NumberOfFocussedSets}
Any labelled open graph $\Gamma = (G, I, O, \lambda)$ with a focussed Pauli flow has $2^{|O|-|I|}$ distinct focussed sets.
\end{lemma}

\begin{proof}
We can find the number of subsets of $\overline{I}$ that are focussed over some $S \subseteq \overline{O}$ inductively over the size of $S$ using the following induction hypothesis:

\begin{equation}
\Phi(k) := \forall S \subseteq \overline{O} . |S| = k \Rightarrow \left|\left\{\hat{p} \subseteq \overline{I} \middle| \hat{p} \text{ is focussed over } S \right\}\right| = 2^{|\overline{I}| - |S|}
\end{equation}

For $k = 0$, we only have to consider $S = \emptyset$. Any subset of $|\overline{I}|$ satisfies the required properties, so there are $2^{|\overline{I}|}$ many such sets, so $\Phi(0)$ holds.

Suppose $\Phi(k)$ holds for some $k \geq 0$. Suppose $S \subseteq \overline{O}$ is of size $k+1$ and choose any vertex $v \in S$. Let $S_v := S \setminus v$, so $|S_v| = k$. $\Phi(k)$ implies there are $2^{|\overline{I}|-k}$ subsets of $\overline{I}$ that are focussed over $S_v$. The subsets focussed over $S$ are those that are focussed over both $S_v$ and $\{v\}$.

Let $(p, \prec)$ be some focussed Pauli flow for $\Gamma$, so $p(v)$ is focussed over $\overline{O} \setminus v$, and hence for $S_v$. Lemma \ref{lemma:AddFocussed} implies the sets focussed over $S_v$ are closed under symmetric difference with $p(v)$. However, Lemma \ref{lemma:AddNotFocussed} implies $\hat{p}$ is focussed over $\{p\}$ if and only if $\hat{p} \Delta p(v)$ is not because $p(v)$ corrects the measurement at $v$. This means taking the symmetric difference with $p(v)$ defines a bijection between the sets that are focussed for $S$ and those that are focussed for $S_v$ but not for $\{v\}$. Since these are disjoint and partition the $2^{|\overline{I}|-k}$ that are focussed over $S_v$, we must have $2^{|\overline{I}|-(k+1)}$ that are focussed over $S$, giving $\Phi(k+1)$.

Focussed sets are those that are focussed over the entirety of $\overline{O}$, so $\Phi(|\overline{O}|)$ tells us that there are $2^{|\overline{I}| - |\overline{O}|} = 2^{|O| - |I|}$ focussed sets.
\end{proof}


\begin{lemma}\label{lemma:StabAlgorithmProof}
Given a labelled open graph $\Gamma$ with focussed Pauli flow, there exists an algorithm that identifies $|O|-|I|$ independent generators for the group of focussed sets of $\Gamma$. Furthermore, this algorithm completes in time polynomial in the number of vertices in $\Gamma$.
\end{lemma}

\begin{proof}
Similar to Theorem \ref{thrm:Identification}, we can encode the conditions for focussed sets into a linear equation system $MX = S$ in $\mathbb{F}_2$ which can be solved by Gaussian elimination and back substitution to obtain a single focussed set.

\begin{align}
M &:= \left[ \begin{array}{c} \Gamma \cap \mathbb{P} \times \mathbb{O}^C \\ \hline (\Gamma + \mathrm{Id}) \cap \mathbb{P} \times \mathbb{Y} \end{array} \right] \\
S &:= \left[ \begin{array}{c} 0 \\ \hline 0 \end{array} \right] \\
\mathbb{P} &:= (O \cup \{w \in \overline{O} | \lambda(w) \in \{XY, X, Y\}\}) \cap \overline{I} \\
\mathbb{O} &:= O \cup \{w \in \overline{O} | \lambda(w) \in \{XZ, YZ, Y, Z\}\} \\
\mathbb{Y} &:= \{w \in \overline{O} | \lambda(w) = Y\}
\end{align}

In the above, we let $\Gamma$ stand for the adjacency matrix of its graph. We define $\mathbb{P}$ to be the set of vertices that could be included in a focussed set and $\mathbb{O}$ the set of vertices that could be in the odd neighbourhood. Solutions satisfy \ref{FOC1} since $X$ only ranges over $\mathbb{P}$. The top block of the system encodes \ref{FOC2} since multiplying by $\Gamma$ in $\mathbb{F}_2$ gives the odd neighbourhood, and similarly the bottom block encodes \ref{FOC3}.

We need to add additional conventions to make sure that the focussed sets we obtain are non-empty and independent. Since this system is underconstrained (when there exist non-empty focussed sets), there will always be some freedom of choice during the back substitution step. Since every focussed set must be a solution, these free substitutions must generate the full set. Hence we can obtain independent, non-empty focussed sets by taking a single free substitution for each focussed set, iterating through each free substitution in turn.

Since both $|\mathbb{P}|$ and $|\mathbb{O}^C| + |\mathbb{Y}|$ are at most $|V|$, the Gaussian elimination takes $O(|V|^3)$ time. Each back substitution takes $O(|V|^2)$ time, which we must repeat $|O|-|I|$ times giving $O(|V|^3)$ again.
\end{proof}

It should be noted that similar conditions can be incorporated into Algorithm \ref{algo:FlowIdentify} to directly obtain focussed Pauli flows by imposing further restrictions on the vertices that can be included in the correction sets and adding extra equations to restrict the odd neighbourhood.

A final comment we will make about the focussed conditions is that it removes the need for any ordering relation between Pauli vertices. When the measurement pattern uses only Pauli measurements, all of them can be measured simultaneously, giving a single round of corrections on the outputs. When there are additional planar measurements, we can do all Pauli measurements in a single layer before considering any of the planar measurements.


\begin{lemma}\label{lemma:InitialPaulis}
If a labelled open graph has a Pauli flow, then there exists a Pauli flow $(p, \prec)$ which satisfies $\forall v \in \overline{O} . \lambda(v) \in \{X, Y, Z\} \Rightarrow \forall u \in V . v \preceq u$.
\end{lemma}

\begin{proof}
Given a Pauli flow $(p, \prec)$, we can assume wlog that it is focussed from Lemma \ref{lemma:Focussing}. Let $\prec' := \prec \setminus \{(u, v) \in V \times \overline{O} | \lambda(v) \in \{X, Y, Z\}\}$. By construction, this satisfies $\forall v \in \overline{O} . \lambda(v) \in \{X, Y, Z\} \Rightarrow \forall u \in V . v \preceq' u$. To show that $(p, \prec')$ is a valid Pauli flow, we just need to show that \ref{PF1}-\ref{PF3} are unaffected which all follow from the focussed conditions.
\end{proof}

\section{Proofs for Circuit Extraction}\label{sec:ExtractionProofs}

This section will explore how to identify extraction strings and their properties, building up to a proof of Theorem \ref{thrm:Extraction} for circuit extraction. As dicussed in Section \ref{sec:Extraction}, the key principle is to extract planar measurement angles as rotations over the outputs, to leave behind a stabilizer process. We start by building up an explicit construction for extraction strings from a focussed Pauli flow.


\begin{lemma}\label{lemma:EvenImaginary}
Given a Pauli flow $(p, \prec)$ for a labelled open graph $(G, I, O, \lambda)$, for any vertex $v \in \overline{O}$ the size of $p(v) \cap \mathrm{Odd}(p(v))$ is even.
\end{lemma}

\begin{proof}
Let $G = (V, E)$ and define the subgraph $G' = (p(v), E \cap (p(v) \times p(v)))$. For any vertex $w \in V$, $w \in p(v) \cap \mathrm{Odd}(p(v))$ if and only if $w$ has odd degree in $G'$. Since the sum of the vertex degrees must equal $2$ times the number of edges in $G'$, there must be an even number of vertices with odd degree.
\end{proof}


\begin{lemma}\label{lemma:ExtractionFromFocussedProof}(Restatement of Lemma \ref{lemma:ExtractionFromFocussed})
Let $\Gamma = (G, I, O, \lambda)$ be a labelled open graph with some measurement angles $\alpha : \overline{O} \to \left[0, 2\pi \right)$ and a focussed Pauli flow $(p, \prec)$. Then for any vertex $v \in \overline{O}$, $p(v)$ determines a $P^{\bot v}$-extraction string.
\end{lemma}

\begin{proof}
Consider an arbitrary measured vertex $v \in \overline{O}$.

Since each correction set $p(v)$ is a subset of the non-input vertices, we can combine the graph state stabilizers. We may reorder the $Z$ and $X$ terms with the possible introduction of a $(-1)$.

\begin{equation}
E_G N_{\overline{I}} = (-1)^a \left( \prod_{u \in p(v)} X_u \right) \left( \prod_{u \in \mathrm{Odd}(p(v))} Z_u \right) E_G N_{\overline{I}}
\end{equation}
where $a = |E \cap (p(v) \times p(v))|$ is the number of edges in the subgraph of $p(v)$, since for each edge here we have to reorder the $Z$ and $X$ terms on exactly one of the two vertices.

There are an even number of vertices with both an $X$ and a $Z$ (see Lemma \ref{lemma:EvenImaginary}), so we can apply $Y = iXZ$ on all such instances, again introducing a possible $(-1)$ term.

\begin{equation}
\begin{split}
E_G N_{\overline{I}} = (-1)^b & \left( \prod_{u \in p(v) \setminus \mathrm{Odd}(p(v))} X_u \right) \left( \prod_{u \in \mathrm{Odd}(p(v)) \setminus p(v)} Z_u \right) \\ & \left( \prod_{u \in p(v) \cap \mathrm{Odd}(p(v))} Y_u \right) E_G N_{\overline{I}}
\end{split}
\end{equation}
where $b = a + |p(v) \cap \mathrm{Odd}(p(v))|/2$.

Since $(p, \prec)$ is a Pauli flow, every vertex in $p(v)$ or $\mathrm{Odd}(p(v))$ is either a Pauli measurement or greater than $v$ in $\prec$ (from conditions \ref{PF1} and \ref{PF2}). To fit the form from Equation \ref{eq:ExtractionStringDef}, we consider adding the corresponding projections. Since $(p, \prec)$ is focussed, each vertex (besides outputs and $v$) in $p(v) \setminus \mathrm{Odd}(p(v))$ is projected into an $X$ basis eigenvector, and similarly $\mathrm{Odd}(p(v)) \setminus p(v)$ into $Z$ and $p(v) \cap \mathrm{Odd}(p(v))$ into $Y$. This means we can absorb these Pauli operators into the projections, again with the possible introduction of a $(-1)$.

\begin{equation}\label{eq:FlowExtractionString}
\begin{split}
\left( \prod_{\substack{u \succ v \\ \lambda(u) \notin \{X, Y, Z\}}} \bra{+_{\lambda(u), 0}}_u \right) \left( \prod_{\substack{u \in \overline{O} \setminus \{v\} \\ \lambda(u) \in \{X, Y, Z\}}} \bra{+_{\lambda(u), \alpha(u)}}_u \right) E_G N_{\overline{I}} \\
= (-1)^c \left( \prod_{\substack{u \succ v \\ \lambda(u) \notin \{X, Y, Z\}}} \bra{+_{\lambda(u), 0}}_u \right) \left( \prod_{\substack{u \in \overline{O} \setminus \{v\} \\ \lambda(u) \in \{X, Y, Z\}}} \bra{+_{\lambda(u), \alpha(u)}}_u \right) \\
\left( \prod_{u \in (p(v) \setminus \mathrm{Odd}(p(v))) \cap (O \cup \{v\})} X_u \right)
\left( \prod_{u \in (\mathrm{Odd}(p(v)) \setminus p(v)) \cap (O \cup \{v\})} Z_u \right) \\
\left( \prod_{u \in p(v) \cap \mathrm{Odd}(p(v)) \cap (O \cup \{v\})} Y_u \right) E_G N_{\overline{I}}
\end{split}
\end{equation}
where $c = b + \left| \left( p(v) \cup \mathrm{Odd}(p(v)) \right) \cap \{u \in \overline{O} | \lambda(u) \in \{X, Y, Z\} \wedge \alpha(u) = \pi\} \right|$.

This is now a Pauli operator with $P^{\bot v}$ on $v$ and remainder over only the outputs, making it a valid primary extraction string.
\end{proof}

Recall that the \ref{lemma:GadgetStabCorrespondence} can use these extraction strings to remove the measurement angles. We can summarise Equation \ref{eq:RotateBasis} by
\begin{align}
\bra{+_{\lambda(v), \alpha(v)}} &\approx \bra{+_{\lambda(v), 0}} e^{(-1)^{D_v} i \tfrac{\alpha(v)}{2} P^{\bot v}} \\
D_v &:= \begin{cases}
1 & \text{if } \lambda(v) = YZ \\
0 & \text{otherwise}
\end{cases}
\end{align}
where $D_v$ dictates the direction of rotation about the Bloch sphere. This means the rotation we extract over the outputs is $e^{(-1)^{D_v} i \tfrac{\alpha(v)}{2} \mathbf{P}^{\bot v}}$.

After extracting all planar measurement angles, we are left with the following stabilizer process:
\begin{equation}
\mathcal{C} = \left( \prod_{\substack{u \in \overline{O} \\ \lambda(u) \notin \{X, Y, Z\}}} \bra{+_{\lambda(u), 0}}_u \right) \left( \prod_{\substack{u \in \overline{O} \\ \lambda(u) \in \{X, Y, Z\}}} \bra{+_{\lambda(u), \alpha(u)}}_u \right) E_G N_{\overline{I}}
\end{equation}

We can characterise this by finding its isometry tableau. Recall from Section \ref{sec:PDDAG} that the isometry tableau generalises variants of the stabilizer tableau representation from previous literature for both stabilizer states \cite{VandenNest2004}, where the rows represent generators for the stabilizer group, and unitaries \cite{Aaronson2008}, where the rows represent the actions of the unitary on $Z_u$ and $X_u$ for each input $u$ - that is, the operator $\mathbf{P}$ such that $\mathbf{P} \mathcal{C} Z_u = \mathcal{C}$. It is enough for us to find such $Z_u$ and $X_u$ actions for each input along with $|O| - |I|$ generators for the free stabilizers over the outputs.

Since $Z$ on an input will commute through the \CZ~gates, we can get the corresponding stablizer row from the primary extraction string assuming $P^{\bot u} = Z$ for this input. This is guaranteed since the correction set cannot contain any inputs, including itself.


We can use the same trick to obtain the $X$ rows from primary extraction strings. In this case, we find an altered pattern $(\Gamma', \alpha')$ such that $\mathcal{C} \approx \mathcal{C}' H_u$. The Hadamard maps $X$ to a $Z$, meaning we can use the extraction string with respect to $(\Gamma', \alpha')$. One way to achieve this is by input extension.

\begin{definition}[Input Extension]
Given an open graph $(G, I, O)$, \textit{input extension} about $u \in I$ adds a new vertex $u'$ that is only connected to $u$ and replaces $u$ in the input set. Formally, $V' = V \cup \{u'\}$, $E' = E \cup \{u' \sim u\}$, $I' = (I \cup \{u'\}) \setminus \{u\}$.
\end{definition}

\begin{lemma}\label{lemma:InputExtensionSemantics}
Let $(\Gamma, \alpha)$ describe a measurement pattern with some chosen input $u \in I$. Taking an input extension about $u$ gives an equivalent measurement pattern $(\Gamma', \alpha')$ assuming $\lambda' = \lambda \cup \{u' \mapsto XY\}$, $\alpha' = \alpha \cup \{u' \mapsto 0\}$, and a Hamadard gate is applied to input $u'$ before the start of the pattern.
\end{lemma}

\begin{proof}
It is enough to verify the following equation:
\begin{equation}
H = \sqrt{2} \left( \bra{+_{XY, 0}} \otimes \mathbb{I} \right) CZ \left( \mathbb{I} \otimes \ket{+_{X, 0}} \right)
\end{equation}
\end{proof}

In order to actually use this, we need to show that Pauli flows are preserved by input extensions.

\begin{lemma}\label{lemma:InputExtension}(Generalisation of \cite[Lemma 3.8]{Backens2020})
Let $\Gamma$ and $\Gamma'$ be labelled open graphs related by an input extension on vertex $u \in I$ (generating $u' \in I'$). If $\Gamma$ has a Pauli flow, then so does $\Gamma'$.
\end{lemma}

\begin{proof}
Suppose $\Gamma$ has a Pauli flow $(p, \prec)$. Let $p' = p \cup \{u' \mapsto \{u\} \}$ and $\prec'$ be the transitive closure of $\prec \cup \{(u', w) | w \in N_G(u) \cup \{u\}\}$.

For any $v \in V \setminus O$, $u' \notin p'(v) = p(v)$ and $u' \notin \mathrm{Odd}_{G'}(p'(v)) = \mathrm{Odd}_{G}(p(v))$ since its only neighbour $u$ is an input in $\Gamma$ which could not appear in any correction sets from $p$. This allows us to inherit all of the Pauli flow properties from $(p, \prec)$ for $v$ as they have remained unchanged.

For $u'$, the definition of $\prec'$ guarantees conditions \ref{PF1}-\ref{PF3} since $u' \prec' w$ for every $w \in N_G(u) \cup \{u\} = (\mathrm{Odd}_{G'}(p'(u')) \setminus \{u'\}) \cup p'(u')$. From the construction of $p'(u')$, we satisfy condition \ref{PF4}, and the remainder hold trivially.
\end{proof}

The remaining stabilizer generators can be obtained from the focussed sets of the measurement pattern.


\begin{lemma}\label{lemma:StabsFromFocussedSets}
Given a measurement pattern $(\Gamma, \alpha)$ where the labelled open graph $\Gamma = (G, I, O, \lambda)$ has a focussed set $\hat{p}$, then $\hat{p}$ determines a stabilizer $\mathbf{\hat{P}}$ of the linear map:
\begin{equation}
\left( \prod_{\substack{w \in \overline{O} \\ \lambda(w) \notin \{X, Y, Z\}}} \bra{+_{\lambda(w), 0}}_w \right) \left( \prod_{\substack{w \in \overline{O} \\ \lambda(w) \in \{X, Y, Z\}}} \bra{+_{\lambda(w), \alpha(w)}}_w \right) E_G N_{\overline{I}}
\end{equation}
\end{lemma}

\begin{proof}
This follows similarly to Lemma \ref{lemma:ExtractionFromFocussed}. After applying the measurement projections, the only Pauli terms that are not absorbed into the projections are over the output qubits.
\end{proof}

Lemma \ref{lemma:NumberOfFocussedSets} implies that we can get $|O|-|I|$ generators for the group of focussed sets, and we are looking for $|O|-|I|$ generators for the free stabilizers. For this to work, we need the map from Lemma \ref{lemma:StabsFromFocussedSets} to be a bijection (i.e. the generators for the focussed sets don't degenerate when mapped into stabilizers). We go one step further by showing that it also preserves the group action as a result of Lemmas \ref{lemma:MultiplyStrings} and \ref{lemma:EqualStrings}.


\begin{lemma}\label{lemma:MultiplyStrings}
Let $\Gamma$ be a labelled open graph. For any set of vertices $S \subseteq \overline{I}$, denote the output substring of the stabilizer as $\mathbf{P}^S = \left( \prod_{w \in O \cap S} X_w \right) \left( \prod_{w \in O \cap \mathrm{Odd}(S)} Z_w \right)$. Then for any $S, T \subseteq \overline{I}$ we have $\mathbf{P}^S \mathbf{P}^T \approx \mathbf{P}^{S \Delta T}$.
\end{lemma}

\begin{proof}
This follows from the anticommutativity of $X$ and $Z$ and $\mathrm{Odd}(S \Delta T) = \mathrm{Odd}(S) \Delta \mathrm{Odd}(T)$.

\pagebreak
\begin{equation}
\begin{split}
\mathbf{P}^S \mathbf{P}^T &= \left( \prod_{w \in O \cap S} X_w \right) \left( \prod_{w \in O \cap \mathrm{Odd}(S)} Z_w \right) \left( \prod_{w \in O \cap T} X_w \right) \left( \prod_{w \in O \cap \mathrm{Odd}(T)} Z_w \right) \\
&\approx \left( \prod_{w \in O \cap S} X_w \right) \left( \prod_{w \in O \cap T} X_w \right) \left( \prod_{w \in O \cap \mathrm{Odd}(S)} Z_w \right) \left( \prod_{w \in O \cap \mathrm{Odd}(T)} Z_w \right) \\
&= \left( \prod_{w \in O \cap (S \Delta T)} X_w \right) \left( \prod_{w \in O \cap \mathrm{Odd}(S \Delta T)} Z_w \right) \\
&= \mathbf{P}^{S \Delta T}
\end{split}
\end{equation}
\end{proof}


\begin{lemma}\label{lemma:EqualStrings}
For any non-zero linear map $\mathcal{C}$ with stabilizers $A$ and $B$, if $A$ and $B$ share the same Pauli string then $A = B$.
\end{lemma}

\begin{proof}
If $A$ and $B$ share the same Pauli string, they can only differ by phase, so $B = e^{i\theta} A$. Since stabilizers form a group, $AB = A(e^{i\theta} A) = e^{i\theta}I$ is also a stabilizer, i.e. $\mathcal{C} = e^{i\theta} \mathcal{C}$. $e^{i\theta} \neq 1$ implies $\mathcal{C}$ is a zero map, so we must have $A = B$.
\end{proof}

Because the groups of focussed sets and free stabilizers have the same finite cardinality, injectivity is enough to show that we have an isomorphism. We can prove this by showing that only the empty set gets mapped to the identity.


\begin{lemma}\label{lemma:AntiCommutingStringsProof}
Let $(\Gamma, \alpha)$ describe a measurement pattern with a focussed Pauli flow $(p, \prec)$. Then for any vertices $u, v \in \overline{O}$ the primary extraction strings satisfy
\begin{equation}
\mathbf{P}^{\bot u} \mathbf{P}^{\bot v} = (-1)^{F_{u \to v}^p + F_{v \to u}^p} \mathbf{P}^{\bot v} \mathbf{P}^{\bot u}
\end{equation}
where $F_{x \to y}^p := |\{y\} \cap (p(x) \cup \mathrm{Odd}(p(x)))|$ indicates whether $y$ is used in the correction of $x$. Similarly, given a focussed set $\hat{p}$, its stabilizer satisfies
\begin{equation}
\mathbf{P}^{\bot u} \mathbf{\hat{P}} = (-1)^{G_{\hat{p} \to u}} \mathbf{\hat{P}} \mathbf{P}^{\bot u}
\end{equation}
where $G_{\hat{p} \to u} = |\{u\} \cap (\hat{p} \cup \mathrm{Odd}(\hat{p}))|$ indicates whether $u$ is used in the generation of the stabilizer.
\end{lemma}

\begin{proof}
Since $\mathbf{P}^{\bot u}$ and $\mathbf{P}^{\bot v}$ are tensor products of Pauli matrices, they must either commute or anticommute. Consider the linear map $\mathcal{C}$ given by:
\begin{equation}
\mathcal{C} := \left( \prod_{\substack{w \in \overline{O} \setminus \{u, v\} \\ \lambda(w) \notin \{X, Y, Z\} \\ u \prec w \vee v \prec w}} \bra{+_{\lambda(w), 0}}_w \right) \left( \prod_{\substack{w \in \overline{O} \setminus \{u, v\} \\ \lambda(w) \in \{X, Y, Z\}}} \bra{+_{\lambda(w), \alpha(w)}}_w \right) E_G N_{\overline{I}}
\end{equation}

This may not be in the exact form needed to introduce extraction strings for $u$ and $v$ if they use each other for corrections. Let $P^{x \to y}$ be the Pauli induced on $y$ when correcting $x$.

\pagebreak
\begin{equation}
P^{x \to y} := \begin{cases}
I & \text{if } y \notin p(x) \cup \mathrm{Odd}(p(x)) \\
X & \text{if } y \in p(x) \setminus \mathrm{Odd}(p(x)) \\
Y & \text{if } y \in p(x) \cap \mathrm{Odd}(p(x)) \\
Z & \text{if } y \in \mathrm{Odd}(p(x)) \setminus p(x)
\end{cases}
\end{equation}

We can still follow the construction in the proof of Lemma \ref{lemma:ExtractionFromFocussed} to deduce that $(-1)^{A_{u \to v}^p} P^{\bot u}_u P^{u \to v}_v \mathbf{P}^{\bot u}$ and $(-1)^{A_{v \to u}^p} P^{v \to u}_u P^{\bot v}_v \mathbf{P}^{\bot v}$ are stabilizers of $\mathcal{C}$ where

\begin{equation}
A_{x \to y}^p := \begin{cases}
F_{x \to y}^p & \text{if } \lambda(y) \in \{X, Y, Z\} \wedge \alpha(y) = \pi \\
0 & \text{otherwise}
\end{cases}
\end{equation}

Comparing conditions \ref{PF4}-\ref{PF9} and the focussed conditions, we see that, for any vertex $x \in \overline{O}$, $P^{\bot x}$ is orthogonal to the measurement plane/Pauli of $x$ and hence anticommutes with any $P^{y \to x} \neq I$. Equationally, $P^{\bot x} P^{y \to x} = (-1)^{F_{y \to x}^p} P^{y \to x} P^{\bot x}$.

The stabilizer group of $\mathcal{C}$ is abelian, so we have

\begin{equation}
\begin{split}
& (-1)^{A_{u \to v}^p} P^{\bot u}_u P^{u \to v}_v \mathbf{P}^{\bot u} (-1)^{A_{v \to u}^p} P^{v \to u}_u P^{\bot v}_v \mathbf{P}^{\bot v} \\
=& (-1)^{A_{v \to u}^p} P^{v \to u}_u P^{\bot v}_v \mathbf{P}^{\bot v} (-1)^{A_{u \to v}^p} P^{\bot u}_u P^{u \to v}_v \mathbf{P}^{\bot u} \\
=& (-1)^{F_{u \to v}^p + F_{v \to u}^p} (-1)^{A_{u \to v}^p} P^{\bot u}_u P^{u \to v}_v \mathbf{P}^{\bot v} (-1)^{A_{v \to u}^p} P^{v \to u}_u P^{\bot v}_v \mathbf{P}^{\bot u}
\end{split}
\end{equation}
i.e. $\mathbf{P}^{\bot u} \mathbf{P}^{\bot v} = (-1)^{F_{u \to v}^p + F_{v \to u}^p} \mathbf{P}^{\bot v} \mathbf{P}^{\bot u}$.

The proof for focussed sets is virtually the same, though we only need care about anticommuting Paulis on $u$.
\end{proof}


\begin{lemma}\label{lemma:EmptyFocussedSets}
Given a labelled open graph $\Gamma = (G, I, O, \lambda)$ with a focussed Pauli flow and any focussed set $\hat{p}$, then $(\hat{p} \cup \mathrm{Odd}(\hat{p})) \cap O = \emptyset \Leftrightarrow \hat{p} = \emptyset$.
\end{lemma}

\begin{proof}
$\hat{p} = \emptyset \Rightarrow (\hat{p} \cup \mathrm{Odd}(\hat{p})) \cap O = \emptyset$ is trivial. For the other direction, suppose for a contradiction that we have a non-empty $\hat{p}$ for which $(\hat{p} \cup \mathrm{Odd}(\hat{p})) \cap O = \emptyset$, so we have some $v \in \hat{p} \cap \overline{O}$ and the corresponding stabilizer is the identity. However, Lemma \ref{lemma:AntiCommutingStringsProof} implies that the stabilizer must anticommute with the primary extraction string $\mathbf{P}^{\bot v}$, contradicting the identity assumption.
\end{proof}


\begin{theorem}\label{thrm:StabBijection}
Let $(\Gamma, \alpha)$ describe a measurement pattern with a focussed Pauli flow. Then the construction from Lemma \ref{lemma:StabsFromFocussedSets} gives an isomorphism between the groups of focussed sets and the stabilizers of the linear map for $(\Gamma, \alpha)$ after removing all planar measurement angles.
\end{theorem}

\begin{proof}
The group action is preserved due to Lemmas \ref{lemma:MultiplyStrings} and \ref{lemma:EqualStrings}. By Lemma \ref{lemma:EmptyFocussedSets}, the map is injective (if we had two distinct focussed sets with the same stabilizer, combining them with $\Delta$ would give a non-empty focussed set that maps to the identity stabilizer) and subsequently bijective by cardinality from Lemma \ref{lemma:NumberOfFocussedSets}.
\end{proof}

This isomorphism guarantees that the focussed sets identified by Lemma \ref{lemma:StabAlgorithmProof} actually give a generating set for the stabilizer group, giving a valid isometry tableau. All that remains is to put everything together into a formal procedure.


\begin{theorem}\label{thrm:ExtractionProof}(Restatement of Theorem \ref{thrm:Extraction})
Let $(\Gamma, \alpha)$ describe a measurement pattern where $\Gamma$ has a Pauli flow. Then there is an algorithm that identifies an equivalent circuit requiring no ancillae which completes in time polynomial in the number of vertices in $\Gamma$.
\end{theorem}

\begin{proof}
Using Theorem \ref{thrm:Identification} and Lemma \ref{lemma:Focussing}, we can identify a maximally delayed, focussed Pauli flow $(p, \prec)$. Let $J : \mathbb{Z}_{|\overline{O}|} \to \overline{O}$ be an indexing over the measured qubits such that $J(i) \prec J(j) \Rightarrow j < i$, i.e. $J$ iterates from outputs to inputs respecting $\prec$. We will define a sequence of patterns $(\Gamma, \alpha_k)$ and circuits $C_k$ such that $\llbracket C_k \rrbracket \cdot \llbracket \Gamma; \alpha_k \rrbracket \approx \llbracket \Gamma; \alpha \rrbracket$ (where $\llbracket - \rrbracket$ denotes the interpretation as a linear function). We start with $\alpha_0 = \alpha$ and $C_0$ is the empty (identity) circuit.

For any $k \geq 0$, the update from $\alpha_k$ to $\alpha_{k+1}$ is just setting the angle of $J(k)$ to $0$ if it is planar.

\begin{equation}
\alpha_{k}(J(j)) := \begin{cases}
0 & \text{if } j < k \wedge \lambda(J(i)) \in \{XY, XZ, YZ\} \\
\alpha(J(j)) & \text{otherwise}
\end{cases}
\end{equation}

Defining the update from $C_k$ to $C_{k+1}$ requires us to find a Pauli string $\mathbf{P}$ over the outputs such that:

\begin{equation}\label{eq:MBQCextract}
\begin{split}
\left( \prod_{j < k \wedge \lambda(J(j)) \notin \{X, Y, Z\}} \bra{+_{\lambda(J(j)), 0}}_{J(j)} \right) & \\ \left( \prod_{j \geq k \vee \lambda(J(j)) \in \{X, Y, Z\}} \bra{+_{\lambda(J(j)), \alpha(J(j))}}_{J(j)} \right) & E_G N_{\overline{I}} \\
\approx e^{i\tfrac{\alpha(J(k))}{2}\mathbf{P}} \left( \prod_{j \leq k \wedge \lambda(J(j)) \notin \{X, Y, Z\}} \bra{+_{\lambda(J(j)), 0}}_{J(j)} \right) & \\ \left( \prod_{j > k \vee \lambda(J(j)) \in \{X, Y, Z\}} \bra{+_{\lambda(J(j)), \alpha(J(j))}}_{J(j)} \right) & E_G N_{\overline{I}}
\end{split}
\end{equation}

When $\lambda(J(k)) \in \{X, Y, Z\}$ the projections on each side of Equation \ref{eq:MBQCextract} are identical, meaning it suffices to let $\mathbf{P}$ be the identity and set $C_{k+1} = C_k$.

When $\lambda(J(k)) \in \{XY, XZ, YZ\}$, the projectors differ by a $e^{(-1)^{D_{J(k)}} i \tfrac{\alpha(J(k))}{2}P^{\bot J(k)}_{J(k)}}$ according to Equation \ref{eq:RotateBasis}. Equation \ref{eq:MBQCextract} then holds by using the \ref{lemma:GadgetStabCorrespondence} with a primary extraction string for $J(k)$, which we can obtain by Lemma \ref{lemma:ExtractionFromFocussed}. We decompose this rotation into basic gates (for example, by \CX~ladder constructions \cite{Barkoutsos2018, Cowtan2020}) and add them to the start of $C_k$ so that $\llbracket C_{k+1} \rrbracket = \llbracket C_k \rrbracket \cdot e^{(-1)^{D_{J(k)}} i \tfrac{\alpha(J(k))}{2} \mathbf{P}^{\bot J(k)}}$.

At the end of this sequence, we have the pattern $(\Gamma, \alpha_{|\overline{O}|})$ where every measurement is in a Pauli basis and the associated circuit $C_{|\overline{O}|}$. This means the corresponding isometry for this pattern is a stabilizer process.

To characterise this process, it is sufficient to identify its isometry tableau. For each input, we can find the $Z$ row from the primary extraction string of the input vertex, and by taking an input extension with Lemma \ref{lemma:InputExtension} we can similarly get the $X$ row. Each of our $|O|-|I|$ generators for the free stabilizers is given by Lemma \ref{lemma:StabAlgorithmProof}.

We can efficiently synthesise the isometry tableau and qubit preparations into a circuit $C_{\mathrm{stab}}$. The final circuit is the combination of $C_{\mathrm{stab}}$ followed by $C_{|\overline{O}|}$ which implements the same operator as $(\Gamma, \alpha)$.

To examine the complexity, we suppose an explicit Pauli flow is not given up front. Theorem \ref{thrm:Identification} says that we can identify it in $O(|V|^5)$ time. We can extend the flow to consider all of the input extensions simultaneously using Lemma \ref{lemma:InputExtension} in $O(|I|)$ time (ignoring the update to $\prec$ since it won't affect the outcome of the extraction procedure). The construction of Lemma \ref{lemma:Focussing} focusses this Pauli flow in $O(|V|^3)$ time (we apply at most $O(|V|^2)$ updates, each of which takes $O(|V|)$ time to identify and apply). Adding the extra focussed sets also takes $O(|V|^3)$ time by Lemma \ref{lemma:StabAlgorithmProof}. Synthesising this into a circuit using the naive \CX~ladder construction takes $O(|V| \cdot |O|)$ time for the rotations from non-Pauli measurements, and $O(|O|^3)$ to synthesise a stabilizer tableau using Aaronson and Gottesman's method \cite{Aaronson2008} (mapping an isometry tableau into that form just requires identifying the additional generators to fill the Pauli group, which can be done by Gaussian elimination in $O(|O|^3)$ time). The overall complexity is therefore dominated by the $O(|V|^5)$ Pauli flow identification.
\end{proof}

\begin{example}\label{ex:ExtractionExample}
Suppose we start with the following measurement pattern and focussed Pauli flow.

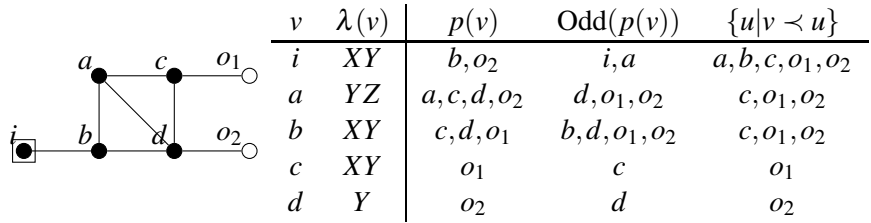
\begin{center}
\input{tikz_figs/mbqc_rewrite_initial.tikz}
\begin{tabular}{cc|ccc}
$v$ & $\lambda(v)$ & $p(v)$ & $\mathrm{Odd}(p(v))$ & $\{u | v \prec u\}$ \\
\hline
$i$ & $XY$ & $b, o_2$ & $i, a$ & $a, b, c, o_1, o_2$ \\
$a$ & $YZ$ & $a, c, d, o_2$ & $d, o_1, o_2$ & $c, o_1, o_2$ \\
$b$ & $XY$ & $c, d, o_1$ & $b, d, o_1, o_2$ & $c, o_1, o_2$ \\
$c$ & $XY$ & $o_1$ & $c$ & $o_1$ \\
$d$ & $Y$ & $o_2$ & $d$ & $o_2$ \\
\end{tabular}
\end{center}

Because we have two outputs and only one input, there is also an extra focussed set $\hat{p} = \{c, o_2\}$. We start by constructing the primary extraction strings for each vertex from Lemma \ref{lemma:ExtractionFromFocussed}. Let $a_d \in \{0, 1\}$ be such that $\alpha(d) = a_d \pi$. Recall that we get one phase flip per edge between adjacent vertices in the correction/focussed set ($|E \cap (p(v) \times p(v))|$), one phase flip for every two $Y$s that appear in the graph state stabilizer ($|p(v) \cap \mathrm{Odd}(p(v))| / 2$), and one for each term in the graph state stabilizer that is absorbed from a Pauli measurement with angle $\pi$ ($|(p(v) \cup \mathrm{Odd}(p(v))) \cap \{w | \lambda(w) \in \{X, Y, Z\} \wedge \alpha(w) = \pi\}|$).

\begin{center}
\begin{tabular}{c|cc|ccc|c}
$S$ & $S \cap O$ & $\mathrm{Odd}(S) \cap O$ & Edges & $Y$s & Paulis & $\mathbf{P}^{S}$ \\
\hline
$p(i)$ & $o_2$ & $ $ & $0$ & $0$ & $ $ & $I_1X_2$ \\
$p(a)$ & $o_2$ & $o_1, o_2$ & $4$ & $2$ & $d$ & $(-1)^{a_d+1}Z_1 Y_2$ \\
$p(b)$ & $o_1$ & $o_1, o_2$ & $2$ & $2$ & $d$ & $(-1)^{a_d+1}Y_1 Z_2$ \\
$p(c)$ & $o_1$ & $ $ & $0$ & $0$ & $ $ & $X_1I_2$ \\
$p(d)$ & $o_2$ & $ $ & $0$ & $0$ & $ $ & $I_1X_2$ \\
\hline
$\hat{p}$ & $o_2$ & $o_1$ & $0$ & $0$ & $ $ & $Z_1X_2$ \\
\end{tabular}
\end{center}

We start by extracting the planar angles as rotations in $\succ$-order. We must start with $c$, extracing it as $e^{i \tfrac{\alpha(c)}{2} X_1 I_2}$ over the outputs.

Having extracted this, the projection on $c$ becomes the positive $X$ basis vector so it can absorb the corresponding Paulis on the graph state stabilizer from $p(a)$ and $p(b)$. Either $a$ or $b$ can be extracted next, giving $e^{(-1)^{a_d} \tfrac{\alpha(a)}{2} Z_1Y_2}$ and $e^{(-1)^{a_d+1} \tfrac{\alpha(b)}{2} Y_1Z_2}$ respectively. Note that the phase of the rotation from $a$ got an extra $-1$ since $D_a = 1$ ($\lambda(a) = YZ$).

The final rotation to be extracted is $e^{i \tfrac{\alpha(i)}{2} I_1X_2}$ from $i$. We don't extract any rotation from $d$ since the measurement is already in a Pauli basis, and hence we can treat it as part of the leftover stabilizer process.

We can get the isometry tableau for this from the primary extraction strings. Having both the input and its only neighbour labelled with $XY$ means we don't necessarily have to perform input extensions to obtain the rows of the tableau during extraction, since they will function identically to $i$ and $b$ here because $\Gamma$ can already be viewed as an input extension of $\Gamma \setminus \{i\}$ with input $b$. This process maps $Z$ on the input to $I_1X_2$ ($\mathbf{P}^{\bot i}$) over the outputs and $X$ on the input to $(-1)^{a_d+1}Y_1Z_2$ ($\mathbf{P}^{\bot b}$). We have an extra row from the free stabilizer $Z_1X_2$ obtained from $\hat{p}$.

The structure we now have exactly maps to the following PDDAG:

\begin{center}
\addtolength{\tabcolsep}{-3pt}
\begin{tabular}{c|cc|c}
Ins & \multicolumn{2}{c|}{Outs} & Sign \\
\hline
$X$ & $Y$ & $Z$ & $(-)^{a_d+1}$ \\
\hline
$Z$ & & $X$ & $+$ \\
\hline
& $Z$ & $X$ & $+$ \\
\end{tabular}
\addtolength{\tabcolsep}{3pt}
\begin{tikzcd}[ampersand replacement=\&,row sep=0.cm,column sep=0.5cm]
(I_1 X_2, \alpha(i)) \arrow{r} \arrow{rd} \& ((-1)^{a_d} Z_1 Y_2, \alpha(a)) \arrow{r} \& (X_1 I_2, \alpha(c)) \\
\& ((-1)^{a_d+1}Y_1 Z_2, \alpha(b)) \arrow{ru} \&
\end{tikzcd}
\end{center}

Note that producing this PDDAG just amounted to reading off the focussed Pauli flow, using the correction sets to produce the strings and the order $\prec$ for the dependency relation between the rotations.

For the rotations from $a$ and $b$, we put the $(-1)$ phase with the strings rather than the angle. Whilst nodes $(-\mathbf{P}, \theta)$ and $(\mathbf{P}, -\theta)$ represent the same rotations, it is useful to keep the convention of writing the extracted rotation as $((-1)^{D_v} \mathbf{P}^{\bot v}, \alpha(v))$ to simplify the simulations in Section \ref{sec:Rewrites}.

We can then synthesise a circuit from the PDDAG component-wise. We start with the isometry tableau, for which an example circuit is given below.

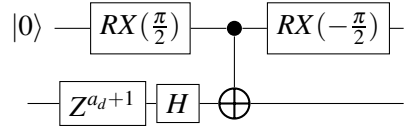
\begin{center}
\input{tikz_figs/circ_extracted_tableau.tikz}
\end{center}

For the rotations, the $(I_1X_2, \alpha(i))$ and $(X_1I_2, \alpha(c))$ are simply \RX~gates. The other two rotations commute, and so can be efficiently diagonalised simultaneously.

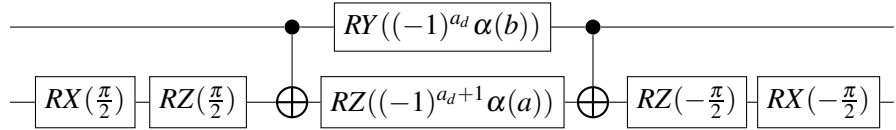
\begin{center}
\input{tikz_figs/circ_extracted_rotations.tikz}
\end{center}

Note that the angles of rotations acquire an extra $-1$ phase flip, since we defined $(\mathbf{P}, \theta)$ to map to $e^{i \tfrac{\theta}{2} \mathbf{P}}$, but basic rotation gates are defined as $RP(\theta) = e^{-i \tfrac{\theta}{2} P}$.

Putting it all together, we obtain the following circuit:

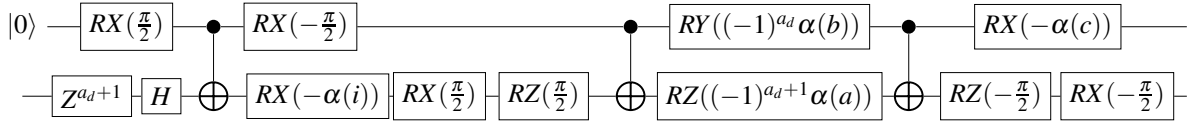
\begin{center}
\resizebox{\linewidth}{!}{\input{tikz_figs/circ_extracted.tikz}}
\end{center}
\end{example}

\section{Proofs for Rewrite Correspondences}\label{sec:RewritesProofs}

Throughout this section, we will follow a general pattern when deriving the correspondences between rewrites of measurement patterns and PDDAGs:

\begin{enumerate}
\item Prove that the existence of Pauli flows is preserved by giving explicit updates to the focussed Pauli flows (and focussed sets);
\item Give the explicit updates to the primary extraction strings (and free stabilizers) for each rewrite;
\item Justify that the same updates are performed by the PDDAG rewrite.
\end{enumerate}

We will also add additional proofs that semantics are preserved by rewriting the measurement pattern as required. When deriving equalities in the following proofs, we will freely use $\approx$ to ignore global constants that are later removed symmetrically using the same substitution.

\subsection{Relabelling Pauli Measurements}

\begin{lemma}\label{lemma:FixPauliSemantics}
Let $(\Gamma, \alpha)$ describe a measurement pattern and $u \in \overline{O}$ be a vertex with $\lambda(u) \in \{XY, XZ, YZ\}$ and $\alpha(u) \in \{0, \tfrac{\pi}{2}, \pi, \tfrac{3\pi}{2}\}$. Then there exists an equivalent pattern $(\Gamma', \alpha')$ from relabelling $u$ as a Pauli measurement, where the new labels and angles are given by
\begin{equation}\label{eq:PauliPlanarUpdate}
(\lambda'(u), \alpha'(u)) = \begin{cases}
(X, \alpha(u)) & \text{if } \lambda(u) = XY \wedge \alpha(u) \in \{0, \pi\} \\
(Y, \alpha(u) - \tfrac{\pi}{2}) & \text{if } \lambda(u) = XY \wedge \alpha(u) \in \{\tfrac{\pi}{2}, \tfrac{3\pi}{2}\} \\
(Z, \alpha(u)) & \text{if } \lambda(u) = XZ \wedge \alpha(u) \in \{0, \pi\} \\
(X, \alpha(u) - \tfrac{\pi}{2}) & \text{if } \lambda(u) = XZ \wedge \alpha(u) \in \{\tfrac{\pi}{2}, \tfrac{3\pi}{2}\} \\
(Z, \alpha(u)) & \text{if } \lambda(u) = YZ \wedge \alpha(u) \in \{0, \pi\} \\
(Y, \alpha(u) - \tfrac{\pi}{2}) & \text{if } \lambda(u) = YZ \wedge \alpha(u) \in \{\tfrac{\pi}{2}, \tfrac{3\pi}{2}\} \\
\end{cases}
\end{equation}
and $(\lambda'(v), \alpha'(v)) = (\lambda(v), \alpha(v))$ for any $v \in \overline{O} \setminus \{u\}$.
\end{lemma}

\begin{proof}
The following equations describe how the Pauli measurements are embedded into the planes:

\begin{equation}\label{eq:PauliPlanarMap}
\begin{split}
\bra{+_{X, \alpha}} &\approx \bra{+_{XY, \alpha}} \approx \bra{+_{XZ, \alpha + \tfrac{\pi}{2}}} \\
\bra{+_{Y, \alpha}} &\approx \bra{+_{XY, \alpha + \tfrac{\pi}{2}}} \approx \bra{+_{YZ, \alpha + \tfrac{\pi}{2}}} \\
\bra{+_{Z, \alpha}} &\approx \bra{+_{XZ, \alpha}} \approx \bra{+_{YZ, \alpha}} \\
\end{split}
\end{equation}

The equivalent Pauli measurement label and angle are determined uniquely by these embeddings.
\end{proof}


\begin{lemma}\label{lemma:FixPauliProof}
Let $(\Gamma, \alpha)$ and $(\Gamma', \alpha')$ describe measurement patterns related by relabelling some planar $u \in \overline{O}$ as a Pauli according to Lemma \ref{lemma:FixPauliSemantics}. If $\Gamma$ has a focussed Pauli flow $(p, \prec)$, then $(p', \prec)$ is a focussed Pauli flow for $\Gamma'$ where for any $v \in \overline{O}$
\begin{equation}
p'(v) := \begin{cases}
p(v) \Delta p(u) & \text{if } v \neq u \wedge u \in p(v) \cup \mathrm{Odd}(p(v)) \wedge \alpha(u) \in \{\tfrac{\pi}{2}, \tfrac{3\pi}{2}\} \\
p(v) & \text{otherwise}
\end{cases}
\end{equation}

Furthermore, focussed sets $\hat{p}$ can be updated as
\begin{equation}
\hat{p'} := \begin{cases}
\hat{p} \Delta p(u) & \text{if } u \in \hat{p} \cup \mathrm{Odd}(\hat{p}) \wedge \alpha(u) \in \{\tfrac{\pi}{2}, \tfrac{3\pi}{2}\} \\
\hat{p} & \text{otherwise}
\end{cases}
\end{equation}
\end{lemma}

\begin{proof}
Suppose that $(p, \prec)$ is a focussed Pauli flow for $\Gamma$. Since $\lambda(u) \in \{XY, XZ, YZ\}$, any $v \in \overline{O} \setminus \{u\}$ such that $u \in p(v) \cup \mathrm{Odd}(p(v))$ must satisfy $v \prec u$ by conditions \ref{PF1} and \ref{PF2}. $(p', \prec)$ is then also a Pauli flow for $\Gamma$ using Lemma \ref{lemma:AddFlows}. Any Pauli flow for $\Gamma$ is also a Pauli flow for $\Gamma'$, since \ref{PF1}-\ref{PF3} are strictly weakened and \ref{PF7}-\ref{PF9} for $\Gamma'$ follow from \ref{PF4}-\ref{PF6} for $\Gamma$.

To show that $(p', \prec)$ is focussed for $\Gamma'$, consider and vertex $v \in \overline{O}$. As $\lambda$ and $\lambda'$ only differ on $u$, when $p'(v) = p(v)$ it remains focussed over $\overline{O} \setminus \{u, v\}$ in $\Gamma'$. This is enough when $v = u$. If $u \notin p(v) \cup \mathrm{Odd}(p(v))$, then it is trivially focussed over $\{u\}$ in $\Gamma'$. Finally, in each of the cases of Equation \ref{eq:PauliPlanarUpdate} where $\alpha(u) \in \{0, \pi\}$, the compatible corrections in the focussed conditions do not change, so it remains focussed over $\{u\}$ in $\Gamma'$.

Suppose instead that $v \neq u$, $u \in p(v) \cup \mathrm{Odd}(p(v))$, and $\alpha(u) \in \{\tfrac{\pi}{2}, \tfrac{3\pi}{2}\}$, so $p'(v) = p(v) \Delta p(u)$. Both $p(v)$ and $p(u)$ remain focussed over $\overline{O} \setminus \{u, v\}$ in $\Gamma'$, and Lemma \ref{lemma:AddFocussed} gives the same for $p'(v)$. For each case of $\lambda(u) \in \{XY, XZ, YZ\}$, what was a compatible correction for $\lambda(u)$ is now incompatible with $\lambda'(u)$, so $p(v)$ is not focussed over $\{u\}$ in $\Gamma'$. Additionally, conditions \ref{PF7} and \ref{PF9} mean that $p(u)$ is also not focussed over $\{u\}$. However, we can then use Lemma \ref{lemma:AddNotFocussed} to conclude that their combination ($p'(v)$) is focussed.

The proof of updating a focussed set $\hat{p}$ to $\hat{p'}$ follows similarly.
\end{proof}


\begin{lemma}\label{lemma:FixPauliStrings}
Let $(\Gamma, \alpha)$ and $(\Gamma', \alpha')$ be the measurement patterns and respective focussed Pauli flows $(p, \prec)$ and $(p', \prec)$ from Lemma \ref{lemma:FixPauliProof} formed by changing the label of some $u \in \overline{O}$ from a planar measurement to a Pauli measurement. Then the new primary extraction string for each vertex $v \in \overline{O} \setminus \{u\}$ is given by
\begin{equation}
\mathbf{P'}^{\bot v} = P'^{\bot v}_v P^{\bot v}_v \mathbf{P}^{\bot v} \left( \cos \alpha(u) I - i \sin (-1)^{D_u + A_{u \to v}^p} \alpha(u) P^{u \to v}_v \mathbf{P}^{\bot u} \right)^{F_{v \to u}^p}
\end{equation}
and $\mathbf{P'}^{\bot u} = \mathbf{P}^{\bot u}$. Furthermore, a stabilizer from a focussed set $\hat{p}$ is updated to
\begin{equation}
\mathbf{\hat{P'}} = \mathbf{\hat{P}} \left( \cos \alpha(u) I - i \sin (-1)^{D_u} \alpha(u) \mathbf{P}^{\bot u} \right)^{G_{\hat{p} \to u}}
\end{equation}
\end{lemma}

\begin{proof}
The case for $u$ is simple because $p'(u) = p(u)$ and the linear map $\mathcal{C}$ in the definition of extraction strings is unchanged. For the remainder, we consider some $v \in \overline{O} \setminus \{u\}$.

Firstly, we show that $P'^{\bot v} P^{\bot v} \approx I$ if $\alpha(u) \in \{0, \pi\}$ or $(P^{u \to v})^{F_{v \to u}^p}$ if $\alpha(u) \in \{\tfrac{\pi}{2}, \tfrac{3\pi}{2}\}$ to deduce that $\mathbf{P'}^{\bot v}$ only acts on the outputs. If either $u \notin p(v) \cup \mathrm{Odd}(p(v))$ (i.e. $F_{v \to u}^p = 0$) or $\alpha(u) \in \{0, \pi\}$, then $p'(v) = p(v)$ and hence $P'^{\bot v} = P^{\bot v}$. Otherwise, we have $p'(v) = p(v) \Delta p(u)$. By following a similar logic to the proof of Lemma \ref{lemma:MultiplyStrings} considering $\{v\}$ instead of $O$, we end up with $P'^{\bot v} \approx P^{\bot v} P^{u \to v}$.

Now we show that $\mathbf{P'}^{\bot v}$ is a primary extraction string in $(\Gamma', \alpha')$. Let $\mathcal{C'}$ be the linear map from Definition \ref{def:ExtractionStringDef} for $(\Gamma', \alpha')$ and let $\mathcal{C}$ be the corresponding for $(\Gamma, \alpha)$ (plus an extra projection on $u$ if $u \preceq v$).

\begin{align}
\mathcal{C'} &:= \left( \prod_{\substack{w \succ v \nonumber \\ \lambda'(w) \notin \{X, Y, Z\}}} \bra{+_{\lambda'(w), 0}}_w \right) \left( \prod_{\substack{w \in \overline{O} \setminus \{v\} \nonumber \\ \lambda'(w) \in \{X, Y, Z\}}} \bra{+_{\lambda'(w), \alpha'(w)}}_w \right) E_G N_{\overline{I}} \nonumber \\
&\approx \bra{+_{\lambda(u), \alpha(u)}}_u \left( \prod_{\substack{w \succ v \nonumber \\ w \neq u \nonumber \\ \lambda(w) \notin \{X, Y, Z\}}} \bra{+_{\lambda(w), 0}}_w \right) \left( \prod_{\substack{w \in \overline{O} \setminus \{v\} \nonumber \\ \lambda(w) \in \{X, Y, Z\}}} \bra{+_{\lambda(w), \alpha(w)}}_w \right) E_G N_{\overline{I}} \nonumber \\
&\approx \left( \prod_{\substack{w \succ v \vee w = u \nonumber \\ \lambda(w) \notin \{X, Y, Z\}}} \bra{+_{\lambda(w), 0}}_w \right) \left( \prod_{\substack{w \in \overline{O} \setminus \{v\} \nonumber \\ \lambda(w) \in \{X, Y, Z\}}} \bra{+_{\lambda(w), \alpha(w)}}_w \right) e^{(-1)^{D_u} i \tfrac{\alpha(u)}{2} P^{\bot u}_u} E_G N_{\overline{I}} \nonumber \\
&= e^{(-1)^{D_u + A_{u \to v}^p} i \tfrac{\alpha(u)}{2} P^{u \to v}_v \mathbf{P}^{\bot u}} \mathcal{C} \\
&= e^{(-1)^{D_u + A_{u \to v}^p} i \tfrac{\alpha(u)}{2} P^{u \to v}_v \mathbf{P}^{\bot u}} P^{\bot v}_v \mathbf{P}^{\bot v} \mathcal{C} \nonumber \\
&= P^{\bot v}_v e^{(-1)^{D_u + A_{u \to v}^p + F_{u \to v}^p} i \tfrac{\alpha(u)}{2} P^{u \to v}_v \mathbf{P}^{\bot u}} \mathbf{P}^{\bot v} \mathcal{C} \nonumber \\
&= P^{\bot v}_v \mathbf{P}^{\bot v} e^{(-1)^{D_u + A_{u \to v}^p + F_{v \to u}^p} i \tfrac{\alpha(u)}{2} P^{u \to v}_v \mathbf{P}^{\bot u}} \mathcal{C} \nonumber \\
&= P^{\bot v}_v \mathbf{P}^{\bot v} e^{-(-1)^{D_u + A_{u \to v}^p} F_{v \to u}^p i \alpha(u) P^{u \to v}_v \mathbf{P}^{\bot u}} e^{(-1)^{D_u + A_{u \to v}^p} i \tfrac{\alpha(u)}{2} P^{u \to v}_v \mathbf{P}^{\bot u}} \mathcal{C} \nonumber \\
&\approx P^{\bot v}_v \mathbf{P}^{\bot v} e^{-(-1)^{D_u + A_{u \to v}^p} F_{v \to u}^p i \alpha(u) P^{u \to v}_v \mathbf{P}^{\bot u}} \mathcal{C'} \nonumber \\
&= P^{\bot v}_v \mathbf{P}^{\bot v} \left( \cos \alpha(u) I - i \sin (-1)^{D_u + A_{u \to v}^p} \alpha(u) P^{u \to v}_v \mathbf{P}^{\bot u} \right)^{F_{v \to u}^p} \mathcal{C'} \nonumber \\
&= P'^{\bot v}_v \mathbf{P'}^{\bot v} \mathcal{C'} \nonumber
\end{align}

In the above derivation, moving $P^{\bot v}_v \mathbf{P}^{\bot v}$ through the rotation requires Lemma \ref{lemma:AntiCommutingStringsProof} and $P^{\bot v} P^{u \to v} = (-1)^{F_{u \to v}^p} P^{u \to v} P^{\bot v}$.

Finally, by doing a case distinction on $F_{v \to u}^p$ and Lemma \ref{lemma:MultiplyStrings}, the extraction string from $(p', \prec)$ matches $\mathbf{P'}^{\bot v}$ up to phase, with Lemma \ref{lemma:EqualStrings} making them equal.

The proof for focussed sets and stabilizers follows similarly.
\end{proof}


\begin{theorem}\label{thrm:FixPauliRelateProof}(Restatement of Theorem \ref{thrm:FixPauliRelate})
Let $(\Gamma, \alpha)$ describe a measurement pattern with some vertex $u \in \overline{O}$ such that $\lambda(u) \in \{XY, XZ, YZ\}$ and $\alpha(u) \in \{0, \tfrac{\pi}{2}, \pi, \tfrac{3\pi}{2}\}$. Relabelling $u$ to the equivalent Pauli label corresponds to pushing the rotation from $u$ to the start of the PDDAG and absorbing it into the initial stabilizer process.
\end{theorem}

\begin{proof}
The corresponding PDDAG initially has a rotation node $((-1)^{D_u}\mathbf{P}^{\bot u}, \alpha(u))$ (representing the operator $e^{(-1)^{D_u} i \tfrac{\alpha(u)}{2} \mathbf{P}^{\bot u}}$) since $u$ originally has a planar label. Moving this to the front of the circuit updates the Pauli strings of other vertices according to the \ref{lemma:CommutationRules}, so we need to check that the same update is made to planar vertices and the isometry tableau is updated appropriately.

For any planar vertex $v \preceq u$, $v \notin p(u) \cup \mathrm{Odd}(p(u))$ and hence $P^{u \to v} = I$ and $A_{u \to v}^p = 0$. If $u \notin p(v) \cup \mathrm{Odd}(p(v))$ then $\mathbf{P}^{\bot v}$ and $\mathbf{P}^{\bot u}$ commute, so no update is needed which matches the formula for $\mathbf{P'}^{\bot v}$ with $F_{v \to u}^p = 0$. Otherwise, they anticommute. If $\alpha(u) = 0$, the rotation is an identity so $\mathbf{P}^{\bot v}$ is unchanged. If $\alpha(u) = \pi$, the rotation is a Pauli operator and hence flips the phase of $v$'s rotation, matching the $\cos \alpha(u)$ update. For $\alpha(u) \in \{\tfrac{\pi}{2}, \tfrac{3\pi}{2}\}$, the extraction string is updated by multiplication with the correct phase.

For the tableau, we need to handle rows from both the inputs and free stabilizers. For an input $i$ gained by input extension with $\lambda(i) = XY$, the same reductions can be made as for planar vertices and the updates for free stabilizers are essentially identical. In each case of $\alpha(u)$, these updates match the effect on the isometry tableau since the \ref{lemma:CommutationRules} show the same effect for Paulis and generic rotations.
\end{proof}

\begin{example}\label{ex:FixPauliExample}
Let's start with the same measurement pattern from Example \ref{ex:ExtractionExample}. Suppose $\alpha(c) = \tfrac{\pi}{2}$, so we can obtain a semantically equivalent pattern by relabelling $(\lambda'(c), \alpha'(c)) := (Y, 0)$. This would make $p(a)$, $p(b)$ and $\hat{p}$ no longer focussed. Refocussing gives:

\begin{center}
\begin{tabular}{cc|ccc}
$v$ & $\lambda'(v)$ & $p'(v)$ & $\mathrm{Odd}(p'(v))$ & $\{u | v \prec' u\}$ \\
\hline
$i$ & $XY$ & $b, o_2$ & $i, a$ & $a, b, o_1, o_2$ \\
$a$ & $YZ$ & $a, c, d, o_1, o_2$ & $c, d, o_1, o_2$ & $o_1, o_2$ \\
$b$ & $XY$ & $c, d$ & $b, c, d, o_1, o_2$ & $o_1, o_2$ \\
$c$ & $Y$ & $o_1$ & $c$ & $o_1$ \\
$d$ & $Y$ & $o_2$ & $d$ & $o_2$ \\
\end{tabular}
\end{center}
and $\hat{p'} = \{c, o_1, o_2\}$. The restriction of $\prec$ to $\prec'$ is not necessary, but will help to see the resemblance to the PDDAG later.

Now let's look at applying the corresponding rewrite on the PDDAG from \ref{ex:ExtractionExample}. That is, we use the \ref{lemma:CommutationRules} to move $(X_1 I_2, \tfrac{\pi}{2})$ (i.e. $e^{i \tfrac{\pi}{4} X_1 I_2}$) into the tableau. $X_1 I_2$ anticommutes with the rotations from $a$ and $b$, the tableau row from $b$, and the free stabilizer from $\hat{p}$. Each of these strings will then be updated by multiplication by $X_1 I_2$. These are exactly the correction/focussed sets that were updated in the measurement pattern by combining with $p(c)$. The result of this rewrite on the PDDAG is:

\begin{center}
\addtolength{\tabcolsep}{-3pt}
\begin{tabular}{c|cc|c}
Ins & \multicolumn{2}{c|}{Outs} & Sign \\
\hline
$X$ & $Z$ & $Z$ & $(-)^{a_d}$ \\
\hline
$Z$ & & $X$ & $+$ \\
\hline
& $Y$ & $X$ & $+$ \\
\end{tabular}
\addtolength{\tabcolsep}{3pt}
\begin{tikzcd}[ampersand replacement=\&,row sep=0.cm,column sep=0.5cm]
(I_1 X_2, \alpha(i)) \arrow{r} \arrow{rd} \& ((-1)^{a_d+1}Y_1 Y_2, \alpha(a)) \\
\& ((-1)^{a_d}Z_1 Z_2, \alpha(b))
\end{tikzcd}
\end{center}

One can check that this is identical to the PDDAG extracted using $(p', \prec)$.
\end{example}

\subsection{$Z$ Measurement Elimination}


\begin{lemma}\label{lemma:ZElimSemantics}
Let $(\Gamma, \alpha)$ describe a measurement pattern with some vertex $u \in \overline{O}$ such that $\lambda(u) \in \{XZ, YZ, Z\}$ and $\alpha(u) = a\pi$ ($a \in \{0, 1\}$). Then eliminating vertex $u$ gives an equivalent measurement pattern $(\Gamma', \alpha')$ where $\Gamma' = (G \setminus \{u\}, I, O, \lambda)$, for all $w \in \overline{O} \setminus \{u\}$
\begin{equation}
\alpha'(w) := \begin{cases}
\alpha(w) + a\pi & \text{if } w \in N_G(u) \wedge \lambda(w) \in \{XY, X, Y\} \\
(-1)^a \alpha(w) & \text{if } w \in N_G(u) \wedge \lambda(w) \in \{XZ, YZ\} \\
\alpha(w) & \text{otherwise}
\end{cases}
\end{equation}
and each output neighbouring $u$ is followed by a $Z^a$~gate.
\end{lemma}

\begin{proof}
This is just an extension of the planar case given in the proof of \cite[Lemma 4.7]{Backens2020} to cover Pauli labels by considering the embeddings from Equation \ref{eq:PauliPlanarMap}. The idea is that the $Z$-basis projections $\bra{0}$ and $\bra{1}$ reduce each of the \CZ~gates on that qubit to either an identity or a $Z$ gate (modelled as a $Z$ conditioned on the measurement angle), in the same way that they reduce when applied to a $\ket{0}$ or $\ket{1}$. This $Z$ commutes with other \CZ~gates to the neighbouring outputs or to change the effective measurement on the neighbour. The qubit $u$ can then be removed since the only operations on it are the preparation and measurement.
\end{proof}


\begin{lemma}\label{lemma:PreserveZElimProof}(Generalisation of \cite[Lemma 3.4]{Backens2020})
Let $(\Gamma, \alpha)$ and $(\Gamma', \alpha')$ be measurement patterns related by elimination of a vertex $u$ as in Lemma \ref{lemma:ZElimSemantics}. If $(p, \prec)$ is a Pauli flow for $\Gamma$, then $(p', \prec)$ is a Pauli flow for $\Gamma'$ where for any $v \in \overline{O} \setminus \{u\}$
\begin{equation}
p'(v) := \begin{cases}
p(v) \Delta p(u) & \text{if } u \in p(v) \\
p(v) & \text{if } u \notin p(v)
\end{cases}
\end{equation}

If $(p, \prec)$ is focussed, then so is $(p', \prec)$. Furthermore, if $\hat{p}$ is a focussed set for $\Gamma$, then it is also a focussed set for $\Gamma'$.
\end{lemma}

\begin{proof}
We first need to show that $p'$ is a function from $\overline{O} \setminus \{u\}$ to $\mathcal{P}[\overline{I} \setminus \{u\}]$. This follows since $\lambda(u) \in \{XZ, YZ, Z\}$ implies $u \in p(u)$ (by conditions \ref{PF5}, \ref{PF6}, and \ref{PF8}), so both cases of the definition of $p'(v)$ give $u \notin p'(v)$.

Because $\lambda(u) \in \{XZ, YZ, Z\}$, if $u \in p(v)$ then $v \prec u$ from \ref{PF1}. This means $(p', \prec)$ is a valid Pauli flow for $(G, I, O, \lambda)$ by repeated applications of Lemma \ref{lemma:AddFlows}. In particular, this means it satisfies all of the Pauli flow conditions for any $v \in \overline{O} \setminus \{u\}$. These aren't invalidated when restricting to $G \setminus \{u\}$ since $\mathrm{Odd}_{G \setminus \{u\}}(A) = \mathrm{Odd}_G(A) \setminus \{u\}$ for any $A \subseteq V \setminus \{u\}$ and $\forall v \in \overline{O} \setminus \{u\} . p'(v) \subseteq V \setminus \{u\}$.

If $(p, \prec)$ is a focussed Pauli flow, then $u \notin p(v)$ for any $v \neq u$ since $\lambda(u) \in \{XZ, YZ, Z\}$. This means $p'(v) = p(v)$ for every $v$, so the focussed property is trivially preserved. Similarly, the same holds for focussed sets.
\end{proof}


\begin{lemma}
Let $(\Gamma, \alpha)$ and $(\Gamma', \alpha')$ be the measurement patterns and respective Pauli flows $(p, \prec)$ and $(p', \prec)$ from Lemma \ref{lemma:PreserveZElimProof} formed by the elimination of a vertex $u$. Then the new primary extraction string for each vertex $v \in \overline{O} \setminus \{u\}$ is given by
\begin{equation}\label{eq:ZElimStrings}
\mathbf{P'}^{\bot v} = (-1)^{a|p(v) \cap N_G(u) \cap (O \cup \{v\})|}  \left( \prod_{\substack{n \in (N_G(u) \cup \{u\}) \cap \overline{O} \\ n \succ v \\ n = u \Rightarrow \lambda(n) \in \{XZ, YZ\} \\ n \neq u \Rightarrow \lambda(n) = XY}} (-1)^{F_{v \to n}^p} \right)^{a} \mathbf{P}^{\bot v}
\end{equation}
where $\alpha(u) = a \pi$. Furthermore, the stabilizer from a focussed set $\hat{p}$ is updated similarly.
\end{lemma}

\begin{proof}
The focussed property implies $u \notin p(v)$ (resp. $u \notin \hat{p}$), so $p'(v) = p(v)$ (resp. $\hat{p}$ is still a focussed set). Then $p'(v) \cup \mathrm{Odd}_{G \setminus u}(p'(v)) = (p(v) \cup \mathrm{Odd}_G(p(v))) \setminus \{u\}$ (resp. $\hat{p} \cup \mathrm{Odd}_{G \setminus u}(\hat{p}) = (\hat{p} \cup \mathrm{Odd}_G(\hat{p})) \setminus \{u\}$). This is enough to show that the strings are preserved up to phase.

We find a primary extraction string for $v$ in $(\Gamma', \alpha')$ by considering the linear map $\mathcal{C'}$ from Definition \ref{def:ExtractionStringDef}, using the respective linear map $\mathcal{C}$ for $(\Gamma, \alpha)$.
\begin{equation}
\mathcal{C'} = \left( \prod_{\substack{w \succ v \\ w \neq u \\ \lambda(w) \notin \{X, Y, Z\}}} \bra{+_{\lambda(w), 0}}_w \right) \left( \prod_{\substack{w \in \overline{O} \setminus \{u\} \\ \lambda(w) \in \{X, Y, Z\}}} \bra{+_{\lambda(w), \alpha'(w)}}_w \right) E_{G \setminus u} N_{\overline{I} \setminus \{u\}}
\end{equation}

We can reintroduce qubit $u$ with measurement angle $0$. We can then map the angles for the Pauli projections from $\alpha'$ to $\alpha$ using orthogonal Paulis, letting $z_u := 1$ if $\lambda(u) = Z$ and $0$ if $\lambda(u) \in \{XZ, YZ\}$.

\begin{equation}
\begin{split}
\mathcal{C'} \approx& \sqrt{2} \left( \prod_{\substack{w \succ v \\ w \neq u \\ \lambda(w) \notin \{X, Y, Z\}}} \bra{+_{\lambda(w), 0}}_w \right) \left( \prod_{\substack{w \in \overline{O} \setminus \{u\} \\ \lambda(w) \in \{X, Y, Z\}}} \bra{+_{\lambda(w), \alpha'(w)}}_w \right) \bra{+_{\lambda(u), 0}}_u E_{G} N_{\overline{I}} \\
\approx& \sqrt{2} \left( \prod_{\substack{w \succ v \vee w = u \\ \lambda(w) \notin \{X, Y, Z\}}} \bra{+_{\lambda(w), 0}}_w \right) \left( \prod_{\substack{w \in \overline{O} \\ \lambda(w) \in \{X, Y, Z\}}} \bra{+_{\lambda(w), \alpha(w)}}_w \right) \\ & (P^{\bot u}_u)^{z_u a} \left( \prod_{\substack{n \in N_G(u) \\ \lambda(n) \in \{X, Y\}}} P^{\bot n}_n \right)^{a} E_{G} N_{\overline{I}}
\end{split}
\end{equation}

We can then apply the stabilizer of the graph state centred around $u$ to move these Paulis to other qubits, some of which can be absorbed into projections again.

\begin{equation}
\begin{split}
\mathcal{C'} \approx& \sqrt{2} \left( \prod_{\substack{w \succ v \\ \lambda(w) \notin \{X, Y, Z\}}} \bra{+_{\lambda(w), 0}}_w \right) \left( \prod_{\substack{w \in \overline{O} \\ \lambda(w) \in \{X, Y, Z\}}} \bra{+_{\lambda(w), \alpha(w)}}_w \right) \\ & (P^{\bot u}_u)^{(1-z_u) a} \left( \prod_{n \in N_G(u) \cap (O \cup \{v\})} Z_n \right)^a \left( \prod_{\substack{n \in N_G(u) \cap \overline{O} \\ n \succ v \\ \lambda(n) = XY}} P^{\bot n}_n \right)^{a} \left( \prod_{\substack{n \in N_G(u) \cap \overline{O} \\ n \preceq v \\ \lambda(n) \notin \{X, Y, Z\}}} Z_n \right)^{a} E_{G} N_{\overline{I}}
\end{split}
\end{equation}

Each of the Paulis trapped under projections can be replaced by the primary extraction strings of those qubits. We leave the extraction strings without any order since reordering only introduces phases at this point which will be later removed.

\begin{equation}
\begin{split}
\mathcal{C'} \approx& \sqrt{2} \left(\prod_{\substack{w = u \\ w \succ v}} P^{u \to v}_v \mathbf{P}^{\bot u} \prod_{\substack{w = u \\ w \preceq v}} P^{\bot u}_u \right)^{(1-z_u) a} \left( \prod_{n \in N_G(u) \cap (O \cup \{v\})} Z_n \right)^a \\ & \left( \prod_{\substack{n \in N_G(u) \cap \overline{O} \\ n \succ v \\ \lambda(n) = XY}} P^{n \to v}_v \mathbf{P}^{\bot n} \right)^{a} \left( \prod_{\substack{n \in N_G(u) \cap \overline{O} \\ n \preceq v \\ \lambda(n) \notin \{X, Y, Z\}}} Z_n \right)^{a} \mathcal{C} \\
\end{split}
\end{equation}

We can now introduce $P^{\bot v}_v \mathbf{P}^{\bot v}$. Commuting this to the front induces phase changes according to Lemma \ref{lemma:AntiCommutingStringsProof} and $P^{\bot v} P^{u \to v} = (-1)^{F_{u \to v}^p} P^{u \to v} P^{\bot v}$. We then undo the previous steps to get back to $\mathcal{C'}$.

\begin{align}
\mathcal{C'} \approx& \sqrt{2} \left(\prod_{\substack{w = u \\ w \succ v}} P^{u \to v}_v \mathbf{P}^{\bot u} \prod_{\substack{w = u \\ w \preceq v}} P^{\bot u}_u \right)^{(1-z_u) a} \left( \prod_{n \in N_G(u) \cap (O \cup \{v\})} Z_n \right)^a \nonumber \\ & \left( \prod_{\substack{n \in N_G(u) \cap \overline{O} \\ n \succ v \\ \lambda(n) = XY}} P^{n \to v}_v \mathbf{P}^{\bot n} \right)^{a} \left( \prod_{\substack{n \in N_G(u) \cap \overline{O} \\ n \preceq v \\ \lambda(n) \notin \{X, Y, Z\}}} Z_n \right)^{a} P^{\bot v}_v \mathbf{P}^{\bot v} \mathcal{C} \nonumber \\
=& \sqrt{2} P^{\bot v}_v \left(\prod_{\substack{w = u \\ w \succ v}} (-1)^{F_{u \to v}^p} P^{u \to v}_v \mathbf{P}^{\bot u} \prod_{\substack{w = u \\ w \preceq v}} P^{\bot u}_u \right)^{(1-z_u) a} \\ & (-1)^{a|p(v) \cap N_G(u) \cap \{v\}|} \left( \prod_{n \in N_G(u) \cap (O \cup \{v\})} Z_n \right)^a \nonumber \\ & \left( \prod_{\substack{n \in N_G(u) \cap \overline{O} \\ n \succ v \\ \lambda(n) = XY}} (-1)^{F_{n \to v}^p} P^{n \to v}_v \mathbf{P}^{\bot n} \right)^{a} \left( \prod_{\substack{n \in N_G(u) \cap \overline{O} \\ n \preceq v \\ \lambda(n) \notin \{X, Y, Z\}}} Z_n \right)^{a} \mathbf{P}^{\bot v} \mathcal{C} \nonumber \\
=& \sqrt{2} P^{\bot v}_v \mathbf{P}^{\bot v} \left(\prod_{\substack{w = u \\ w \succ v}} (-1)^{F_{v \to u}^p} P^{u \to v}_v \mathbf{P}^{\bot u} \prod_{\substack{w = u \\ w \preceq v}} P^{\bot u}_u \right)^{(1-z_u) a} \nonumber \\ & (-1)^{a|p(v) \cap N_G(u) \cap (O \cup \{v\})|}  \left( \prod_{n \in N_G(u) \cap (O \cup \{v\})} Z_n \right)^a \nonumber \\ & \left( \prod_{\substack{n \in N_G(u) \cap \overline{O} \\ n \succ v \\ \lambda(n) = XY}} (-1)^{F_{v \to n}^p} P^{n \to v}_v \mathbf{P}^{\bot n} \right)^{a} \left( \prod_{\substack{n \in N_G(u) \cap \overline{O} \\ n \preceq v \\ \lambda(n) \notin \{X, Y, Z\}}} Z_n \right)^{a} \mathcal{C} \nonumber \\
\approx& P^{\bot v}_v \mathbf{P}^{\bot v} (-1)^{a|p(v) \cap N_G(u) \cap (O \cup \{v\})|}  \left( \prod_{\substack{n \in (N_G(u) \cup \{u\}) \cap \overline{O} \\ n \succ v \\ n = u \Rightarrow \lambda(n) \in \{XZ, YZ\} \\ n \neq u \Rightarrow \lambda(n) = XY}} (-1)^{F_{v \to n}^p} \right)^{a} \mathcal{C'} \nonumber
\end{align}

The calculations for focussed sets follow similarly. We can finish by using Lemma \ref{lemma:EqualStrings} to conclude that these primary extraction strings (resp. stabilizers) are the ones generated by the new Pauli flow (resp. focussed sets).
\end{proof}


\begin{theorem}\label{thrm:ZElimRelateProof}(Restatement of Theorem \ref{thrm:ZElimRelate})
Let $(\Gamma, \alpha)$ describe a measurement pattern with some vertex $u \in \overline{O}$ such that $\lambda(u) \in \{XZ, YZ, Z\}$ and $\alpha(u) \in \{0, \pi\}$. Eliminating $u$ from the graph corresponds to the following sequence of actions on the PDDAG:

\begin{enumerate}
\item\label{ZElimPDDAGInitial} If $u$ has a planar ($XZ$ or $YZ$) label, then its rotation is pulled from the rotation DAG into the stabilizer block;
\item\label{ZElimPDDAGOutputs} For each neighbour $n$ of $u$ that is an output, a $Z_n$ rotation of $\alpha(u)$ is pulled from the stabilizer block through the entire rotation DAG to the end of the circuit;
\item\label{ZElimPDDAGPlanar} For each neighbour $n$ of $u$ with $\lambda(n) = XY$, a $\mathbf{P}^{\bot n}$ rotation of $\alpha(u)$ is pulled from the stabilizer block and merged with the existing rotation for $n$.
\end{enumerate}
\end{theorem}

\begin{proof}
Each action is moving Pauli operators through rotations. The \ref{lemma:CommutationRules} simplify to inducing phase flips. In the PDDAG, we can flip either the operator or angle of the rotation. The $\alpha(u) = 0$ case is trivial, so we suppose that $\alpha(u) = \pi$ ($a = 1$).

For both the rotations and tableau rows, we need only consider extraction strings of planar vertices. Each of the rotations moved through the PDDAG to/from location $w$ (either $w = u$ or $w \in N_G(u)$) would induce a phase shift in the rotation or tableau row for $v$ if the corresponding Pauli string $\mathbf{P}^{\bot w}$ anticommutes with $\mathbf{P}^{\bot v}$ and originated beyond $v$ (either $w$ is an output or $v \preceq w$), i.e. $w \in p(v) \cap O$ or $F_{v \to w}^p = 1$. Equation \ref{eq:ZElimStrings} also negates the angle on $v$ if $v \in N_G(u)$ and $\lambda(v) \in \{XZ, YZ\}$ (i.e. $v \in p(v)$), as covered in Lemma \ref{lemma:ZElimSemantics}.
\end{proof}

\begin{example}\label{ex:ZElimExample}
We will use the same example measurement pattern as Example \ref{ex:ExtractionExample}. $a$ is the only vertex with a label ready for vertex elimination, so to see something actually interesting we suppose $\alpha(a) = \pi$. Eliminating $a$ gives the following residual pattern.

\begin{center}
\input{tikz_figs/mbqc_rewrite_elim.tikz}
\begin{tabular}{ccc|ccc}
$v$ & $\lambda'(v)$ & $\alpha'(v)$ & $p'(v)$ & $\mathrm{Odd}(p'(v))$ & $\{u | v \prec u\}$ \\
\hline
$i$ & $XY$ & $\alpha(i)$ & $b, o_2$ & $i$ & $b, c, o_1, o_2$ \\
$b$ & $XY$ & $\alpha(b)+\pi$ & $c, d, o_1$ & $b, d, o_1, o_2$ & $c, o_1, o_2$ \\
$c$ & $XY$ & $\alpha(c)+\pi$ & $o_1$ & $c$ & $o_1$ \\
$d$ & $Y$ & $\alpha(d)+\pi$ & $o_2$ & $d$ & $o_2$ \\
\end{tabular}
\end{center}

The focussed set $\hat{p} = \{c, o_2\}$ remains valid. Following the rewrite sequence in the PDDAG, step \ref{ZElimPDDAGInitial} sees the rotation from $a$ pulled into the tableau. The $Z_1 Y_2$ term anticommutes with the rotation and tableau row from $i$, as well as the free stabilizer from $\hat{p}$.

\begin{center}
\addtolength{\tabcolsep}{-3pt}
\begin{tabular}{c|cc|c}
Ins & \multicolumn{2}{c|}{Outs} & Sign \\
\hline
$X$ & $Y$ & $Z$ & $(-)^{a_d+1}$ \\
\hline
$Z$ & & $X$ & $-$ \\
\hline
& $Z$ & $X$ & $-$ \\
\end{tabular}
\addtolength{\tabcolsep}{3pt}
\begin{tikzcd}[ampersand replacement=\&,row sep=0.cm,column sep=0.5cm]
(-I_1 X_2, \alpha(i)) \arrow{r} \& ((-1)^{a_d+1}Y_1 Z_2, \alpha(b)) \arrow{r} \& (X_1 I_2, \alpha(c))
\end{tikzcd}
\end{center}

For step \ref{ZElimPDDAGOutputs}, none of the neighbours of $a$ are outputs, so there is no change. For step \ref{ZElimPDDAGPlanar}, only $b$ and $c$ are labelled $XY$, so we move $Y_1Z_2$ and $X_1I_2$ from the tableau and merge into their corresponding rotations.

\begin{center}
\addtolength{\tabcolsep}{-3pt}
\begin{tabular}{c|cc|c}
Ins & \multicolumn{2}{c|}{Outs} & Sign \\
\hline
$X$ & $Y$ & $Z$ & $(-)^{a_d}$ \\
\hline
$Z$ & & $X$ & $+$ \\
\hline
& $Z$ & $X$ & $+$ \\
\end{tabular}
\addtolength{\tabcolsep}{3pt}
\begin{tikzcd}[ampersand replacement=\&,row sep=0.cm,column sep=0.5cm]
(I_1 X_2, \alpha'(i)) \arrow{r} \& ((-1)^{a_d}Y_1 Z_2, \alpha'(b)) \arrow{r} \& (X_1 I_2, \alpha'(c))
\end{tikzcd}
\end{center}

This same PDDAG is obtained by extracting from the reduced measurement pattern above. Since extraction will give everything in terms of $\alpha'$, we note that the phase on the primary extraction string from $b$ is now $(-1)^{1 + \alpha'(d)/\pi} = (-1)^{a_d}$.
\end{example}

\subsection{Local Complementation of Graphs}


\begin{definition}
Given a graph $G = (V, E)$ with some designated vertex $u \in V$, the local complementation of $G$ about $u$ is the operation resulting in the graph:
\begin{equation}
G \star u := (V, E \Delta \{(v, w) | (v, u), (w, u) \in E \wedge v \neq w \})
\end{equation}
\end{definition}


\begin{lemma}\label{lemma:LocalCompSemantics}
Let $(\Gamma, \alpha)$ describe a measurement pattern with $\Gamma = (G, I, O, \lambda)$ and some vertex $u \in \overline{I}$. Then performing a local complementation about $u$ gives an equivalent measurement pattern $(\Gamma', \alpha')$ where $\Gamma' = (G \star u, I, O, \lambda')$, for all $v \in \overline{O} \setminus \{u\}$

\begin{align}
(\lambda'(u), \alpha'(u)) &:= \begin{cases}
(XZ, \alpha(u) + \tfrac{\pi}{2}) & \text{if } \lambda(u) = XY \\
(XY, \tfrac{\pi}{2} - \alpha(u)) & \text{if } \lambda(u) = XZ \\
(YZ, \alpha(u) + \tfrac{\pi}{2}) & \text{if } \lambda(u) = YZ \\
(X, \alpha(u)) & \text{if } \lambda(u) = X \\
(Z, \alpha(u) + \pi) & \text{if } \lambda(u) = Y \\
(Y, \alpha(u)) & \text{if } \lambda(u) = Z \\
\end{cases} \\
(\lambda'(v), \alpha'(v)) &:= \begin{cases}
(XY, \alpha(v) + \tfrac{\pi}{2}) & \text{if } v \in N_G(u) \wedge \lambda(v) = XY \\
(YZ, \alpha(v)) & \text{if } v \in N_G(u) \wedge \lambda(v) = XZ \\
(XZ, -\alpha(v)) & \text{if } v \in N_G(u) \wedge \lambda(v) = YZ \\
(Y, \alpha(v)) & \text{if } v \in N_G(u) \wedge \lambda(v) = X \\
(X, \alpha(v) + \pi) & \text{if } v \in N_G(u) \wedge \lambda(v) = Y \\
(\lambda(v), \alpha(v)) & \text{otherwise}
\end{cases}
\end{align}
plus each output neighbouring $u$ is followed by an $RZ(-\tfrac{\pi}{2})$~gate and if $u$ is an output then it is followed by an $RX(\tfrac{\pi}{2})$~gate.
\end{lemma}

\begin{proof}
The case for planar measurements is given in the proof of \cite[Lemma 4.3]{Backens2020}. We can extend this to the Pauli measurements using the embeddings in Equation \ref{eq:PauliPlanarMap}. For completeness, the rough idea is as follows:

In our notation, Van den Nest's Theorem for local complementation \cite{VandenNest2004} (specifically, the version allowing for inputs as in \cite[Lemma 4.3]{Backens2020}) becomes:

\begin{equation}\label{eq:VandenNest}
E_{G \star u} N_{\overline{I}} \approx e^{-i \tfrac{\pi}{4} X_u} \left( \prod_{w \in N_G(u)} e^{i \tfrac{\pi}{4} Z_w} \right) E_G N_{\overline{I}} \approx e^{i \tfrac{\pi}{4} X_u} \left( \prod_{w \in N_G(u)} e^{-i \tfrac{\pi}{4} Z_w} \right) E_G N_{\overline{I}}
\end{equation}

Each of these rotations then modifies any subsequent measurement on that qubit. By considering each case for $\lambda(v)$, we find the updates for $\lambda'$ and $\alpha'$ are suitable to capture this rotation.
\end{proof}


We quote the following Lemma from Duncan et al. \cite{Duncan2020} as it is useful in showing the preservation of Pauli flow here.

\begin{lemma}\label{lemma:OddStar}(\cite[Lemma B.4]{Duncan2020})
Given a graph $G = (V, E)$, a subset $A \subseteq V$ and a vertex $u \in V$, we have

\begin{equation}
\mathrm{Odd}_{G \star u}(A) = \begin{cases}
\mathrm{Odd}_G(A) \Delta (N_G(u) \cap A) & \text{if } u \notin \mathrm{Odd}_G(A) \\
\mathrm{Odd}_G(A) \Delta (N_G(u) \setminus A) & \text{if } u \in \mathrm{Odd}_G(A) \\
\end{cases}
\end{equation}
\end{lemma}


\begin{lemma}\label{lemma:LocalCompPreserveProof}
Let $(\Gamma, \alpha)$ and $(\Gamma', \alpha')$ be measurement patterns related by local complementation about a vertex $u \in \overline{I}$ as in Lemma \ref{lemma:LocalCompSemantics}. If $(p, \prec)$ is a Pauli flow for $\Gamma$, then $(p', \prec)$ is a Pauli flow for $\Gamma'$ where for any $v \in \overline{O}$

\begin{equation}
p'(v) := \begin{cases}
p(v) \Delta \{u\} & \text{if } u \in \mathrm{Odd}_G(p(v)) \\
p(v) & \text{if } u \notin \mathrm{Odd}_G(p(v))
\end{cases}
\end{equation}
\end{lemma}

\begin{proof}
For this proof, we will suppose $u \in \overline{O}$. The case where $u$ is an output has strictly weaker requirements from the flow conditions and hence will be covered by the cases examined here.

Firstly, we prove a useful decomposition of the odd neighbourhoods of correction sets. If $u \in \mathrm{Odd}_G(p(v))$, then using Lemma \ref{lemma:OddStar} we have

\begin{equation}
\begin{split}
\mathrm{Odd}_{G \star u}(p'(v)) &= \mathrm{Odd}_{G \star u}(p(v)) \Delta \mathrm{Odd}_{G \star u}(\{u\}) \\
&= \mathrm{Odd}_G(p(v)) \Delta \left( N_G(u) \setminus p(v) \right) \Delta N_G(u) \\
&= \mathrm{Odd}_G(p(v)) \Delta \left( N_G(u) \cap p(v) \right)
\end{split}
\end{equation}
and similarly if $u \notin \mathrm{Odd}_G(p(v))$ we have
\begin{equation}
\begin{split}
\mathrm{Odd}_{G \star u}(p'(v)) &= \mathrm{Odd}_{G \star u}(p(v)) \\
&= \mathrm{Odd}_G(p(v)) \Delta \left( N_G(u) \cap p(v) \right)
\end{split}
\end{equation}

In either case, we can now use the following equation for the remainder of the proof:

\begin{equation}\label{eq:OddStarDecomp}
\mathrm{Odd}_{G \star u}(p'(v)) = \mathrm{Odd}_G(p(v)) \Delta \left( N_G(u) \cap p(v) \right)
\end{equation}

We will first prove the Pauli flow conditions for $v = u$ and then separately show that they all hold for any $v \in \overline{O} \setminus \{u\}$. For $u$, \ref{PF1} is preserved since $p'(u)$ and $p(u)$ are either identical or differ by $u$, and $\forall w \in \overline{O} \setminus \{u\} . \lambda'(w) \in \{X, Y\} \Leftrightarrow \lambda(w) \in \{X, Y\}$.

For condition \ref{PF2}, consider any $w \in \mathrm{Odd}_{G \star u}(p'(u))$ such that $w \neq u$ and $w \preceq u$. By using Equation \ref{eq:OddStarDecomp}, we can consider the following cases:

\begin{itemize}
\item Suppose $w \in \mathrm{Odd}_G(p(u))$ and $w \notin N_G(u) \cap p(u)$. From condition \ref{PF2} for $(p, \prec)$, we have that $\lambda(w) \in \{Y, Z\}$. If $\lambda(w) = Z$ then $\lambda'(w) = Z$ which is sufficient, so suppose that $\lambda(w) = Y$. Conditon \ref{PF3} for $(p, \prec)$ gives $w \in p(u)$, and hence $w \notin N_G(u)$ which implies $\lambda'(w) = Y$.
\item Suppose $w \notin \mathrm{Odd}_G(p(u))$ and $w \in N_G(u) \cap p(u)$. It must be the case that $\lambda(w) \neq Y$ by condition \ref{PF3} for $(p, \prec)$, and condition \ref{PF1} gives $\lambda(w) \in \{X, Y\}$. So $\lambda(w) = X$ and $\lambda'(w) = Y$ which is sufficient.
\end{itemize}

For condition \ref{PF3}, consider any $w \preceq u$ such that $u \neq w$ and $\lambda'(w) = Y$. $w \in p'(u)$ if and only if $w \in p(u)$ from the definition of $p'$.

\begin{itemize}
\item Suppose $w \in N_G(u)$. The definition of $\lambda'$ gives that $\lambda(w) = X$, therefore $w \notin \mathrm{Odd}_G(p(u))$ by condition \ref{PF2} for $(p, \prec)$. Equation \ref{eq:OddStarDecomp} now reduces to give $w \in \mathrm{Odd}_{G \star u}(p'(u)) \Leftrightarrow w \in p(u)$ ($w \in p'(u)$).
\item Suppose $w \notin N_G(u)$. We similarly find $\lambda'(w) = Y$, and then $w \in \mathrm{Odd}_G(p(u)) \Leftrightarrow w \in p(u)$ from condition \ref{PF3} for $(p, \prec)$. Equation \ref{eq:OddStarDecomp} reduces to give $w \in \mathrm{Odd}_{G \star u}(p'(u)) \Leftrightarrow w \in \mathrm{Odd}_G(p(u))$. Therefore, we have $w \in \mathrm{Odd}_{G \star u}(p'(u)) \Leftrightarrow w \in p'(u)$ as required.
\end{itemize}

Suppose $u \in \mathrm{Odd}_G(p(u))$. Equation \ref{eq:OddStarDecomp} implies $u \in \mathrm{Odd}_{G \star u}(p'(u))$ because $u \notin N_G(u)$. We also have $u \in p'(u)$ if and only if $u \notin p(u)$ from the definition of $p'$. Each of the conditions \ref{PF4}-\ref{PF9} then follow straightforwardly since the flipping of set membership matches the changes to labels between $\lambda$ and $\lambda'$.

Suppose instead that $u \notin \mathrm{Odd}_G(p(u))$. Equation \ref{eq:OddStarDecomp} now implies $u \notin \mathrm{Odd}_{G \star u}(p'(u))$. Since $u$ must be corrected in some way, $u \in p(u)$ must hold, and hence $u \in p'(u)$. The only measurement bases that correspond to this correction give $\lambda(u) \in \{YZ, Y, Z\}$. This also gives $\lambda'(u) \in \{YZ, Y, Z\}$, so the conditions \ref{PF4}-\ref{PF9} still hold.

Now we need to show that the same conditions hold for any $v \in \overline{O} \setminus \{u\}$. For \ref{PF1}-\ref{PF3}, we will do this by case distinction over the cases of $p'(v)$.

Suppose that $u \in \mathrm{Odd}_G(p(v))$. Note that $u \preceq v$ implies $\lambda(u) = Z \Leftrightarrow u \notin p(v)$ and $\lambda(u) = Y \Leftrightarrow u \in p(v)$ from conditions \ref{PF1}-\ref{PF3} for $(p, \prec)$.

\ref{PF1}: Consider some $w \in p'(v) = p(v) \Delta \{u\}$ with $w \neq v$ and $w \preceq v$.

\begin{itemize}
\item If $w = u$, then $w \notin p(v)$ and $u \preceq v$, so $\lambda(u) = Z$ and $\lambda'(u) = Y$.
\item If $w \in p(v)$ and $w \neq u$, then \ref{PF1} gives $\lambda(w) \in \{X, Y\}$, so $\lambda'(w) \in \{X, Y\}$.
\end{itemize}

\ref{PF2}: Consider some $w \in \mathrm{Odd}_{G \star u}(p'(v)) = \mathrm{Odd}_G(p(v)) \Delta (N_G(u) \cap p(v))$ with $w \neq v$ and $w \preceq v$.

\begin{itemize}
\item If $w = u$ then $u \preceq v$, so $\lambda(w) \in \{Y, Z\}$ and $\lambda'(w) \in \{Y, Z\}$.
\item If $w \in \mathrm{Odd}_G(p(v))$, $w \neq u$, and $w \notin N_G(u) \cap p(v)$ then by \ref{PF2} we have $\lambda(w) \in \{Y, Z\}$. If $\lambda(w) = Z$ then $\lambda'(w) = Z$. If $\lambda(w) = Y$, then $w \in p(v)$ by \ref{PF3}, so $w \notin N_G(u)$ by Equation \ref{eq:OddStarDecomp} and therefore $\lambda'(w) = Y$.
\item If $w \notin \mathrm{Odd}_G(p(v))$ and $w \in N_G(u) \cap p(v)$, then by \ref{PF1} and \ref{PF3} we must have $\lambda(w) = X$, so $\lambda'(w) = Y$.
\end{itemize}

\ref{PF3}: Consider any $w \preceq v$ with $w \neq v$ and $\lambda'(w) = Y$.

\begin{itemize}
\item If $w \in N_G(u)$, then $\lambda(w) = X$ and $w \neq u$. Using \ref{PF2} we get $w \notin \mathrm{Odd}_G(p(v))$, so Equation \ref{eq:OddStarDecomp} reduces to $w \in \mathrm{Odd}_{G \star u}(p'(v))$ if and only if $w \in p(v)$. Since $w \neq u$, this happens if and only if $w \in p'(v) = p(v) \Delta \{u\}$.
\item If $w = u$, then $\lambda(w) = Z$. By \ref{PF1} we have $u \notin p(v)$ and by assumption $u \in \mathrm{Odd}_G(p(v))$. These mean $u \in p'(v) = p(v) \Delta \{u\}$ and $u \in \mathrm{Odd}_{G \star u}(p'(v))$ by Equation \ref{eq:OddStarDecomp}.
\item If $w \notin N_G(u)$ and $w \neq u$, then $\lambda(w) = Y$. Equation \ref{eq:OddStarDecomp} gives $w \in \mathrm{Odd}_{G \star u}(p'(v))$ if and only if $w \in \mathrm{Odd}_G(p(v))$. Using \ref{PF3}, this happens if and only if $w \in p(v)$ and thus if and only if $w \in p'(v) = p(v) \Delta \{u\}$.
\end{itemize}

Otherwise, suppose $u \notin \mathrm{Odd}_G(p(v))$.

\ref{PF1}: Consider some $w \in p'(v) = p(v)$ with $w \neq v$ and $w \preceq v$.

\begin{itemize}
\item If $w = u$, then by \ref{PF1} and \ref{PF3} we have $\lambda(u) = X = \lambda'(w)$.
\item If $w \neq u$, then by \ref{PF1} we have $\lambda(w) \in \{X, Y\}$, so $\lambda'(w) \in \{X, Y\}$.
\end{itemize}

\ref{PF2}: Consider some $w \in \mathrm{Odd}_{G \star u}(p'(v)) = \mathrm{Odd}_G(p(v)) \Delta ( N_G(u) \cap p(v) )$ with $w \neq v$ and $w \preceq v$. Since $u \notin N_G(u)$ and $u \notin \mathrm{Odd}_G(p(v))$, $w \neq u$.

\begin{itemize}
\item If $w \in N_G(u)$, then Equation \ref{eq:OddStarDecomp} gives $w \in \mathrm{Odd}_G (p(v)) \Delta p(v)$. Conditions \ref{PF1}-\ref{PF3} imply that $\lambda(w) \in \{X, Z\}$, so $\lambda'(w) \in \{Y, Z\}$.
\item If $w \notin N_G(u)$, $w \in \mathrm{Odd}_G (p(v))$ and $\lambda'(w) = \lambda(w)$, so the result follows straightforwardly from \ref{PF2} for $(p, \prec)$.
\end{itemize}

\ref{PF3}: Consider any $w \preceq v$ with $w \neq v$ and $\lambda'(w) = Y$.

\begin{itemize}
\item If $w = u$, then $\lambda(u) = Z$. From \ref{PF1}, $u \notin p(v) = p'(v)$, and $u \notin \mathrm{Odd}_{G \star u}(p'(v))$ since $u \notin \mathrm{Odd}_G(p(v))$ by assumption and $u \notin N_G(u)$.
\item If $w \in N_G(u)$, then $\lambda(w) = X$. From \ref{PF2}, $w \notin \mathrm{Odd}_G(p(v))$. Equation \ref{eq:OddStarDecomp} then reduces to $w \in \mathrm{Odd}_{G \star u}(p'(v))$ if and only if $w \in p(v) = p'(v)$.
\item If $w \notin N_G(u)$ and $w \neq u$, then $\lambda(w) = Y$. Equation \ref{eq:OddStarDecomp} gives $w \in \mathrm{Odd}_{G \star u}(p'(v))$ if and only if $w \in \mathrm{Odd}_G(p(v))$, which happens if and only if $w \in p(v) = p'(v)$ by \ref{PF3} for $(p, \prec)$.
\end{itemize}

\ref{PF4}-\ref{PF6}: Suppose $\lambda(v) \in \{XY, XZ, YZ\}$. Since $v \in \overline{O} \setminus \{u\}$, $v \in p'(v)$ if and only if $v \in p(v)$. Whether or not $v \in p(v)$ and $v \in \mathrm{Odd}(p(v))$ is given by \ref{PF4}-\ref{PF6} for $(p, \prec)$ and the value of $\lambda(v)$.

\begin{itemize}
\item If $v \in N_G(u)$, then Equation \ref{eq:OddStarDecomp} gives $v \in \mathrm{Odd}_{G \star u}(p'(v))$ if and only if $v \in \mathrm{Odd}_G(p(v)) \Delta p(v)$. This change in the odd neighbourhood coincides with the definition of $\lambda'$ flipping the labels of $XZ$ and $YZ$.
\item If $v \notin N_G(u)$, then Equation \ref{eq:OddStarDecomp} gives $v \in \mathrm{Odd}_{G \star u}(p'(v))$ if and only if $v \in \mathrm{Odd}_G(p(v))$, which is enough since $\lambda'(v) = \lambda(v)$.
\end{itemize}

\ref{PF7}-\ref{PF9}: We will prove these conditions by considering the possible cases for $\lambda(v) \in \{X, Y, Z\}$ in turn. In each case, since $v \neq u$, $v \in p'(v)$ if and only if $v \in p(v)$.

If $\lambda(v) = Z$, then $\lambda'(v) = Z$ and we get $v \in p(v)$ ($v \in p'(v)$) from \ref{PF8} for $(p, \prec)$.

Suppose $\lambda(v) = X$, so $v \in \mathrm{Odd}_G(p(v))$ from \ref{PF7} for $(p, \prec)$.

\begin{itemize}
\item If $v \in N_G(u)$, then $\lambda'(v) = Y$ and Equation \ref{eq:OddStarDecomp} gives $v \in \mathrm{Odd}_{G \star u} (p'(v))$ if and only if $v \notin p(v)$ ($v \notin p'(v)$).
\item If $v \notin N_G(u)$, then $\lambda'(v) = X$ and Equation \ref{eq:OddStarDecomp} gives $v \in \mathrm{Odd}_{G \star u}(p'(v))$.
\end{itemize}

Suppose $\lambda(v) = Y$, so $v \in p(v) \Leftrightarrow v \notin \mathrm{Odd}_G(p(v))$ from \ref{PF9} for $(p, \prec)$.

\begin{itemize}
\item If $v \in N_G(u)$, then $\lambda'(v) = X$. Since $v \in \mathrm{Odd}_G(p(v)) \Delta p(v)$, Equation \ref{eq:OddStarDecomp} gives $v \in \mathrm{Odd}_{G \star u}(p'(v))$.
\item If $v \notin N_G(u)$, then $\lambda'(v) = Y$. Equation \ref{eq:OddStarDecomp} gives $v \in \mathrm{Odd}_{G \star u}(p'(v))$ if and only if $v \in \mathrm{Odd}_G(p(v))$. From \ref{PF9} and $v \neq u$, this happens if and only if $v \notin p'(v)$.
\end{itemize}
\end{proof}


\begin{lemma}\label{lemma:LocalCompFocussed}
Let $(\Gamma, \alpha)$ and $(\Gamma', \alpha')$ be the measurement patterns from Lemma \ref{lemma:LocalCompSemantics} related by local complementation about a vertex $u \in \overline{I}$. If $(p, \prec)$ is a focussed Pauli flow for $\Gamma$, then $(p', \prec)$ is a focussed pauli flow for $\Gamma'$, where $p'$ is defined in $\succ$-order (from outputs backwards) as

\begin{equation}
p'(v) := \left\{ \begin{array}{ll} p(v) \Delta \{u\} & \text{if } u \in \mathrm{Odd}_G(p(v)) \\ p(v) & \text{if } u \notin \mathrm{Odd}_G(p(v)) \end{array} \right\} \Delta \bigdelta_{\substack{w \in \overline{O} \setminus \{v\} \\ w \in p(v) \cup \mathrm{Odd}_{G}(p(v)) \\ w \in N_G(u) \cup \{u\} \\ w \in N_G(u) \Rightarrow \lambda(w) = XY \\ w = u \Rightarrow \lambda(u) \notin \{X, Y, Z\}}} p'(w)
\end{equation}

Furthermore, if $\hat{p}$ is a focussed set for $\Gamma$, then we obtain a corresponding focussed set for $\Gamma'$ as

\begin{equation}
\hat{p'} := \left\{ \begin{array}{ll} \hat{p} \Delta \{u\} & \text{if } u \in \mathrm{Odd}_G(\hat{p}) \\ \hat{p} & \text{if } u \notin \mathrm{Odd}_G(\hat{p}) \end{array} \right\} \Delta \bigdelta_{\substack{w \in \overline{O} \\ w \in \hat{p} \cup \mathrm{Odd}_{G}(\hat{p}) \\ w \in N_G(u) \cup \{u\} \\ w \in N_G(u) \Rightarrow \lambda(w) = XY \\ w = u \Rightarrow \lambda(u) \notin \{X, Y, Z\}}} p'(w)
\end{equation}
\end{lemma}

\begin{proof}
Lemma \ref{lemma:LocalCompPreserveProof} gives a Pauli flow for $\Gamma'$, which we will rename as $(q, \prec)$. We can view $(q, \prec)$ as an intermediate in the construction of $(p', \prec)$. It is not necessarily focussed, but we can refocus it using Lemma \ref{lemma:Focussing}. Recall that this adds correction sets together according to Lemma \ref{lemma:AddFlows} for any case where the focussed conditions do not hold for $q$. We just need to show that these cases match those in the definition of $p'$.

As usual, the proofs for Pauli flows and focussed sets are almost identical, so herein we could change out references of $q$ for

\begin{equation}
\hat{q} := \begin{cases}
\hat{p} \Delta \{u\} & \text{if } u \in \mathrm{Odd}_G(\hat{p}) \\
\hat{p} & \text{if } u \notin \mathrm{Odd}_G(\hat{p})
\end{cases}
\end{equation}

as an intermediate in the construction of $\hat{p'}$.

We will find the counterexamples to the focussed conditions for $(q, \prec)$ by considering $q(v)$ for an arbitrary $v \in \overline{O}$, and all of the cases for $w \in \overline{O} \setminus \{v\}$, covering each case of $\lambda(w)$ for each of $w = u$, $w \in N_G(u)$, and $w \notin N_G(u) \cup \{u\}$.

Firstly, we have $w = u$:

\begin{itemize}
\item $\lambda(u) = XY$, $\lambda'(u) = XZ$: By \ref{FOC2}, $u \notin \mathrm{Odd}_G(p(v))$, so $q(v) = p(v)$. This hence breaks \ref{FOC1} for $q$ if $u \in p(v)$, i.e. if $u \in p(v) \cup \mathrm{Odd}_G(p(v))$. \xmark
\item $\lambda(u) = XZ$, $\lambda'(u) = XY$: \ref{FOC1} gives $u \notin p(v)$. By Equation \ref{eq:OddStarDecomp}, we break \ref{FOC2} for $q$ if $u \in \mathrm{Odd}_G(p(v))$, i.e. if $u \in p(v) \cup \mathrm{Odd}_G(p(v))$. \xmark
\item $\lambda(u) = YZ = \lambda'(u)$: \ref{FOC1} gives $u \notin p(v)$. From the definition of $q(v)$, $u \in q(v)$ (breaking \ref{FOC1} for $q$) if $u \in \mathrm{Odd}_G(p(v))$, i.e. if $u \in p(v) \cup \mathrm{Odd}_G(p(v))$. \xmark
\item $\lambda(u) = X = \lambda'(u)$: By \ref{FOC2}, $u \notin \mathrm{Odd}_G(p(v))$. We hence have $u \notin \mathrm{Odd}_{G \star u}(q(v))$ by Equation \ref{eq:OddStarDecomp}. \cmark
\item $\lambda(u) = Y$, $\lambda'(u) = Z$: \ref{FOC3} gives $u \in p(v) \Leftrightarrow u \in \mathrm{Odd}_G(p(v))$. In both cases for the definition of $q(v)$, this gives $u \notin q(v)$. \cmark
\item $\lambda(u) = Z$, $\lambda'(u) = Y$: We have $u \notin p(v)$ from \ref{FOC1}. The definition of $q(v)$ then gives $u \in q(v) \Leftrightarrow u \in \mathrm{Odd}_G(p(v))$, and Equation \ref{eq:OddStarDecomp} gives $u \in \mathrm{Odd}_{G \star u}(q(v)) \Leftrightarrow u \in \mathrm{Odd}_G(p(v))$. Hence $u \in q(v) \Leftrightarrow u \in \mathrm{Odd}_{G \star u}(q(v))$. \cmark
\end{itemize}

Next, suppose $w \in N_G(u)$. In each case, the definition of $q(v)$ gives $w \in q(v) \Leftrightarrow w \in p(v)$ and Equation \ref{eq:OddStarDecomp} gives $w \in \mathrm{Odd}_{G \star u}(q(v)) \Leftrightarrow w \in p(v) \Delta \mathrm{Odd}_G(p(v))$.

\begin{itemize}
\item $\lambda(w) = XY$, $\lambda'(w) = XY$: By \ref{FOC2}, we have $w \notin \mathrm{Odd}_G(p(v))$. This means it breaks \ref{FOC2} for $q$ when $w \in p(v)$, i.e. when $w \in p(v) \cup \mathrm{Odd}_G(p(v))$. \xmark
\item $\lambda(w) \in \{XZ, YZ, Z\}$, $\lambda'(w) \in \{XZ, YZ, Z\}$: $w \notin p(v)$ by \ref{FOC1}, so $w \notin q(v)$. \cmark
\item $\lambda(w) = X$, $\lambda'(w) = Y$: $w \notin \mathrm{Odd}_G(p(v))$ holds by \ref{FOC2}, so $w \in \mathrm{Odd}_{G \star u}(q(v)) \Leftrightarrow w \in p(v)$ and hence $w \in q(v) \Leftrightarrow w \in \mathrm{Odd}_{G \star u}(q(v))$. \cmark
\item $\lambda(w) = Y$, $\lambda'(w) = X$: By \ref{FOC3}, we have $w \notin p(v) \Delta \mathrm{Odd}_G(p(v))$, and therefore $w \notin \mathrm{Odd}_{G \star u}(q(v))$. \cmark
\end{itemize}

For the final case of $w \notin N_G(u) \cup \{u\}$, we have $\lambda'(w) = \lambda(w)$, $w \in q(v) \Leftrightarrow w \in p(v)$, and $w \in \mathrm{Odd}_{G \star u}(q(v)) \Leftrightarrow w \in \mathrm{Odd}_G(p(v))$, so all of the focussed conditions are preserved.

In summary, we find that applying Lemma \ref{lemma:Focussing} to $(q, \prec)$ will add focussed corrections to $q(v)$ for any $w \in p(v) \cup \mathrm{Odd}_G(p(v))$ such that $w = u$ and $\lambda(u) \in \{XY, XZ, YZ\}$ or $w \in N_G(u)$ and $\lambda(w) = XY$. This exactly results in $(p', \prec)$, meaning this is a focussed Pauli flow for $\Gamma'$.
\end{proof}


\begin{lemma}\label{lemma:LocalCompStrings}
Let $(\Gamma, \alpha)$ and $(\Gamma', \alpha')$ be the measurement patterns and corresponding focussed Pauli flows $(p, \prec)$ and $(p', \prec)$ from Lemma \ref{lemma:LocalCompFocussed} formed by local complementation around a vertex $u \in \overline{I}$. Let $\mathbf{P}^{\bot v}$ be the primary extraction string for $v$ in $(\Gamma, \alpha)$ from $(p, \prec)$. Then $\mathbf{P'}^{\bot v} := P'^{\bot v}_v \mathcal{R}^\dagger P^{\bot v}_v \mathbf{P}^{\bot v} \mathcal{R}$ is the corresponding primary extraction string in $(\Gamma', \alpha')$ from $(p', \prec)$, where

\begin{equation}
\begin{split}
\mathcal{R} :=& \left( \prod_{\substack{w \in N_G(u) \\ w \in O \cup \{v\}}} e^{i \tfrac{\pi}{4} Z_w} \right) \left( \prod_{\substack{w = u \\ w \in O \cup \{v\}}} e^{-i \tfrac{\pi}{4} X_w} \right) \\ & \left( \prod_{\substack{w \in \overline{O} \cap (N_G(u) \cup \{u\}) \setminus \{v\} \\ w \succ v \\ w \in N_G(u) \Rightarrow \lambda(u) = XY \\ w = u \Rightarrow \lambda(w) \notin \{X, Y, Z\}}}^\succ e^{(-1)^{D_w + A_{w \to v}^{p'}} i \tfrac{\pi}{4} P^{w \to v}_v \mathbf{P'}^{\bot w}} \right)
\end{split}
\end{equation}

Furthermore, if $\mathbf{P}$ is the stabilizer from a focussed set $\hat{p}$ in $(\Gamma, \alpha)$, then $\mathbf{P'} := \mathcal{R}^\dagger \mathbf{P} \mathcal{R}$ is the corresponding stabilizer for $(\Gamma', \alpha')$.
\end{lemma}

\begin{proof}
We proceed following a similar strategy to Lemma \ref{lemma:FixPauliStrings} of simultaneously showing that $\mathbf{P'}^{\bot v}$ operates only on the outputs and that it matches the extraction string from $(p', \prec)$ up to phase, then showing that it is a valid extraction string.

The conjugation of $P^{\bot v}_v \mathbf{P}^{\bot v}$ by the $e^{-i \tfrac{\pi}{4} X_u}$ and $e^{i \tfrac{\pi}{4} Z_w}$ terms adjusts the Pauli terms of the operator in the same way as the intermediate Pauli flow from Lemma \ref{lemma:LocalCompPreserveProof}. The $e^{-i \tfrac{\pi}{4} X_u}$ term adds an $X$ when conjugating $Z$ or $Y$, matching the addition of $\{u\}$ to $p(v)$ if $u \in \mathrm{Odd}(p(v))$. Similarly, the $e^{i \tfrac{\pi}{4} Z_w}$ term adds a $Z$ when conjugating $X$ or $Y$, matching the update of the odd neighbourhood according to Equation \ref{eq:OddStarDecomp}.

Each of the remaining terms in $\mathcal{R}$ cover the cases for the focussing in Lemma \ref{lemma:LocalCompFocussed}. The movement of $e^{(-1)^{D_w + A_{w \to v}^{p'}} i \tfrac{\pi}{4} P^{w \to v}_v \mathbf{P'}^{\bot w}}$ adds $P^{w \to v}_v \mathbf{P'}^{\bot w}$ if it anticommutes with the operator. From Lemma \ref{lemma:AntiCommutingStringsProof} and $P^{\bot v}P^{w \to v} = (-1)^{F_{w \to v}^{p_k}} P^{w \to v} P^{\bot v}$, the combined operators anticommute if $w \in p_k(v) \cup \mathrm{Odd}(p_k(v))$ (where $p_k$ is the intermediate Pauli flow at this stage in the focussing). This matches the focussing step, which adds $p'(w)$ when $w \in p_k(v) \cup \mathrm{Odd}(p_k(v))$ doesn't satisfy the focussed conditions - by Lemma \ref{lemma:LocalCompFocussed}, this corresponds to the conditions on $w$ in the definition of $\mathcal{R}$.

Combining these together, we get that $\mathcal{R}$ maps $P^{\bot v}_v \mathbf{P}^{\bot v}$ to the equivalent for $(p', \prec)$ up to phase.

To complete this for the phase, it remains to show that $\mathbf{P'}^{\bot v}$ is a valid primary extraction string. Recall from Van den Nest's Theorem (Equation \ref{eq:VandenNest}) that local complementation of the graph state is balanced by rotations on $u$ and its neighbours. Let's look at how these rotations affect the measurements on those qubits in each case.

\begin{center}
\begin{tabular}{c|r@{}l|r@{}l}
$\lambda(w)$ & \multicolumn{2}{c|}{$w = u$} & \multicolumn{2}{c}{$w \in N_G(u)$} \\
\hline
$XY$ & $\bra{+_{XZ, \tfrac{\pi}{2}}} e^{i \tfrac{\pi}{4}X}$ & $= \bra{+_{XY, 0}}$ & $\bra{+_{XY, \tfrac{\pi}{2}}} e^{-i \tfrac{\pi}{4}Z}$ & $\approx \bra{+_{XY, 0}}$ \\
$XZ$ & $\bra{+_{XY, \tfrac{\pi}{2}}} e^{i \tfrac{\pi}{4}X}$ & $\approx \bra{+_{XZ, 0}}$ & $\bra{+_{YZ, 0}} e^{-i \tfrac{\pi}{4}Z}$ & $\approx \bra{+_{XZ, 0}}$ \\
$YZ$ & $\bra{+_{YZ, \tfrac{\pi}{2}}} e^{i \tfrac{\pi}{4}X}$ & $= \bra{+_{YZ, 0}}$ & $\bra{+_{XZ, 0}} e^{-i \tfrac{\pi}{4}Z}$ & $\approx \bra{+_{YZ, 0}}$ \\
$X$ & $\bra{+_{X, \alpha(w)}} e^{i \tfrac{\pi}{4}X}$ & $\approx \bra{+_{X, \alpha(w)}}$ & $\bra{+_{Y, \alpha(w)}} e^{-i \tfrac{\pi}{4}Z}$ & $\approx \bra{+_{X, \alpha(w)}}$ \\
$Y$ & $\bra{+_{Z, \alpha(w) + \pi}} e^{i \tfrac{\pi}{4}X}$ & $\approx \bra{+_{Y, \alpha(w)}}$ & $\bra{+_{X, \alpha(w) + \pi}} e^{-i \tfrac{\pi}{4}Z}$ & $\approx \bra{+_{Y, \alpha(w)}}$ \\
$Z$ & $\bra{+_{Y, \alpha(w)}} e^{i \tfrac{\pi}{4}X}$ & $\approx \bra{+_{Z, \alpha(w)}}$ & $\bra{+_{Z, \alpha(w)}} e^{-i \tfrac{\pi}{4}Z}$ & $\approx \bra{+_{Z, \alpha(w)}}$ \\
\end{tabular}
\end{center}

We can now consider the primary extraction string of some $v \in \overline{O}$. As usual, we define $\mathcal{C'}$ to be the linear map of interest from Definition \ref{def:ExtractionStringDef} for $(\Gamma', \alpha')$ and let $\mathcal{C}$ refer to the equivalent for $(\Gamma, \alpha)$.

\begin{equation}
\mathcal{C'} := \left( \prod_{\substack{w \succ v \\ \lambda'(w) \notin \{X, Y, Z\}}} \bra{+_{\lambda'(w), 0}}_w \right) \left( \prod_{\substack{w \in \overline{O} \setminus \{v\} \\ \lambda'(w) \in \{X, Y, Z\}}} \bra{+_{\lambda'(w), \alpha'(w)}}_w \right) E_{G \star u} N_{\overline{I}} \\
\end{equation}

We start by applying the reverse of the extraction algorithm to modify the measurement angles in the diagram to match those in the table above. Since this changes the angles (which must be $0$ to extract any rotation before them), we must perform this in $\prec$-order.

\begin{equation}
\begin{split}
\mathcal{C'} \approx& \left( \prod_{\substack{w \succ v \\ w \in N_G(u) \cup \{u\} \\ w \in N_G(u) \Rightarrow \lambda'(w) = XY \\ w = u \Rightarrow \lambda'(w) \notin \{X, Y, Z\}}} \bra{+_{\lambda'(w), \tfrac{\pi}{2}}}_w e^{(-1)^{1+D_w}i\tfrac{\pi}{4} P^{\bot w}_w} \right) \\ & \left( \prod_{\substack{w \succ v \\ w \neq u \\ w \in N_G(u) \Rightarrow \lambda'(w) \in \{XZ, YZ\} \\ w \notin N_G(u) \Rightarrow \lambda'(w) \notin \{X, Y, Z\}}} \bra{+_{\lambda'(w), 0}}_w \right) \left( \prod_{\substack{w \in \overline{O} \setminus \{v\} \\ \lambda'(w) \in \{X, Y, Z\}}} \bra{+_{\lambda'(w), \alpha'(w)}}_w \right) E_{G \star u} N_{\overline{I}} \\
=& \left( \prod_{\substack{w \succ v \\ w \in N_G(u) \cup \{u\} \\ w \in N_G(u) \Rightarrow \lambda'(w) = XY \\ w = u \Rightarrow \lambda'(w) \notin \{X, Y, Z\}}}^\prec e^{(-1)^{1+D_w+A_{w \to v}^{p'}}i\tfrac{\pi}{4} P^{w \to v}_v \mathbf{P'}^{\bot w}} \right) \left( \prod_{\substack{w \succ v \\ w \in N_G(u) \cup \{u\} \\ w \in N_G(u) \Rightarrow \lambda'(w) = XY \\ w = u \Rightarrow \lambda'(w) \notin \{X, Y, Z\}}} \bra{+_{\lambda'(w), \tfrac{\pi}{2}}}_w \right) \\ & \left( \prod_{\substack{w \succ v \\ w \neq u \\ w \in N_G(u) \Rightarrow \lambda'(w) \in \{XZ, YZ\} \\ w \notin N_G(u) \Rightarrow \lambda'(w) \notin \{X, Y, Z\}}} \bra{+_{\lambda'(w), 0}}_w \right) \left( \prod_{\substack{w \in \overline{O} \setminus \{v\} \\ \lambda'(w) \in \{X, Y, Z\}}} \bra{+_{\lambda'(w), \alpha'(w)}}_w \right) E_{G \star u} N_{\overline{I}}
\end{split}
\end{equation}

\pagebreak
Next, we apply Van den Nest's Theorem to yield $\mathcal{C}$.

\begin{equation}
\begin{split}
\mathcal{C'} \approx& \left( \prod_{\substack{w \succ v \\ w \in N_G(u) \cup \{u\} \\ w \in N_G(u) \Rightarrow \lambda'(w) = XY \\ w = u \Rightarrow \lambda'(w) \notin \{X, Y, Z\}}}^\prec e^{(-1)^{1+D_w+A_{w \to v}^{p'}}i\tfrac{\pi}{4} P^{w \to v}_v \mathbf{P'}^{\bot w}} \right) \left( \prod_{\substack{w \succ v \\ w \in N_G(u) \cup \{u\} \\ w \in N_G(u) \Rightarrow \lambda'(w) = XY \\ w = u \Rightarrow \lambda'(w) \notin \{X, Y, Z\}}} \bra{+_{\lambda'(w), \tfrac{\pi}{2}}}_w \right) \\ & \left( \prod_{\substack{w \succ v \\ w \neq u \\ w \in N_G(u) \Rightarrow \lambda'(w) \in \{XZ, YZ\} \\ w \notin N_G(u) \Rightarrow \lambda'(w) \notin \{X, Y, Z\}}} \bra{+_{\lambda'(w), 0}}_w \right) \left( \prod_{\substack{w \in \overline{O} \setminus \{v\} \\ \lambda'(w) \in \{X, Y, Z\}}} \bra{+_{\lambda'(w), \alpha'(w)}}_w \right) \\ & e^{i\tfrac{\pi}{4}X_u} \left( \prod_{w \in N_G(u)} e^{-i\tfrac{\pi}{4}Z_w} \right) E_G N_{\overline{I}} \\
&\approx \left( \prod_{\substack{w \succ v \\ w \in N_G(u) \cup \{u\} \\ w \in N_G(u) \Rightarrow \lambda'(w) = XY \\ w = u \Rightarrow \lambda'(w) \notin \{X, Y, Z\}}}^\prec e^{(-1)^{1+D_w+A_{w \to v}^{p'}}i\tfrac{\pi}{4} P^{w \to v}_v \mathbf{P'}^{\bot w}} \right) \left( \prod_{\substack{w = u \\ w \in O \vee w \preceq v \\ w \in \overline{O} \setminus \{v\} \Rightarrow \lambda(w) \notin \{X, Y, Z\}}} e^{i\tfrac{\pi}{4}X_w} \right) \\ & \left( \prod_{\substack{w \in N_G(u) \\ w \in O \vee w \preceq v \\ w \in \overline{O} \setminus \{v\} \Rightarrow \lambda(w) \notin \{X, Y, Z\}}} e^{-i\tfrac{\pi}{4}Z_w} \right) \left( \prod_{\substack{w \succ v \\ \lambda(w) \notin \{X, Y, Z\}}} \bra{+_{\lambda(w), 0}}_w \right) \\ & \left( \prod_{\substack{w \in \overline{O} \setminus \{v\} \\ \lambda(w) \in \{X, Y, Z\}}} \bra{+_{\lambda(w), \alpha(w)}}_w \right) E_G N_{\overline{I}} \\
&= \mathcal{R}^\dagger \left( \prod_{\substack{w = u \neq v \\ w \preceq v \\ \lambda(w) \notin \{X, Y, Z\}}} e^{i\tfrac{\pi}{4}X_w} \right) \left( \prod_{\substack{w \in \overline{O} \cap N_G(u) \setminus \{v\} \\ w \preceq v \\ \lambda(w) \notin \{X, Y, Z\}}} e^{-i\tfrac{\pi}{4}Z_w} \right) \mathcal{C}
\end{split}
\end{equation}

\pagebreak
The final step is then to introduce $P^{\bot v}_v \mathbf{P}^{\bot v}$, move it through $\mathcal{R}$ and reverse the construction to return to $\mathcal{C'}$.

\begin{equation}
\begin{split}
\mathcal{C'} &\approx \mathcal{R}^\dagger \left( \prod_{\substack{w = u \neq v \\ w \preceq v \\ \lambda(w) \notin \{X, Y, Z\}}} e^{i\tfrac{\pi}{4}X_w} \right) \left( \prod_{\substack{w \in \overline{O} \cap N_G(u) \setminus \{v\} \\ w \preceq v \\ \lambda(w) \notin \{X, Y, Z\}}} e^{-i\tfrac{\pi}{4}Z_w} \right) \mathbf{P}^{\bot v} P^{\bot v}_v \mathcal{C} \\
&= \mathbf{P'}^{\bot v} P'^{\bot v}_v \mathcal{R}^\dagger \left( \prod_{\substack{w = u \neq v \\ w \preceq v \\ \lambda(w) \notin \{X, Y, Z\}}} e^{i\tfrac{\pi}{4}X_w} \right) \left( \prod_{\substack{w \in \overline{O} \cap N_G(u) \setminus \{v\} \\ w \preceq v \\ \lambda(w) \notin \{X, Y, Z\}}} e^{-i\tfrac{\pi}{4}Z_w} \right) \mathcal{C} \\
&\approx \mathbf{P'}^{\bot v} P'^{\bot v}_v \mathcal{C'}
\end{split}
\end{equation}

Hence $\mathbf{P'}^{\bot v}$ is a valid primary extraction string for $v$, and by Lemma \ref{lemma:EqualStrings} it is exactly the one generated by $(p', \prec)$. As usual, the proof is virtually identical for stabilizers from focussed sets.
\end{proof}


\begin{theorem}\label{thrm:LocalCompRelateProof}(Restatement of Theorem \ref{thrm:LocalCompRelate})
Performing local complementation about a vertex $u$ corresponds to the following sequence of actions on the PDDAG:

\begin{enumerate}
\item\label{LocalCompOutputs} For each output $w$ neighbouring $u$, a $(Z_w, \tfrac{\pi}{2})$ rotation is pulled from the initial stabilizer block all the way to the end of the rotation DAG;
\item\label{LocalCompAxis} If $u$ is an output, a $(X_u, -\tfrac{\pi}{2})$ rotation is pulled from the initial stabilizer block all the way to the end of the rotation DAG;
\item\label{LocalCompNeighbours} For each vertex $w$ neighbouring $u$ with $\lambda(w) = XY$, and for $u$ itself if $\lambda(u)$ is planar, we pull a $\left(\mathbf{P'}^{\bot w}, (-1)^{D_w}\tfrac{\pi}{2}\right)$ rotation from the initial stabilizer block to merge it into the existing rotation for $w$ or $u$. We do this in $\succ$-order over such $w$ vertices.
\end{enumerate}
\end{theorem}

\begin{proof}
Lemma \ref{lemma:LocalCompSemantics} shows that the additions of $\tfrac{\pi}{2}$ to measurement angles and outputs matches the set of rotations that are moved in the above PDDAG rewrite sequence.

The updates for the Pauli strings for each rotation and the stabilizers/tableau are covered by Lemma \ref{lemma:LocalCompStrings}. In particular, in each case we consider planar vertices and so all of the $P^{w \to v}_v$ terms in $\mathcal{R}$ are $I$, meaning the update to the extraction string corresponds to the effect of the \ref{lemma:CommutationRules}. However, the phase may differ in the update from $P^{\bot v}$ to $P'^{\bot v}$.

\begin{center}
\begin{tabular}{c|c|c}
$P^{\bot v}$ & $e^{i \tfrac{\pi}{4} X_v} P^{\bot v} e^{-i \tfrac{\pi}{4} X_v}$ & $e^{-i \tfrac{\pi}{4} Z_v} P^{\bot v} e^{i \tfrac{\pi}{4} Z_v}$ \\
\hline
$X$ & $X$ & $Y$ \\
$Y$ & $-Z$ & $-X$ \\
$Z$ & $Y$ & $Z$ \\
\end{tabular}
\end{center}

This shows that we experience an extra phase flip in the extraction string for $u$ when $\lambda(v) = XZ$, balanced by the negating of $\alpha(u)$ in the statement of Lemma \ref{lemma:LocalCompSemantics}. For $v \in N_G(u)$, we get an extra phase flip for $\lambda(v) \in \{XZ, YZ\}$ from Equation \ref{eq:RotateBasis}, which cancels with the above for $\lambda(v) = XZ$ and is balanced by the negating of $\alpha(v)$ when $\lambda(v) = YZ$.
\end{proof}

\begin{example}\label{ex:LocalCompExample}
Again, we start from the same measurement pattern from Example \ref{ex:ExtractionExample} and apply local complementation about vertex $d$. This complements the connectivity within $\{a, b, c, o_2\}$.

\begin{center}
\input{tikz_figs/mbqc_rewrite_lc.tikz}
\begin{tabular}{ccc|ccc}
$v$ & $\lambda'(v)$ & $\alpha'(v)$ & $p'(v)$ & $\mathrm{Odd}(p'(v))$ & $\{u | v \prec u\}$ \\
\hline
$i$ & $XY$ & $\alpha(i)$ & $b, c, o_2$ & $i, a, d, o_1$ & $a, b, c, o_1, o_2$ \\
$a$ & $XZ$ & $-\alpha(a)$ & $a, c, o_1, o_2$ & $a, d, o_1$ & $c, o_1, o_2$ \\
$b$ & $XY$ & $\alpha(b)+\tfrac{\pi}{2}$ & $c$ & $b, d, o_1, o_2$ & $c, o_1, o_2$ \\
$c$ & $XY$ & $\alpha(c)+\tfrac{\pi}{2}$ & $o_1$ & $c$ & $o_1$ \\
$d$ & $Z$ & $\alpha(d) + \pi$ & $d, o_2$ & $d, o_2$ & $o_2$ \\
\end{tabular}
\end{center}

We update the focussed set $\hat{p}$ to $\hat{p'} = \{c, o_1, o_2\}$, and to preserve semantics we also have an $RZ(-\tfrac{\pi}{2})$ applied to output $o_2$.

Now we turn to the sequence of rewrites on the PDDAG. For step \ref{LocalCompOutputs}, $o_2$ neighbours $d$ so we pull a $(I_1Z_2, \tfrac{\pi}{2})$ rotation through to the end. The \ref{lemma:CommutationRules} update the rotations with anticommuting Pauli strings (the rotation from $a$, the tableau row and rotation from $i$, and the free stabilizer from $\hat{p}$) by multiplication.

\begin{center}
\addtolength{\tabcolsep}{-3pt}
\begin{tabular}{c|cc|c}
Ins & \multicolumn{2}{c|}{Outs} & Sign \\
\hline
$X$ & $Y$ & $Z$ & $(-)^{a_d+1}$ \\
\hline
$Z$ & & $Y$ & $+$ \\
\hline
& $Z$ & $Y$ & $+$ \\
\end{tabular}
\addtolength{\tabcolsep}{3pt}
\begin{tikzcd}[ampersand replacement=\&,row sep=0.cm,column sep=0.5cm]
(I_1 Y_2, \alpha(i)) \arrow{r} \arrow{rd} \& ((-1)^{a_d+1}Z_1 X_2, \alpha(a)) \arrow{r} \arrow{rd} \& (I_1 Z_2, \tfrac{\pi}{2}) \\
\& ((-1)^{a_d+1}Y_1 Z_2, \alpha(b)) \arrow{r} \& (X_1 I_2, \alpha(c))
\end{tikzcd}
\end{center}

Step \ref{LocalCompAxis} does not apply since $d$ is not an output. Finally, in step \ref{LocalCompNeighbours}, $d$ is not planar but it is adjacent to $b$ and $c$ that are labelled as $XY$. Since $b \prec c$, we first pull out the $(X_1 I_2, \tfrac{\pi}{2})$ rotation.

\begin{center}
\addtolength{\tabcolsep}{-3pt}
\begin{tabular}{c|cc|c}
Ins & \multicolumn{2}{c|}{Outs} & Sign \\
\hline
$X$ & $Z$ & $Z$ & $(-)^{a_d+1}$ \\
\hline
$Z$ & & $Y$ & $+$ \\
\hline
& $Y$ & $Y$ & $-$ \\
\end{tabular}
\addtolength{\tabcolsep}{3pt}
\begin{tikzcd}[ampersand replacement=\&,row sep=0.cm,column sep=0.5cm]
(I_1 Y_2, \alpha(i)) \arrow{r} \arrow{rd} \& ((-1)^{a_d}Y_1 X_2, \alpha(a)) \arrow{r} \arrow{rd} \& (I_1 Z_2, \tfrac{\pi}{2}) \\
\& ((-1)^{a_d+1}Z_1 Z_2, \alpha(b)) \arrow{r} \& (X_1 I_2, \alpha(c)+\tfrac{\pi}{2})
\end{tikzcd}
\end{center}

Then we pull out the $((-1)^{a_d+1}Z_1Z_2, \tfrac{\pi}{2})$ rotation for $b$.

\begin{center}
\addtolength{\tabcolsep}{-3pt}
\begin{tabular}{c|cc|c}
Ins & \multicolumn{2}{c|}{Outs} & Sign \\
\hline
$X$ & $Z$ & $Z$ & $(-)^{a_d+1}$ \\
\hline
$Z$ & $Z$ & $X$ & $(-)^{a_d}$ \\
\hline
& $Y$ & $Y$ & $-$ \\
\end{tabular}
\addtolength{\tabcolsep}{3pt}
\begin{tikzcd}[ampersand replacement=\&,row sep=0.cm,column sep=0.5cm]
((-1)^{a_d}Z_1 X_2, \alpha(i)) \arrow{r} \arrow{rd} \& ((-1)^{a_d}Y_1 X_2, \alpha(a)) \arrow{r} \arrow{rd} \& (I_1 Z_2, \tfrac{\pi}{2}) \\
\& ((-1)^{a_d+1}Z_1 Z_2, \alpha(b)+\tfrac{\pi}{2}) \arrow{r} \& (X_1 I_2, \alpha(c)+\tfrac{\pi}{2})
\end{tikzcd}
\end{center}

We now find that this is the same PDDAG that would be extracted after applying the local complementation on the measurement pattern (take care when determining the phases during extraction as the phase from including $d$ in $p'(v) \cup \mathrm{Odd}(p'(v))$ is now $(-1)^{a_d+1}$ because $\alpha'(d) = \alpha(d) + \pi$). In the final step, it is important to have maintained the convention of extracting as $((-1)^{D_v} \mathbf{P}^{\bot v}, \alpha(v))$ and keeping the $-1$ terms with the string rather than the angle. If we had instead moved $((-1)^{a_d} Z_1Z_2, -\tfrac{\pi}{2})$, we would end up with a rotation of $\alpha(b) - \tfrac{\pi}{2}$ and different phases on the rotation and tableau row for $i$.
\end{example}

\begin{remark}\label{rem:LCDirection}
One should note that the action of local complementation is not quite self-inverse due to some measurement angles being updated by $\tfrac{\pi}{2}$. However, performing the inverse update is also valid and just relies on using the other case of Equation \ref{eq:VandenNest} to flip the angles that are added. This can similarly be simulated by just flipping the angles of the rotations being moved through the PDDAG.
\end{remark}

Now that we have the ability to simulate the effects of local complementation, extending this to pivoting is trivial since it simply decomposes into a sequence of local complementations.


\begin{definition}
Given a graph $G = (V, E)$ with some designated edge $u \sim v \in E$, the pivot of $G$ about the edge $u \sim v$ is the operation resulting in the graph:

\begin{equation}
G \wedge uv := G \star u \star v \star u (= G \star v \star u \star v)
\end{equation}
\end{definition}

\begin{lemma}
Let $(\Gamma, \alpha)$ describe a measurement pattern with $\Gamma = (G, I, O, \lambda)$ and some chosen vertices $u, v \in \overline{I}$ such that $u \sim v$. Then performing a pivot around the edge $u \sim v$ gives an equivalent measurement pattern $(\Gamma', \alpha')$ where $\Gamma' = (G \wedge uv, I, O, \hat{\lambda})$, for all $a \in \{u, v\}$ and $w \in \overline{O} \setminus \{u, v\}$

\begin{align}
(\lambda'(a), \alpha'(a)) &:= \begin{cases}
(YZ, -\alpha(a)) & \text{if } \lambda(a) = XY \\
(XZ, \tfrac{\pi}{2} - \alpha(a)) & \text{if } \lambda(a) = XZ \\
(XY, -\alpha(a)) & \text{if } \lambda(a) = YZ \\
(Z, \alpha(a)) & \text{if } \lambda(a) = X \\
(Y, \alpha(a) + \pi) & \text{if } \lambda(a) = Y \\
(X, \alpha(a)) & \text{if } \lambda(a) = Z \\
\end{cases} \\
(\lambda'(w), \alpha'(w)) &:= \begin{cases}
(\lambda(w), \alpha(w) + \pi) & \text{if } w \in N_G(u) \cap N_G(v) \wedge \lambda(w) \in \{XY, X, Y\} \\
(\lambda(w), -\alpha(w)) & \text{if } w \in N_G(u) \cap N_G(v) \wedge \lambda(w) \in \{XZ, YZ\} \\
(\lambda(w), \alpha(w)) & \text{otherwise}
\end{cases}
\end{align}
plus each output neighbouring both $u$ and $v$ is followed by a $Z$~gate and if each of $u$ or $v$ is an output then it is followed by a Hadamard gate.
\end{lemma}

\begin{proof}
We can similarly extend the case for planar measurements given in the proof of \cite[Lemma 4.5]{Backens2020} to Pauli measurements by considering how they embed into the planes as in Equation \ref{eq:PauliPlanarMap}. Similarly, this follows by using Van den Nest's Theorem \cite{VandenNest2004} (Equation \ref{eq:VandenNest}) multiple times, alternating the phases of the angles each time, and rotating the measurement projections accordingly.
\end{proof}

Since this proof requires alternating between the different phase variants of Equation \ref{eq:VandenNest}, simulating the effect in the PDDAG requires alternating between the different variants of local complementation as mentioned in Remark \ref{rem:LCDirection} but is otherwise just a trivial application of Theorem \ref{thrm:LocalCompRelate}.

\subsection{Alternative Focussed Pauli Flows}


\begin{lemma}\label{lemma:AddFlowsAndFocussedStrings}
Let $(\Gamma, \alpha)$ describe a measurement pattern with a focussed set $\hat{p}$, vertex $v \in \overline{O}$, and two focussed Pauli flows $(p, \prec)$ and $(p', \prec)$ related by the addition of $\hat{p}$ to $p(v)$ as in Lemma \ref{lemma:AddFlowsAndFocussed}. Let $\mathbf{\hat{P}}$ be the stabilizer corresponding to $\hat{p}$, and $\mathbf{P}^{\bot v}$ and $\mathbf{P'}^{\bot v}$ be the primary extraction strings for $v$ according to $(p, \prec)$ and $(p', \prec')$ respectively. Then $\mathbf{P'}^{\bot v} = \mathbf{P}^{\bot v} \mathbf{\hat{P}}$.
\end{lemma}

\begin{proof}
Lemma \ref{lemma:MultiplyStrings} immediately gives this result up to phase. To invoke Lemma \ref{lemma:EqualStrings}, we consider the linear map $\mathcal{C}$ for extraction strings of $v$. From the conditions of Lemma \ref{lemma:AddFlowsAndFocussed} and the definition of $\prec'$, any $w \in \hat{p} \cup \mathrm{Odd}(\hat{p})$ must satisfy $w \neq v$ and $\lambda(w) \notin \{X, Y, Z\} \Rightarrow v \prec w$. So both $P^{\bot v}_v \mathbf{P}^{\bot v}$ and $\mathbf{\hat{P}}$ are stabilizers of $\mathcal{C}$, meaning they combine to give the stabilizer $P^{\bot v}_v \mathbf{P}^{\bot v} \mathbf{\hat{P}}$ which defines this as a valid primary extraction string.
\end{proof}


\begin{theorem}\label{thrm:AddFlowsAndFocussedRelateProof}(Restatement of Theorem \ref{thrm:AddFlowsAndFocussedRelate})
Let $(\Gamma, \alpha)$ describe a measurement pattern with some focussed Pauli flow $(p, \prec)$, a focussed set $\hat{p}$, and some vertex $u \in \overline{O}$ such that $\forall w \in \hat{p} \cup \mathrm{Odd}(\hat{p}) . \lambda(w) \in \{XY, XZ, YZ\} \Rightarrow w \neq u \wedge u \prec w$. Updating $p(u)$ to $p(u) \Delta \hat{p}$ corresponds to a free action on the isometry tableau if $u$ is an input, and applying the \ref{lemma:GadgetStabCorrespondence} to the rotation from $u$ with the stabilizer of $\hat{p}$ if $u$ has a planar label.

Therefore, any two focussed Pauli flows for the same labelled open graph yield PDDAGs that are related by a sequence of applications of the \ref{lemma:GadgetStabCorrespondence} and free actions on the isometry tableau.
\end{theorem}

\begin{proof}
This update on the Pauli flow is the same as from Lemma \ref{lemma:AddFlowsAndFocussed}, meaning Lemma \ref{lemma:AddFlowsAndFocussedStrings} gives the update to the primary extraction strings. For both uses of $u$ in extraction, we assume $u$ has a planar label. If $u$ is an input considered for the tableau of the PDDAG, multiplication of a row of the tableau by a free stabilizer is a free action as it just gives another equivalent set of generators for the stabilizer group. On the other hand, if $u$ is some planar vertex giving a rotation in the PDDAG, we note that $\mathbf{\hat{P}}$ remains a stabilizer after applying each of the rotations for vertices $w \preceq u$ (they commute according to Lemma \ref{lemma:AntiCommutingStringsProof}), allowing us to apply the \ref{lemma:GadgetStabCorrespondence} to get the same effect.

To show that any two focussed Pauli flows $(p, \prec)$ and $(p', \prec')$ can be related using this, we start with $(p, \prec)$ and derive the sequence of transformations needed to reach $(p', \prec')$.

We consider weakening of the partial order in a Pauli flow to be a free action, as well as the reverse so long as the Pauli flow conditions are preserved, since neither affect the semantics or extraction. So suppose wlog that $\prec$ and $\prec'$ are minimal, i.e. $\prec$ is the transitive closure of $\{(u, v) | v \in p(u) \cup \mathrm{Odd}(p(u)) \wedge u \neq v \wedge \lambda(v) \in \{XY, XZ, YZ\}\}$ (we ignore Pauli measurements by Lemma \ref{lemma:InitialPaulis}).

We proceed inductively over $\prec$. Under some indexing of the vertices respecting $\prec$, let $v$ be the first for which $p(v) \neq p'(v)$, i.e. $\forall w \prec v . p(v) = p'(v)$. Our assumptions on minimality then imply that $\forall w \in \overline{O} . w \prec v \Leftrightarrow w \prec' v$.

By Lemma \ref{lemma:AddFlowsGivesFocussed}, $\hat{p^v} := p(v) \Delta p'(v)$ is a focussed set for any $v \in \overline{O}$. Conditions \ref{PF4}-\ref{PF6} imply that $\lambda(v) \in \{XY, XZ, YZ\} \Rightarrow v \notin \hat{p^v} \cup \mathrm{Odd}(\hat{p^v})$. Using conditions \ref{PF1} and \ref{PF2}, we also have $\forall w \in \hat{p^v} \cup \mathrm{Odd}(\hat{p^v}) . w \neq v \wedge \lambda(w) \in \{XY, XZ, YZ\} \Rightarrow v \prec w \vee v \prec' w$. Since $v \prec' w \Rightarrow v \preceq' w$ and $v \preceq' w \Leftrightarrow v \preceq w$ by the minimality assumption, we now have $\forall w \in \hat{p^v} \cup \mathrm{Odd}(\hat{p^v}) . \lambda(w) \in \{XY, XZ, YZ\} \Rightarrow w \neq v \wedge v \preceq w$. This means that combining $(p, \prec)$ with $\hat{v}$ according to Lemma \ref{lemma:AddFlowsAndFocussed} (and reducing to the minimal order) gives a focussed Pauli flow equal to $(p', \prec')$ up to $v$ in $\prec$.

Repeating this inductively over the remaining vertices will eventually yield $(p', \prec')$.
\end{proof}

\begin{example}\label{ex:AddFlowsAndFocussedExample}
We revisit the measurement pattern from Example \ref{ex:ExtractionExample}. Recall that we have a focussed set $\hat{p} = \{c, o_2\}$ with $\mathrm{Odd}(\hat{p}) = \{a, o_1\}$, so we can freely add this to the correction sets for any of $i$, $b$, or $d$. Suppose we do so for $b$ and $d$ to give the following focussed Pauli flow:

\begin{center}
\begin{tabular}{cc|ccc}
$v$ & $\lambda(v)$ & $p'(v)$ & $\mathrm{Odd}(p'(v))$ & $\{u | v \prec u\}$ \\
\hline
$i$ & $XY$ & $b, o_2$ & $i, a$ & $a, b, c, o_1, o_2$ \\
$a$ & $YZ$ & $a, c, d, o_2$ & $d, o_1, o_2$ & $c, o_1, o_2$ \\
$b$ & $XY$ & $d, o_1, o_2$ & $a, b, d, o_2$ & $a, c, o_1, o_2$ \\
$c$ & $XY$ & $o_1$ & $c$ & $o_1$ \\
$d$ & $Y$ & $c$ & $a, d, o_1$ & $a, c, o_1, o_2$ \\
\end{tabular}
\end{center}

The difference between $p(d)$ and $p'(d)$ has no effect on the PDDAG since $d$ is not involved in the extraction procedure. So the only rewrites needed are the multiplication of the top row of the isometry tableau by the bottom, and applying the \ref{lemma:GadgetStabCorrespondence} to the $b$ rotation mapping $Y_1Z_2 \mapsto (Y_1Z_2)(Z_1X_2) = -X_1Y_2$.

\begin{center}
\addtolength{\tabcolsep}{-3pt}
\begin{tabular}{c|cc|c}
Ins & \multicolumn{2}{c|}{Outs} & Sign \\
\hline
$X$ & $X$ & $Y$ & $(-)^{a_d}$ \\
\hline
$Z$ & & $X$ & $+$ \\
\hline
& $Z$ & $X$ & $+$ \\
\end{tabular}
\addtolength{\tabcolsep}{3pt}
\begin{tikzcd}[ampersand replacement=\&,row sep=0.cm,column sep=0.5cm]
(I_1 X_2, \alpha(i)) \arrow{r} \arrow{rd} \& ((-1)^{a_d} Z_1 Y_2, \alpha(a)) \arrow{r} \& (X_1 I_2, \tfrac{\pi}{2}) \\
\& ((-1)^{a_d}X_1 Y_2, \alpha(b)) \arrow{ru} \&
\end{tikzcd}
\end{center}
\end{example}

\end{document}

%% file: mbqc_style.tikzstyles
\tikzstyle{Measure}=[fill=black, draw=black, shape=circle, inner sep=0pt, minimum size=2mm]
\tikzstyle{Output}=[fill=white, draw=black, shape=circle, inner sep=0pt, minimum size=2mm]
\tikzstyle{InputMarker}=[fill=none, draw=black, shape=rectangle, minimum size=3mm, tikzit fill=white, inner sep=0pt]


%% file: circuits.tikzstyles
\tikzstyle{box}=[rectangle, fill=white, draw=black, thin]
\tikzstyle{H box}=[box, execute at end node={$H$}]
\tikzstyle{Z box}=[box, execute at end node={$Z$}]
\tikzstyle{Zp box}=[box, execute at end node={$Z_\pi$}]
\tikzstyle{Zpp box}=[box, execute at end node={$Z_\frac{\pi}{2}$}]
\tikzstyle{Za box}=[box, execute at end node={$Z_\alpha$}]
\tikzstyle{Zb box}=[box, execute at end node={$Z_\beta$}]
\tikzstyle{X box}=[box, execute at end node={$X$}]
\tikzstyle{Xp box}=[box, execute at end node={$X_\pi$}]
\tikzstyle{Xpp box}=[box, execute at end node={$X_\frac{\pi}{2}$}]
\tikzstyle{Xa box}=[box, execute at end node={$X_\alpha$}]
\tikzstyle{Xb box}=[box, execute at end node={$X_\beta$}]
\tikzstyle{keto box}=[execute at end node={$\ket{0}$}]
\tikzstyle{keti box}=[execute at end node={$\ket{1}$}]
\tikzstyle{ketp box}=[execute at end node={$\ket{\!+\!}$}]
\tikzstyle{ketm box}=[execute at end node={$\ket{\!-\!}$}]
\tikzstyle{Mx box}=[box, execute at end node={$M_{X}$}]
\tikzstyle{Mo box}=[box, execute at end node={$M_{\pm}$}]
\tikzstyle{Ma box}=[box, execute at end node={$M_{\pm\alpha}$}]
\tikzstyle{Mb box}=[box, execute at end node={$M_{\pm\beta}$}]
\tikzstyle{ctrl vertex}=[circle, fill=black, minimum size=2.0mm, inner sep=0.3mm]
\tikzstyle{target vertex}=[oplus, draw=black, thick, minimum size=4.0mm, inner sep=0]
\tikzstyle{swap vertex}=[times, draw=black, thick, minimum size=4.0mm, inner sep=0]
\tikzstyle{ellipses}=[execute at end node={\small{...}}]


%% file: tikz_figs/circ_gadget_long.tikz
\begin{tikzpicture}
	\begin{pgfonlayer}{nodelayer}
		\node [style=none] (0) at (0, 2) {};
		\node [style=none] (1) at (0, 0) {};
		\node [style=none] (2) at (25.5, 2) {};
		\node [style=none] (3) at (25.5, 0) {};
		\node [style=box] (4) at (8, 2) {$S$};
		\node [style=box] (5) at (11.5, 2) {$RY(\alpha_3)$};
		\node [style=box] (6) at (15.5, 0) {$RX(\alpha_4)$};
		\node [style=box] (7) at (19.5, 2) {$RY(\alpha_5)$};
		\node [style=box] (8) at (23.5, 2) {$RZ(\alpha_6)$};
		\node [style=ctrl vertex] (9) at (9.5, 2) {};
		\node [style=ctrl vertex] (10) at (13.5, 2) {};
		\node [style=ctrl vertex] (11) at (17.5, 2) {};
		\node [style=ctrl vertex] (12) at (21.5, 2) {};
		\node [style=target vertex] (13) at (9.5, 0) {};
		\node [style=target vertex] (14) at (13.5, 0) {};
		\node [style=target vertex] (15) at (17.5, 0) {};
		\node [style=target vertex] (16) at (21.5, 0) {};
		\node [style=box] (17) at (2, 2) {$RZ(\alpha_0)$};
		\node [style=box] (18) at (5.5, 2) {$RY(\alpha_2)$};
		\node [style=box] (19) at (2, 0) {$RZ(\alpha_1)$};
		\node [style=box] (21) at (23.5, 0) {$RY(\alpha_7)$};
	\end{pgfonlayer}
	\begin{pgfonlayer}{edgelayer}
		\draw (4) to (9);
		\draw (9) to (5);
		\draw (5) to (10);
		\draw (10) to (11);
		\draw (11) to (7);
		\draw (7) to (12);
		\draw (12) to (8);
		\draw (8) to (2.center);
		\draw (13) to (14);
		\draw (14) to (6);
		\draw (6) to (15);
		\draw (15) to (16);
		\draw (9) to (13);
		\draw (10) to (14);
		\draw (11) to (15);
		\draw (12) to (16);
		\draw (0.center) to (17);
		\draw (17) to (18);
		\draw (18) to (4);
		\draw (1.center) to (19);
		\draw (16) to (21);
		\draw (21) to (3.center);
		\draw (19) to (13);
	\end{pgfonlayer}
\end{tikzpicture}

%% file: tikz_figs/mbqc_measured_v0.tikz
\begin{tikzpicture}
	\begin{pgfonlayer}{nodelayer}
		\node [style=InputMarker] (0) at (0, -2) {};
		\node [style=Measure, label={[label distance=-5pt]135:$i$}] (1) at (0, -2) {};
		\node [style=Measure, label={[label distance=-5pt]135:$a$}] (2) at (2, -2) {};
		\node [style=Measure, label={[label distance=-5pt]135:$b$}] (3) at (2, 0) {};
		\node [style=Measure, label={[label distance=-5pt]135:$c$}] (4) at (2, 2) {};
		\node [style=Output, label={[label distance=-5pt]135:$o$}] (5) at (4, -2) {};
	\end{pgfonlayer}
	\begin{pgfonlayer}{edgelayer}
		\draw (1) to (2);
		\draw (2) to (5);
		\draw (2) to (3);
		\draw (3) to (4);
	\end{pgfonlayer}
\end{tikzpicture}

%% file: tikz_figs/mbqc_rewrite_initial.tikz
\begin{tikzpicture}
	\begin{pgfonlayer}{nodelayer}
		\node [style=InputMarker] (2) at (0, -1) {};
		\node [style=Measure, label={[label distance=-5pt]135:$i$}] (3) at (0, -1) {};
		\node [style=Measure, label={[label distance=-5pt]135:$d$}] (4) at (4, -1) {};
		\node [style=Measure, label={[label distance=-5pt]135:$a$}] (5) at (2, 1) {};
		\node [style=Output, label={[label distance=-5pt]135:$o_2$}] (6) at (6, -1) {};
		\node [style=Measure, label={[label distance=-5pt]135:$c$}] (7) at (4, 1) {};
		\node [style=Output, label={[label distance=-5pt]135:$o_1$}] (8) at (6, 1) {};
		\node [style=Measure, label={[label distance=-5pt]135:$b$}] (9) at (2, -1) {};
	\end{pgfonlayer}
	\begin{pgfonlayer}{edgelayer}
		\draw (5) to (4);
		\draw (4) to (6);
		\draw (5) to (7);
		\draw (7) to (4);
		\draw (7) to (8);
		\draw (3) to (9);
		\draw (9) to (5);
		\draw (9) to (4);
	\end{pgfonlayer}
\end{tikzpicture}

%% file: tikz_figs/circ_extracted_tableau.tikz
\begin{tikzpicture}
	\begin{pgfonlayer}{nodelayer}
		\node [style=none] (1) at (0, 0) {};
		\node [style=none] (2) at (10, 2) {};
		\node [style=none] (3) at (10, 0) {};
		\node [style=box] (4) at (3, 2) {$RX(\tfrac{\pi}{2})$};
		\node [style=box] (5) at (8, 2) {$RX(-\tfrac{\pi}{2})$};
		\node [style=ctrl vertex] (9) at (5.5, 2) {};
		\node [style=target vertex] (13) at (5.5, 0) {};
		\node [style=box] (19) at (2, 0) {$Z^{a_d + 1}$};
		\node [style=H box] (20) at (4, 0) {};
		\node [style=keto box] (21) at (0, 2) {};
	\end{pgfonlayer}
	\begin{pgfonlayer}{edgelayer}
		\draw (4) to (9);
		\draw (9) to (5);
		\draw (9) to (13);
		\draw (1.center) to (19);
		\draw (19) to (20);
		\draw (20) to (13);
		\draw (13) to (3.center);
		\draw (5) to (2.center);
		\draw (21) to (4);
	\end{pgfonlayer}
\end{tikzpicture}

%% file: tikz_figs/circ_extracted_rotations.tikz
\begin{tikzpicture}
	\begin{pgfonlayer}{nodelayer}
		\node [style=none] (0) at (0, 2) {};
		\node [style=none] (1) at (0, 0) {};
		\node [style=none] (2) at (23.5, 2) {};
		\node [style=none] (3) at (23.5, 0) {};
		\node [style=box] (5) at (11.5, 2) {$RY((-1)^{a_d}\alpha(b))$};
		\node [style=box] (6) at (11.5, 0) {$RZ((-1)^{a_d+1}\alpha(a))$};
		\node [style=box] (7) at (5, 0) {$RZ(\tfrac{\pi}{2})$};
		\node [style=box] (8) at (18, 0) {$RZ(-\tfrac{\pi}{2})$};
		\node [style=ctrl vertex] (9) at (7.5, 2) {};
		\node [style=ctrl vertex] (10) at (15.5, 2) {};
		\node [style=target vertex] (13) at (7.5, 0) {};
		\node [style=target vertex] (14) at (15.5, 0) {};
		\node [style=box] (19) at (2, 0) {$RX(\tfrac{\pi}{2})$};
		\node [style=box] (21) at (21.5, 0) {$RX(-\tfrac{\pi}{2})$};
	\end{pgfonlayer}
	\begin{pgfonlayer}{edgelayer}
		\draw (9) to (5);
		\draw [in=180, out=0] (5) to (10);
		\draw (14) to (6);
		\draw (9) to (13);
		\draw (10) to (14);
		\draw (1.center) to (19);
		\draw (21) to (3.center);
		\draw (13) to (6);
		\draw (0.center) to (9);
		\draw (19) to (7);
		\draw (7) to (13);
		\draw (14) to (8);
		\draw (8) to (21);
		\draw (10) to (2.center);
	\end{pgfonlayer}
\end{tikzpicture}

%% file: tikz_figs/circ_extracted.tikz
\begin{tikzpicture}
	\begin{pgfonlayer}{nodelayer}
		\node [style=none] (1) at (0, 0) {};
		\node [style=box] (4) at (3, 2) {$RX(\tfrac{\pi}{2})$};
		\node [style=box] (5) at (8, 2) {$RX(-\tfrac{\pi}{2})$};
		\node [style=ctrl vertex] (9) at (5.5, 2) {};
		\node [style=target vertex] (13) at (5.5, 0) {};
		\node [style=box] (19) at (2, 0) {$Z^{a_d + 1}$};
		\node [style=H box] (20) at (4, 0) {};
		\node [style=keto box] (21) at (0, 2) {};
		\node [style=none] (24) at (33.5, 2) {};
		\node [style=none] (25) at (33.5, 0) {};
		\node [style=box] (26) at (21.5, 2) {$RY((-1)^{a_d}\alpha(b))$};
		\node [style=box] (27) at (21.5, 0) {$RZ((-1)^{a_d+1}\alpha(a))$};
		\node [style=box] (28) at (15, 0) {$RZ(\tfrac{\pi}{2})$};
		\node [style=box] (29) at (28, 0) {$RZ(-\tfrac{\pi}{2})$};
		\node [style=ctrl vertex] (30) at (17.5, 2) {};
		\node [style=ctrl vertex] (31) at (25.5, 2) {};
		\node [style=target vertex] (32) at (17.5, 0) {};
		\node [style=target vertex] (33) at (25.5, 0) {};
		\node [style=box] (34) at (12, 0) {$RX(\tfrac{\pi}{2})$};
		\node [style=box] (35) at (31.5, 0) {$RX(-\tfrac{\pi}{2})$};
		\node [style=box] (36) at (8.5, 0) {$RX(-\alpha(i))$};
		\node [style=box] (37) at (29.5, 2) {$RX(-\alpha(c))$};
	\end{pgfonlayer}
	\begin{pgfonlayer}{edgelayer}
		\draw [in=180, out=0] (4) to (9);
		\draw (9) to (5);
		\draw (9) to (13);
		\draw (1.center) to (19);
		\draw (19) to (20);
		\draw (20) to (13);
		\draw (21) to (4);
		\draw (30) to (26);
		\draw (26) to (31);
		\draw [in=360, out=180] (33) to (27);
		\draw (30) to (32);
		\draw (31) to (33);
		\draw (35) to (25.center);
		\draw (32) to (27);
		\draw (34) to (28);
		\draw (28) to (32);
		\draw (33) to (29);
		\draw (29) to (35);
		\draw (5) to (30);
		\draw (13) to (36);
		\draw (36) to (34);
		\draw (31) to (37);
		\draw (37) to (24.center);
	\end{pgfonlayer}
\end{tikzpicture}

%% file: tikz_figs/mbqc_rewrite_elim.tikz
\begin{tikzpicture}
	\begin{pgfonlayer}{nodelayer}
		\node [style=InputMarker] (2) at (0, -1) {};
		\node [style=Measure, label={[label distance=-5pt]135:$i$}] (3) at (0, -1) {};
		\node [style=Measure, label={[label distance=-5pt]135:$d$}] (4) at (4, -1) {};
		\node [style=Output, label={[label distance=-5pt]135:$o_2$}] (6) at (6, -1) {};
		\node [style=Measure, label={[label distance=-5pt]135:$c$}] (7) at (4, 1) {};
		\node [style=Output, label={[label distance=-5pt]135:$o_1$}] (8) at (6, 1) {};
		\node [style=Measure, label={[label distance=-5pt]135:$b$}] (9) at (2, -1) {};
	\end{pgfonlayer}
	\begin{pgfonlayer}{edgelayer}
		\draw (4) to (6);
		\draw (7) to (4);
		\draw (7) to (8);
		\draw (3) to (9);
		\draw (9) to (4);
	\end{pgfonlayer}
\end{tikzpicture}

%% file: tikz_figs/mbqc_rewrite_lc.tikz
\begin{tikzpicture}
	\begin{pgfonlayer}{nodelayer}
		\node [style=InputMarker] (2) at (0, 0) {};
		\node [style=Measure, label={[label distance=-5pt]135:$i$}] (3) at (0, 0) {};
		\node [style=Measure, label={[label distance=-5pt]135:$d$}] (4) at (4, 0) {};
		\node [style=Measure, label={[label distance=-5pt]135:$a$}] (5) at (2, 2) {};
		\node [style=Output, label={[label distance=-5pt]135:$o_2$}] (6) at (6, 0) {};
		\node [style=Measure, label={[label distance=-5pt]135:$c$}] (7) at (4, 2) {};
		\node [style=Output, label={[label distance=-5pt]135:$o_1$}] (8) at (6, 2) {};
		\node [style=Measure, label={[label distance=-5pt]135:$b$}] (9) at (2, 0) {};
	\end{pgfonlayer}
	\begin{pgfonlayer}{edgelayer}
		\draw (5) to (4);
		\draw (4) to (6);
		\draw (7) to (4);
		\draw (7) to (8);
		\draw (3) to (9);
		\draw (9) to (4);
		\draw (7) to (9);
		\draw (5) to (6);
		\draw (7) to (6);
		\draw [bend right] (9) to (6);
	\end{pgfonlayer}
\end{tikzpicture}